\documentclass[11pt]{article}
\usepackage{fullpage}
\usepackage{enumitem}

\IfFileExists{orcidlink.sty}{\usepackage{orcidlink}}{\newcommand{\orcidlink}[1]{}}
\usepackage{mathpazo}

\usepackage{dsfont}
\usepackage{comment}
\usepackage{graphicx}
\usepackage{url}
\usepackage{latexsym}
\usepackage{amscd,amsmath}
\usepackage{amsfonts}
\usepackage{color}
\usepackage{float}
\usepackage{longtable}
\usepackage{pgf}
\usepackage{tikz}
\usetikzlibrary{arrows,automata,positioning}
\usepackage{comment}
\usepackage{hyperref}
\hypersetup{pdflang=en,
    urlcolor=blue,colorlinks=true,
    citecolor=black!70!green,
    linkcolor=black!70!red,
}
\usetikzlibrary{patterns}

\usepackage{amsthm}
\usepackage{mathtools}

\usepackage{cryptocode}
\usepackage{tcolorbox}
\tcbuselibrary{skins}

\def\secpar{\lambda}

\DeclarePairedDelimiter{\ceil}{\lceil}{\rceil}

\newcommand{\from}{\leftarrow}

\newcommand{\poly}{\mathrm{poly}}
\newcommand{\polylog}{\mathrm{polylog}}
\newcommand{\tO}{\tilde{O}}

\newcommand{\bits}{\{0,1\}}
\newcommand{\getsr}{\overset{{\scriptscriptstyle\$}}{\leftarrow}}

\newcommand{\A}{\mathcal{A}}
\newcommand{\B}{\mathcal{B}}

\newcommand{\calX}{\mathcal{X}}
\newcommand{\calY}{\mathcal{Y}}

\newcommand{\negl}{\mathsf{negl}}
\newcommand{\minH}{H_{\infty}}
\newcommand{\Ex}{\mathbb{E}}

\newcommand{\tx}{\tau}

\newcommand{\txoff}{\tx_{\mathrm{off}}}
\newcommand{\txon}{\tx_{\mathrm{on}}}
\newcommand{\btxoff}{\btx_{\mathrm{off}}}
\newcommand{\btxon}{\btx_{\mathrm{on}}}

\newcommand{\RO}{\mathsf{RO}}
\newcommand{\Or}{\mathsf{O}}

\newcommand{\Game}{\mathbf{Game}}

\newcommand{\btx}{\boldsymbol{\tx}}

\newcommand{\bqrys}{\mathbf{\qrys}}
\newcommand{\bdb}{\mathbf{\db}}

\newcommand{\bS}{\mathbf{S}}
\newcommand{\bomega}{\boldsymbol{\omega}}

\newcommand{\db}{D}
\newcommand{\edb}{E}
\newcommand{\bedb}{\mathbf{E}}
\newcommand{\st}{\mathsf{st}}
\newcommand{\bst}{\mathbf{st}}
\newcommand{\qry}{q}
\newcommand{\qrys}{Q}
\newcommand{\pir}{\Pi}
\newcommand{\spir}{\Pi}
\newcommand{\setup}{\mathsf{Setup}}
\newcommand{\init}{\mathsf{Init}}
\newcommand{\query}{\mathsf{Query}}
\newcommand{\dpir}{\widetilde{\Pi}}
\newcommand{\advice}{\mathsf{Advice}}
\newcommand{\hint}{\mathsf{Hint}}
\newcommand{\recon}{\mathsf{Recon}}
\newcommand{\ans}{\mathsf{ans}}
\newcommand{\bans}{\mathbf{ans}}

\newcommand{\adviceout}{A}
\newcommand{\hintout}{H}
\newcommand{\reconout}{R}
\newcommand{\badviceout}{\mathbf{\adviceout}}
\newcommand{\bhintout}{\mathbf{\hintout}}
\newcommand{\breconout}{\mathbf{\reconout}}

\newcommand{\bR}{\mathbf{R}}

\newcommand{\client}{\mathsf{Client}}
\newcommand{\server}{\mathsf{Server}}
\newcommand{\prf}{\mathsf{F}}
\newcommand{\key}{\mathsf{k}}

\newcommand{\1}{\mathds{1}}

\newcommand{\heading}[1]{\smallskip\noindent{\bf {#1}}}

\newcommand{\apref}[1]{Appendix~\ref{#1}}

\newcommand{\figref}[1]{Figure~\ref{#1}}

\newcommand{\lemref}[1]{Lemma~\ref{#1}}
\newcommand{\secref}[1]{Section~\ref{#1}}
\newcommand{\thref}[1]{Theorem~\ref{#1}}
\newcommand{\defref}[1]{Definition~\ref{#1}}

\newtheorem{thm}{Theorem}[section]
\newtheorem{obs}[thm]{Observation}
\newtheorem{assm}[thm]{Assumption}
\newtheorem{defn}[thm]{Definition}
\newtheorem{lem}[thm]{Lemma}
\newtheorem{cor}[thm]{Corollary}
\newtheorem{clm}[thm]{Claim}

\newtheorem{rem}[thm]{Remark}

\newenvironment{remark}{\begin{rem}}{\end{rem}}
\newenvironment{theorem}{\begin{thm}}{\end{thm}}
\newenvironment{lemma}{\begin{lem}}{\end{lem}}
\newenvironment{corollary}{\begin{cor}}{\end{cor}}
\newenvironment{definition}{\begin{defn}}{\end{defn}}

\title{Lower Bounds for PIR with Preprocessing\\ from Blackbox Cryptography}

\author{
    Alexander Hoover \orcidlink{0009-0003-9818-1419}\footnote{
    Stevens Institute of Technology, Email: \texttt{\href{ahoover@stevens.edu}{ahoover@stevens.edu}}}
    \and
    Giuseppe Persiano \orcidlink{0000-0001-6579-4807}\footnote{Università di Salerno, Email: \texttt{\href{giuper@gmail.com}{giuper@gmail.com}}}
    \and
    Kevin Yeo \orcidlink{0009-0009-4997-6307}\footnote{Google, Email: \texttt{\href{kwlyeo@google.com}{kwlyeo@google.com}}}
}

\date{} %
\begin{document}

{
\hypersetup{urlcolor=black} %
\maketitle
}
\thispagestyle{empty}
 \begin{abstract}
     Single-server private information retrieval (PIR) schemes are known to require
linear query time. Recent works
circumvent these classical lower bounds by leveraging preprocessing to answer
queries in sublinear time. We prove computation lower bounds for PIR with preprocessing schemes making
blackbox usage of any cryptography (such as random oracles or virtual blackbox
obfuscation). If the client stores $s$ bits about
an $n$-bit database, then answering $k = \Omega(s)$ queries requires either
$\Omega(n/s)$ amortized online communication or $\Omega(n/s)$ amortized
server cryptographic operations. This is tight, as known constructions match
either bound while outperforming the other. Our lower bound is unconditional
and allows arbitrary query
protocols, weakened privacy, and server-side database encodings (including
doubly efficient PIR) whenever the encoding is independent of the blackbox
cryptography. Previous bounds were only known conditionally
and for restricted classes of preprocessing, e.g., under
non-encoding assumptions.
Our framework also yields $\Omega(n/s)$ communication lower bounds for
schemes with $o(n/s)$ server cryptographic operations,
communication-determined server cryptography, or perfect privacy in the
idealized model.
Finally, we prove lower bounds for symmetric PIR with
client preprocessing in the random oracle model and give a matching
construction using only one-way functions in the online phase.

     \vfill
 \end{abstract}
 \clearpage

\pagebreak
{
\hypersetup{linkcolor=black} %
\tableofcontents
}
\pagebreak

\pagenumbering{arabic}
\section{Introduction}\label{sec:intro}

Private Information Retrieval (PIR) is a very powerful
cryptographic primitive that enables a client to retrieve
the $i$-th entry from a $n$-bit database $D \in \{0,1\}^n$ held by a server
with the privacy guarantee that
the server cannot learn the index $i$ queried by the client.
Leveraging its strong privacy guarantees, PIR is a core building block for a wide range of privacy enhancing technologies including
advertising~\cite{CCS:GreLadMie16}, blocklist checks~\cite{USENIX:KogCor21}, contact discovery~\cite{PoPETS:DRRT18,USENIX:KRSSW19}, distributed file systems~\cite{ipfs}, leaked password check~\cite{USENIX:TPYRKI19,CCS:LPASCR19,USENIX:ALPRSS21}, private communication~\cite{SP:ACLS18} and web search~\cite{tiptoe}.
In recent years, PIR has also seen real-world deployments in industry including caller ID~\cite{apple}, device enrollment~\cite{google} and password leak detection~\cite{microsoft}.
Given the importance of PIR, there has been substantial work in studying its complexity and efficiency.

PIR has been studied in two main settings:
the multiple-server setting where the database is held by two or more servers
that are non-colluding,
introduced by Chor, Kushilevitz, Goldreich and Sudan~\cite{FOCS:CGKS95},
and the single-server setting put forth by
Kushilevitz and Ostrovsky~\cite{FOCS:KusOst97}.
It is known that any single-server PIR with 
sublinear communication implies
oblivious transfer~\cite{EC:DiCMalOst00} that cannot be built from blackbox OWFs~\cite{STOC:ImpRud89} and the server must perform a linear
number of public-key cryptography~\cite{EC:DujHaj24}.
In contrast, there exists two-server PIR with 
$\polylog(\lambda, n)$ communication using only
one-way functions (OWFs)~\cite{STOC:ChoGil97,CCS:BoyGilIsh16} as well as $n^{o(1)}$ communication without any cryptographic assumptions at all~\cite{STOC:DviGop15}. 
Even though multi-server PIR constructions are known to be more efficient, their privacy guarantees are more fragile as they
critically depend on stronger non-collusion assumptions between different parties that are challenging to obtain in practice.
For this reason, the study of single-server PIR is important and it will be the focus of our work.

\heading{PIR with Client Preprocessing.} As a way to improve the efficiency of PIR, there has been a recent line of work that studies PIR with client preprocessing where the 
database
is preprocessed by the client in an offline phase and the information gathered
is used to improve the efficiency of online queries 
(see~\cite{CCS:PatPerYeo18,EC:CorKog20,USENIX:KogCor21,C:SACM21,EC:CorHenKog22,USENIX:MZRA22,EC:Yeo23,C:LazPap23,C:Yeo23,TCC:LazPap23,EC:GhoZhoShi24,USENIX:LazPap24} and references therein).
These constructions obtained sublinear query time circumventing known linear query time lower bounds for PIR without preprocessing~\cite{C:BeiIshMal00,EC:DujHaj24}.
Surprisingly, it turns out the preprocessing model also enables building single-server PIR from weaker assumptions.
Recent works, starting from Piano~\cite{SP:ZPZS24} and follow-ups~\cite{CCS:RenMugSun24,EC:HPPY25,EC:WanRen25}, have shown it is possible to build
single-server PIR with sublinear query communication from blackbox OWFs in the client preprocessing model.
In more detail, these works show that, if a client stores $s$ bits about an $n$-bit database, then PIR schemes require only $\tilde{O}(n/s)$ (amortized) communication and computation per query.
In a standard parameter setting with client storage $s = O(\sqrt{n})$, the query computation and communication becomes $\tilde{O}(\sqrt{n})$.
Furthermore, there exist schemes~\cite{SP:ZPZS24,CCS:RenMugSun24,EC:HPPY25,EC:WanRen25} where the server performs no cryptography during queries
(contrasting the linear public key operations lower bound in~\cite{EC:DujHaj24}).

In terms of lower bounds for PIR with client preprocessing model, 
the known limitations for constructions apply only to special cases. 
The results in~\cite{EC:CorKog20,EC:CorHenKog22,EC:Yeo23} 
prove computation lower bounds,
but require very strong non-encoding assumptions for the server that must store the $n$-bit database without any modification for responding to queries.
Recent work~\cite{C:IshShiWic24} showed barriers for communication-efficient schemes
restricted to database-oblivious and single-roundtrip query constructions, where
all client requests may depend only on the queried index and the randomness (and not on the database).
In particular, such a construction would imply a separation between
the complexity classes SZK and BPP, which remains an open complexity theory problem for many natural settings including the random oracle model.
Therefore, it remains open to unconditionally resolve the best possible efficiency achievable by a single-server PIR with client preprocessing.

\heading{PIR with Server Preprocessing.}
We note that there also exists a line of work studying PIR with server preprocessing where all the additional information is stored publicly by the server~\cite{C:BeiIshMal00,TCC:BIPW17,TCC:CanHolRic17,EC:HOWW19,SODA:PerYeo22} that is also known as doubly efficient PIR (DEPIR). A recent breakthrough work presented a single-server DEPIR using $\tilde{O}(1)$ query communication and time assuming the hardness of Ring LWE~\cite{STOC:LinMooWic23}.
In this model, it is already known that
public-key cryptography is necessary to obtain sublinear query communication (as observed in~\cite{EC:HPPY25}).
Furthermore, it was proven that restricted classes
of DEPIR with single-roundtrip queries and passive servers acting as memory (without any computation) are impossible from blackbox cryptography~\cite{EC:LinMooWic25}. However, it remains unknown whether this impossibility applies to all possible
DEPIR schemes from blackbox cryptography without any other restrictions.

\heading{Symmetric PIR with Client Preprocessing.} We also extend the problem to try and understand what additional properties could be achieved
by single-server PIR with preprocessing when depending only on one-way functions.
For example, we could consider the stronger notion of symmetric PIR~\cite{STOC:GIKM98} where we also
wish to guarantee that the client may retrieve at most one database entry per query from the server.
By definition, it is known that symmetric PIR (without preprocessing) implies oblivious transfer even if the symmetric PIR is inefficient and uses linear (or even more) communication~\cite{EC:DiCMalOst00}. One can wonder whether it is possible to circumvent this blackbox impossibility in the
client preprocessing model. Given the above, it is clear that it is insufficient to depend only on one-way functions in the entire construction. However, it remains plausible to build a SPIR with client preprocessing that relies only upon one-way functions during query time.

\subsection{Our Contributions}

As our main result, we prove lower bounds for single-server PIR with preprocessing making blackbox usage of cryptography. 
Throughout our work, we will use PIR with preprocessing to encompass both client preprocessing and server preprocessing (unless otherwise specified).
To model blackbox cryptography usage, we suppose schemes have access to a 
crypto oracle~\cite{EC:DujHaj24,EC:LinMooWic25} that may be used to implement a wide-range of ideal powerful cryptographic primitives,
including random oracles, generic multilinear groups and virtual blackbox obfuscation (see Section~\ref{sec:prelim:oracles} for more details).
All our lower bounds will focus on the efficiency during online query time
and are agnostic to the cost of the offline phase (used to compute the client storage and/or server encodings). In particular, our results apply to all constructions whose offline phase proceeds in two steps. First, the server is able to perform a one-time encoding of the database. Afterwards, the client and server perform joint computation to enable
the client to obtain some private state. We note this encompasses a wide range of PIR with preprocessing
models including PIR with client preprocessing and public-key DEPIR.
Our lower bounds apply to all constructions where the first server database encoding step
is performed without any oracle queries (or no database encoding is performed at all).
To our knowledge, this aligns with all blackbox constructions known to date (as well as state-of-the-art non-blackbox constructions too). All other steps of the PIR protocol are able to have unrestricted oracle
access. See \defref{def:oracle-free} for more details on oracle-free database encodings.

\heading{Computation Lower Bounds.}
We present computation lower bounds showing that any single-server PIR with client preprocessing
that uses any blackbox cryptography (in the crypto oracle model)
cannot be more efficient than currently known constructions.

\begin{theorem}[Informal]
There does not exist a single-server PIR with client preprocessing scheme making blackbox usage of cryptography
for an $n$-bit database with $s$-bit client storage supporting a single batch query of $k = \Theta(s)$ indices with
amortized online query computation of $t = o(n/s)$. In particular, we show each such single-server PIR scheme must satisfy at least
one of the following:
\begin{itemize}[noitemsep]
\item The amortized online query communication must be $\Omega(n/s)$ bits.
\item The amortized server cryptographic bit operations must be $\Omega(n/s)$.
\end{itemize}
\end{theorem}

Our result is tight in multiple ways. First, there exist schemes
that achieve exactly $t = O(n/s)$ online query time~\cite{EC:WanRen25} for
an unbounded number of queries using OWF as a blackbox. Secondly,
there exist schemes that bypass exactly one of the above requirements on
the batch size, computation, and communication.
For small batch sizes, there exist single-server PIR schemes with $o(n/s)$
online communication for a bounded number of less than $s$
queries (such as~\cite{EC:ZLTS23,TCC:LazPap23,EC:GhoZhoShi24}).
To obtain small communication, these schemes
may be combined with blackbox fully homomorphic encryption (FHE) techniques
to support unbounded queries by privately refreshing
the $s$-bit client storage using $o(n)$ communication~\cite{CCS:FLLP24}. 
However, all these FHE approaches require amortized $\Omega(n/s)$ FHE operations
to refresh the $s$-bit client storage (even when ignoring the communication
from the original offline phase).
For server computation, there exists
schemes~\cite{SP:ZPZS24,EC:HPPY25} supporting unbounded queries with
$\tilde{O}(n/s)$ amortized online communication where the server performs no 
cryptographic operations. In other words, our lower bound identifies two 
fundamental bottlenecks such that any scheme will inevitably suffer from at 
least one of them.

Furthermore, our results are more widely applicable than prior lower bounds.
In particular, our lower bounds hold for any
single-server PIR with client preprocessing scheme that makes blackbox usage of cryptography, requiring only that the server's database encoding is computed without crypto oracle queries (the client's private state may still be computed using the crypto oracle).
In contrast, the server time lower bounds in~\cite{EC:CorKog20,EC:CorHenKog22,EC:Yeo23} only apply to schemes with server non-encoding assumptions such that the server must use the database without modification for queries.
Note that our restriction is considerably weaker than non-encoding: the
server may still encode the database arbitrarily, so long as the encoding is
computed without querying the crypto oracle.
Additionally, the communication (and, thus, computation) lower bounds in~\cite{C:IshShiWic24} only apply to database-oblivious, single-roundtrip query constructions where
all client requests depend only on the queried index and internal randomness (along with the lower bound requiring additional complexity assumptions).
Finally, we note that the lower bounds in~\cite{EC:DujHaj24} may be interpreted into the preprocessing model showing
that the total time across the offline and online phases must be linear. In contrast, our lower bounds
show that the amortized online query time must be $\Omega(n/s)$ regardless of the cost of the offline phase (including even if
the client's offline phase performs a linear number of public-key operations).

Our techniques also give lower bounds for doubly efficient PIR (DEPIR)
with server preprocessing.
In particular, we rule out the existence of DEPIR built from blackbox cryptography. 

\begin{theorem}[Informal]
There does not exist a single-server doubly efficient PIR making blackbox usage of cryptography
for an $n$-bit database with query computation of $t = o(n)$.
\end{theorem}

As a caveat, our lower bound only applies to public-key DEPIR
whereas the prior impossibility~\cite{EC:LinMooWic25} also applies to secret-key DEPIR
where the server's encoding may depend on a client private state as studied in~\cite{TCC:BIPW17,TCC:CanHolRic17,eprint:CIMR25}.
Our lower bound places no restriction on the query protocol, which may
consist of arbitrarily many rounds with an actively computing server. In
contrast, the prior impossibility in~\cite{EC:LinMooWic25} only applied to
DEPIR with single-roundtrip query algorithms with a passive server
restricted to acting only as memory and performing no computation, while
allowing the preprocessing to use the oracle (unlike our result). Therefore, the two results are
incomparable. Nevertheless, our oracle-free database encoding restriction
aligns with the current state-of-the-art public-key DEPIR~\cite{STOC:LinMooWic23}, even
though it is a non-blackbox construction.
See
Section~\ref{sec:depir} for further discussion. 

Finally, we note that our lower bounds also apply to PIR schemes that provide privacy guarantees
against weaker adversaries that are restricted to a fixed polynomial number of crypto oracle queries
(similar to weak key exchange from Merkle puzzles).

\begin{theorem}[Informal]
There does not exist a weak single-server PIR with preprocessing scheme 
for an $n$-bit database with $s$-bit client storage
supporting a single batch query of $k = \Theta(s)$ indices with amortized online query computation of $t = o(n/s)$ that makes $q$ crypto oracle
queries and provides privacy against adversaries that make at most $O(q)$ crypto oracle queries.
\end{theorem}

The above lower bound may also be applied to single-server PIR (without preprocessing) implying there does not exist
any such weak schemes that obtain sublinear online query time against weak adversaries that make the same number of
crypto oracle queries as an honest execution of the PIR protocol. To our knowledge, this was not previously known. The
prior blackbox impossibility shows sublinear communication PIR implies oblivious transfer~\cite{EC:DiCMalOst00,EC:DujHaj24}, but there are known
constructions of weak oblivious transfer that are secure against such restricted adversaries~\cite{TCC:BihGorIsh08}.

\heading{Communication Lower Bounds.} We also prove communication lower bounds for certain classes of single-server PIR with preprocessing schemes that match current constructions.

\begin{theorem}[Informal]
There does not exist a single-server PIR with preprocessing scheme making blackbox usage of cryptography for an $n$-bit database with $s$-bit client storage
supporting a batch query of $k = \Theta(s)$ indices with amortized online communication of $o(n/s)$ if the construction satisfies at least one of the following requirements:
\begin{itemize}[noitemsep]
\item The amortized server cryptographic bit operations is $o(n/s)$.
\item All server cryptographic operations depend only on the communication transcript.
\item Perfect privacy (in an idealized model).
\end{itemize}
For the first two settings, we require the server's database
encoding to be computed without crypto oracle queries. Our perfect privacy
result, however, requires no such server's database encoding restriction.
\end{theorem}

The first setting is a direct corollary of our computation lower bound, which is tight as there exist schemes
with $\tilde{O}(n/s)$ communication where the server performs no cryptography~\cite{SP:ZPZS24,EC:HPPY25}.
For the second setting, we consider PIR schemes where the server's cryptographic operations depend only on the communication
transcript.
To our knowledge, all OWF-based single-server PIR schemes satisfy this criterion and, furthermore, there exist OWF-based schemes
with $O(n/s)$ communication exactly matching our lower bound~\cite{EC:WanRen25}.
Additionally, there do exist schemes with cryptographic operations that depend on the database (not appearing in the communication) that achieve $o(n/s)$ communication
using FHE~\cite{CCS:FLLP24}.
Finally, we prove an $\Omega(n/s)$ communication lower bound
for perfect privacy matched
by known constructions~\cite{CCS:RenMugSun24,EC:HPPY25} when adapted to the random oracle
model.
This lower bound applies to schemes that derive their client query
randomness from the random oracle and key material in their state,
as done by the adapted constructions
of~\cite{CCS:RenMugSun24,EC:HPPY25}. Notably, this is our one lower bound
that requires no restriction on the server's database encoding.

Our lower bounds substantially improve over the only existing
communication lower bound in this model~\cite{C:IshShiWic24} that is only established under certain
complexity-theoretic assumptions, and therefore yields only a barrier for
blackbox usage of specific crypto oracles.
Moreover, this bound only
applies to certain restricted classes of PIR by making non-encoding and round
assumptions.
While our first two communication bounds assume an oracle-free database
encoding, they permit arbitrary (oracle-free) encodings otherwise, as well
as multi-round and database-dependent queries. See Section~\ref{sec:related} for a
detailed comparison.

\heading{Symmetric PIR with Client Preprocessing.}
We show that it is possible to obtain even stronger properties while still only relying upon OWFs during online queries.
In particular, we present a construction of symmetric single-server PIR with client preprocessing where it is guaranteed that
the client only obtains the queried entry and no other information about the database.

\begin{theorem}[Informal]
There exists a symmetric single-server PIR with client preprocessing %
where the query algorithm only relies upon OWFs
using $s$-bit client storage with query communication $\tilde{O}(n/s)$ and query time $\tilde{O}(n/s)$ supporting $\tilde{O}(s)$ queries
per offline phase.
\end{theorem}

We note that the query communication costs of our construction match many prior works that only depend on OWFs during queries~\cite{SP:ZPZS24,CCS:RenMugSun24,EC:HPPY25,EC:WanRen25} without the additional symmetric privacy guarantees and
similarly require no cryptographic server operations. Furthermore, the computational times match known lower bounds in the preprocessing model~\cite{EC:CorKog20,EC:CorHenKog22,EC:Yeo23}.
One may notice that our construction is weaker than previous schemes in that it only supports $\tilde{O}(s)$ queries per offline phase
whereas prior non-symmetric PIR schemes can handle an unbounded polynomial number of queries. This turns out to be a limitation of symmetric PIR with preprocessing as we show:

\begin{theorem}[Informal]
In the random oracle model, there does not exist a symmetric single-server PIR with client preprocessing %
where the query algorithm only relies upon OWFs
using $s$-bit client storage that supports $k > s$ queries per offline phase.
\end{theorem}

Our symmetric PIR construction is tight in the number of supported queries up to $\tilde{O}(1)$ factors.
Note, the offline phase of our construction relies upon public-key cryptography.
However, this is necessary since symmetric PIR implies oblivious transfer that cannot be built only from blackbox OWFs~\cite{STOC:ImpRud89}.
If we insist on the online queries to be efficient and avoid public-key cryptography, then the offline phase must rely upon public-key cryptography (as done in our construction).

\subsection{Our Techniques}

We first quickly overview three types of lower bound techniques used in prior
work to highlight where our techniques and results differ.
The line of lower bounds in~\cite{EC:CorKog20,EC:CorHenKog22,SODA:PerYeo22,EC:Yeo23} prove time lower bounds via compression arguments
where a random database may be compressed using only a subset of database entries that are accessed by the server during online
queries. However, this approach inherently requires a strong server non-encoding assumption for the compression method to identify
the exact database entries accessed during the query algorithm.

The communication lower bound in~\cite{C:IshShiWic24} takes a different approach by
constructing PIR privacy adversaries assuming oracle access to a SZK-complete problem.
However, this inherently means the lower bound is conditional on an additional complexity assumption
along with imposing additional restrictions such as database-obliviousness and single-round queries like prior works~\cite{TCC:LiuVai16}.

Finally, recent works~\cite{EC:DujHaj24,EC:LinMooWic25} consider schemes using crypto oracles encapsulating blackbox cryptography usage. These lower bounds show that one can compile out all server crypto oracle usages in single-server PIR schemes to obtain a new construction that effectively
only uses a random oracle. Afterwards, they rely on the fact that such a random oracle construction is impossible.
While our lower bound techniques will borrow the crypto oracle modeling, the overall framework is not compatible
with our goals as there do exist OWF-based single-server PIR with preprocessing schemes.
Instead, we will develop new techniques to directly prove efficiency lower bounds.

\subsubsection{Computation Lower Bounds}

We start by presenting our computation lower bounds. We will consider single-server PIR with preprocessing schemes
with access to a crypto oracle $\Or$: $\Pi^{\Or} = (\init, \setup^{\Or}, \query^{\Or})$.
The $\init$ algorithm enables the server to preprocess the $n$-bit database in polynomial time to produce
an encoded database $\edb$ for the server.
In our model, $\init$ is a plain randomized algorithm that makes no
queries to the crypto oracle $\Or$ (an \emph{oracle-free encoding}), whereas
$\setup$ and $\query$ have full oracle access; we explain below why our
proof requires the encoding to be independent of $\Or$.
The $\setup$ algorithm receives the encoded database $\edb$ and outputs a $s$-bit storage $\st$
for the client. As a note, this could be a two-party protocol between the client and server, but
we assume the server observes nothing from this offline phase (as we are proving a lower bound, this only strengthens the result).
The $\query$ algorithm is an interactive two-party protocol between
the client and the server. The client receives the $k = \Theta(s)$ indices of queries $\qrys$ along with the $s$-bit client storage $\st$
while the server receives the encoded database $\edb$. We denote the execution of the query algorithm as
$(\ans; \bot \mid \tx) \leftarrow \query^{\Or}(\st, \qrys; \edb)$ where the client receives the PIR's output as $\ans$, the server
receives no output and both parties observe the communication transcript $\tx$.

Our goal is to prove a lower bound of $t = \Omega(n/s)$ amortized online query computation across $k = \Theta(s)$ queries.
Equivalently, we want to show that the total computation during the execution of the query algorithm is $\Omega(n)$. In fact,
we will prove that either the total communication is $\Omega(n)$ bits or the server queries the crypto oracle over $\Omega(n)$ bits
during the execution of the query algorithm.
For the sake of contradiction, we will suppose the existence of a scheme $\Pi^{\Or}$ that is more efficient with both $o(n)$ bit communication
and the server queries the crypto oracles on $o(n)$ bits during the query algorithm.
We note that we make no assumptions about the efficiency of the setup or server initialization algorithms. Thus, our results will
lower bound the online computational cost.
The one structural assumption we make is that $\init$ performs no
crypto oracle queries. We note our results extend to encodings that query
$\Or$ on a sublinear total number of bits (by recording $\init$'s queries in
the dual PIR's advice, as in our sublinear-server-cryptography setting) or
whose oracle queries are determined by the communication transcript, though
with additional caveats; see Remark~\ref{rem:init-ext}.

\heading{Dual PIR.}
A central step in our proof is a reduction from PIR with preprocessing to a new notion called \emph{dual PIR} that essentially inverts
the order of the preprocessing and efficiency requirements of PIR. A dual PIR consists of three algorithms:
\begin{itemize}
  \item $\adviceout \leftarrow \advice(\db)$: before any query indices are known, the client
    computes a $t$-bit \emph{advice} $\adviceout$ about the $n$-bit database $\db$. We require the
    advice to be sublinear, $t = |\adviceout| = o(n)$, so that it cannot encode the entire database.
  \item $\hintout \leftarrow \hint(\qrys, \db, \adviceout)$: given a set of $k$ query indices
    $\qrys$, the database $\db$, and the advice $\adviceout$, the client computes a $s$-bit
    \emph{hint} $\hintout$.
  \item $\reconout \leftarrow \recon(\qrys, \adviceout, \hintout)$: using only the query set
    $\qrys$, the advice $\adviceout$, and the hint $\hintout$ --- but, critically, \emph{without}
    the database $\db$ --- the client reconstructs the relevant entries $\db[\qrys]$.
\end{itemize}
Reconstruction would be trivial if the hint $\hintout$ had room to write down $\db[\qrys]$, i.e.,
when $s \ge k$. We are therefore interested in the limits of when a dual PIR can be correct once
the hint is smaller than the number of queries, $s < k$.

The above is suggestively called dual PIR, because it intuitively inverts the efficiency requirements of
standard PIR with preprocessing.
The preprocessing step of dual PIR (without the query) uses $t = o(n)$ bits similar
to the $\query$ algorithm.
The online portion of dual PIR with the query produces a $s$-bit encoding
similar to the $\setup$ algorithm.
Furthermore, dual PIR has no privacy guarantees and is not even a two-party protocol.
Instead, one can view it as a compression problem.

Using dual PIR, our proof will proceed in two steps. First, we show that when $t = o(n)$ and $s < k$, there cannot exist
a dual PIR construction that is correct even with probability $2^{s-\Omega(k)}$ using a connection with leakage-resilient cryptography. Afterwards, we show that our too-efficient-to-be-true PIR with preprocessing $\Pi^{\Or}$ may be used to construct such an impossible dual PIR.

\heading{Connection with Leakage-Resilient Cryptography.}
To prove that there cannot exist dual PIR schemes that are correct (even with exponentially small probability), we make a connection with the well-studied area of leakage-resilient cryptography~\cite{CHES:AARR02,AC:ShiKobIma03,FOCS:DziPie08,C:AlwDodWic09,ICITS:AlwDodWic09,EC:ADNSWW10}. At a high level, this model considers the setting
where an adversary successfully compromises a remote system, but it is infeasible for the adversary to download all the data
in the system. To build cryptography resilient to leakage in this attack, one can consider a very large private key (say, 1 TB) of $n$ bits
where it is assumed the adversary obtains at most $t < n$ bits.
For efficiency purposes, any cryptographic primitive built in this system will only use a subset of $k$-bits of the large $n$-bit key.
Therefore, a core problem in leakage-resilient cryptography is {\em subkey prediction} that upper bounds
the probability
that a computationally unbounded adversary may correctly guess $k$ random locations 
of a $n$-bit random key when the adversary may compute any $t$-bit leakage about the key~\cite{C:BelKanRog16}.

Interestingly, we show that a very good dual PIR (in terms of efficiency and correctness) implies a successful adversary
for the subkey prediction problem in leakage-resilient cryptography. We will view the $n$-bit database in dual PIR
as the $n$-bit private key in subkey prediction. Then, we will view the $t$-bit advice output of dual PIR
as the $t$-bit leakage chosen by the subkey prediction adversary.
Finally, when the subkey prediction adversary is given the $k$ random locations $\qrys$ of the $n$-bit key to guess, we will execute
the final $\recon$ step of dual PIR where we simply choose a $s$-bit random string
to replace the output of the dual PIR's second $\hint$ step. If the underlying dual PIR was correct with probability $\varepsilon$, then the subkey prediction adversary is correct
with probability at least $\varepsilon \cdot 2^{-s}$ corresponding to when the correct $s$-bit encoding output of $\hint$ is guessed correctly.

We can now combine our subkey prediction adversary built from dual PIR with prior bounds proven on the best possible
success probability for any subkey prediction adversary in~\cite{C:BelKanRog16}.
In particular, it was shown that for $t \le n/2$ leakage and $k = \Omega(\log n)$ locations of interest,
the adversary can predict the $k$ bits with probability at most $2^{-\Omega(k)}$.
This immediately implies that there does not exist any dual PIR that is correct with probability larger than
$\varepsilon > 2^{s - \Omega(k)}$ unless the output of the $\advice$ step in the dual PIR is $t \ge n/2$ bits.
We note that the subkey prediction bound in~\cite{C:BelKanRog16} is an information-theoretic result that
holds for any computationally unbounded adversary 
with arbitrary shared randomness across the adversary's algorithms.
Thus, the impossibility holds for dual PIR algorithms that are computationally inefficient 
with access to oracles.

\heading{Dual PIR Construction Overview.}
Next, we will construct a dual PIR using a too-efficient-to-be-true PIR with preprocessing scheme $\Pi^{\Or} = (\init, \setup^{\Or}, \query^{\Or})$ with an oracle-free database encoding. Our dual PIR will have access both to a crypto oracle $\Or$ and a random oracle $\RO$.
Furthermore, our dual PIR will not necessarily be computationally efficient,
but the above impossibility on the subkey prediction problem 
will still apply since it is an information-theoretic result.
However, we will ensure our dual PIR makes at most a polynomial number of queries to the crypto oracle $\Or$ in total as we will use our dual PIR as a tool to build a privacy adversary for the underlying PIR scheme (critical to proving correctness of our dual PIR). In contrast, our dual PIR may make exponentially many queries to the random oracle $\RO$ (for example, to sample the exponentially many candidate crypto oracles $\Or_1,\ldots,\Or_z$); this is harmless, since the dual PIR need not be efficient and only its polynomially many crypto oracle queries are inherited by the privacy adversary.

At a high level, we wish to use the strong privacy properties of the underlying PIR $\Pi^{\Or}$ to build the dual PIR. In the first step, we can have
$\advice^{\Or,\RO}(\db)$ execute $\Pi^{\Or}$ properly for an arbitrary query 
$\qrys'$ to obtain a communication transcript $\tx$,
which will be the $t$-bit advice $\adviceout = \tx$.
Now, the second $\hint$ algorithm receives the real queries $\qrys$.
By privacy, then it should be that this communication transcript $\adviceout = \tx$ must also be consistent with the given query $\qrys$ even if it is different than the original query $\qrys' \ne \qrys$ used in $\advice$.
Then, $\hint$ can iterate through all possible $s$-bit client storage that may be output by $\st \leftarrow \setup^{\Or}(\edb)$ of the underlying PIR  scheme that would produce the communication transcript $\tx$ when executing
$\query^{\Or}(\st, \qrys; \edb)$ where $\edb \leftarrow \init(\db)$ is the server's encoding of the database in the offline phase (recomputed by $\hint$ using encoding randomness derived from the shared random oracle, so that all dual PIR algorithms agree on the same $\edb$).
The output of the hint algorithm will be this $s$-bit client storage $\hintout = \st$.
Effectively, the second hint algorithm is re-randomizing the client of the underlying PIR while keeping the communication transcript
$\tx$ the same. Finally, $\recon^{\Or,\RO}(\qrys, \adviceout, \hintout)$ will execute the client-side portion of the underlying PIR's query algorithm with the $s$-bit client storage discovered by the hint algorithm $\hintout = \st$ using the communication transcript from advice
$\adviceout = \tx$ where the server is simulated using the communication transcript.

While the above is the general framework behind our dual PIR construction, we note that there are many obstacles
to concretely instantiate these ideas into a dual PIR scheme that is 
both provably correct and efficient in terms of crypto oracle queries (to apply PIR privacy later).

\heading{Sampling Client States.} We start with the core challenge of our dual PIR scheme that is to sample a $s$-bit client storage
in the $\hint$ algorithm that is consistent with the communication transcript $\tx$ output by $\advice$.
A trivial approach is to try all the possible internal random coin tosses
of the underlying PIR's $\setup$ algorithm to obtain a candidate $s$-bit client storage and execute the underlying PIR's $\query$ algorithm to check if the resulting communication transcript matches $\tx$
Note, this is possible since $\hint$ receives all the necessary information including the query $\qrys$ and database $\db$.
However, this would require an exponential number of queries to the crypto oracle $\Or$.

Instead, our $\hint$ algorithm will sample new crypto oracles $\Or'$ with the goal of finding a crypto oracle and $s$-bit client storage pair $(\Or', \st')$ such that the execution of the underlying PIR's $\query$ algorithm results in the desired communication transcript $\tx$.
Since the crypto oracle $\Or'$ is locally sampled, we no longer care about the number of queries performed. Thus, we can now use
the above approach of trying all internal random coin tosses to the PIR's $\setup$ algorithm and check if the resulting
$s$-bit client storage used in the execution of the PIR's $\query$ algorithm results in a communication transcript $\tx$.
If no such $s$-bit client storage is found, we can continue to sample new crypto oracles $\Or'$ until finding a successful tuple
of $s$-bit client storage and crypto oracle $\Or'$.

While the above approach avoids too many queries to the crypto oracle $\Or$, it creates several new problems.
First, $\hint$ must also now encode the newly sampled crypto oracle $\Or'$ to be used by $\recon$, which we will discuss in more detail later. An even more subtle (but critical) issue is that $\recon$ can no longer discern whether a crypto oracle $\Or'$ is compatible
with the database $\db$ and communication transcript $\tx$. In particular, $\recon$ cannot determine whether the
server's communication can match the transcript $\tx$ for the given crypto oracle $\Or'$ since $\recon$ does not receive the
database $\db$ as input. As an example, the underlying PIR scheme could include a $\sqrt{n}$ bits consisting of XOR of random bits
from the crypto oracle and the first $\sqrt{n}$ bits of the input database $\db$.
Now, a randomly sampled crypto oracle $\Or'$ is very unlikely to be compatible with the input database $\db$ and communication
transcript $\tx$ (with probability $2^{-\sqrt{n}}$). However, $\recon$ has no way to figure this out since it does not have the database $\db$.

Instead, we will modify our dual PIR construction such that the $\hint$ algorithm only randomly samples crypto oracles $\Or'$
such that the server's communication will always be consistent with communication transcript for the input database $\db$.
To do this, the $\advice$ algorithm's $t$-bit output $A$ will consist of both the communication transcript $\tau$ and of all the queries to the crypto oracle $\Or$
performed by the server during the execution of the $\query$ algorithm that we denote by the set $S$.
As a reminder, we assumed towards a contradiction that the server of the underlying PIR queries the crypto oracle on $o(n)$ bits, meaning that $|S| = o(n)$ bits.
As all algorithms of the dual PIR have access to the crypto oracle $\Or$, they can all compute the crypto oracle responses: $\Or(S) = \{(q, \Or(q)) \mid q \in S\}$. Now, the $\hint$ algorithm will only randomly sample crypto oracles $\Or'$ conditioned on $\Or'(S) = \Or(S)$.
For any such sampled crypto oracle $\Or'$, we note that the server's execution is now completely deterministic since all its crypto oracle
responses are fixed (the server's other input, the encoded database $\edb$, is unchanged as well since $\init$ makes no crypto oracle queries). Thus, the server's resulting communication using $\Or'$ with input database $\db$ will always match
the transcript $\tx$ generated using the original crypto oracle $\Or$.
With this modification, our result will now lower bound the total size of the online communication and the total number of server-side
crypto oracle queries (i.e., blackbox cryptographic operations).

\heading{Efficiently Encoding Crypto Oracles.}
It still remains to show that the $\hint$ algorithm can efficiently encode the newly sampled crypto oracle using $O(s)$ bits
as required by the dual PIR. To do this, we will use the random oracle $\RO$ that is shared across all the dual PIR algorithms.
Using $\RO$, $\hint$ will sample crypto oracles $\Or_1,\ldots,\Or_z$ such that $\Or_i(S) = \Or(S)$ for all $i \in [z]$. Furthermore,
$\recon$ can also re-sample the exact same crypto oracles using $\RO$. A first attempt might be for $\hint$ to find some index $i \in [z]$ that
enables finding a $s$-bit client storage $\st$ that successfully matches the communication transcript $\tx$ and encode both $\st$ and the index $i \in [z]$ using $\log_2(z)$ bits. However, one may need to try as many as $z = 2^{\Omega(|\tx|)}$
different sampled crypto oracles before finding one that successfully matches the communication transcript $\tx$ such as if
the communication transcript distribution is uniformly random.
Note, this is too large as encoding the index requires $\log_2(z) = \Omega(|\tx|)$ bits, and $|\tx|$ may be much larger than $s$, so it will not fit in the $s$-bit hint.

Instead, we design our $\hint$ algorithm to take a different approach. For convenience, we suppose each sampled crypto oracle $\Or_i$ is also associated with a uniformly random $s$-bit string $\st_i$ and client randomness $\omega_{\client, i}$ generated using the random oracle $\RO$ so that it is shared
across both $\hint$ and $\recon$.
For each $i \in [z]$ in increasing order, $\hint$ will do two things. First, it checks whether $\st_i$ is a feasible output the underlying PIR's $\setup$ algorithm for some internal random coin tosses. Secondly, $\hint$ checks whether the execution of the query algorithm $\query^{\Or_i}(\qrys, \st_i; \edb)$ (using $\omega_{\client, i}$ for client randomness) matches the communication transcript $\tx$, where $\edb \leftarrow \init(\db)$ is the recomputed encoded database.
The $\hint$ algorithm will maintain a counter $c$ tracking the number of times that a pair $(\Or_i, \st_i)$ ends up matching
the communication transcript $\tx$ (that is, every time the second check succeeds). For the first time when both checks performed
by the $\hint$ algorithm succeeds, the output of the $\hint$ algorithm will be the current counter value of the counter, $\hintout = c$.
In other words, the output of $\hint$ is the index $c$ such that the $c$-th pair $(\Or_i, \st_i)$ that reproduces the communication transcript $\tx$ under the underlying PIR's $\query$ algorithm
is the first one whose $\st_i$ is also a valid $s$-bit client storage output of the underlying PIR's $\setup$ algorithm.
This counter $H=c$ will be passed to $\recon$
that does not receive the database $\db$. It turns out that $\recon$ can check whether $(\Or_i,\st_i)$ generates transcript $\tau$ (even without $\db$) as we will show later. However, $\recon$ cannot check if $\st_i$ is output of $\setup$ that requires $\db$, but will use this counter instead.
Finally, we also show that this counter may always be encoded using at
most $3s$ bits except with $2^{-s}$ probability.

$\recon$ will iterate through all $z$ sampled crypto oracle and client storage pairs $(\Or_i, \st_i)$ and perform the second
check as done in $\hint$. We note that $\recon$ does not have the database $\db$ as input, so it cannot execute
the full $\query$ algorithm as done in $\hint$. Instead, $\recon$ will only execute the client-side of the $\query$ algorithm
using the server's communication in the transcript $\tx$ to simulate the server's algorithm that requires the database $\db$.
We critically leverage the fact that each crypto oracle $\Or_i$ is sampled to match $\Or(S)$ here. Due to this restriction, we know
that the server's responses will always match the transcript $\tx$. As a result, the client-side
simulation of the server during $\recon$ is performed correctly and will match the execution performed during $\hint$.
To produce the final answer, $\recon$ finds the $c$-th instance of the communication transcripts matching and returns the underlying
PIR's output when executing the client-side of the $\query$ algorithm.

\heading{Proving Correctness of Our Dual PIR.} Our correctness proof for the dual PIR proceeds in two steps. First, we show that the dual
PIR is correct when the arbitrary query $\qrys'$ used in $\advice$ is also given as input to the other steps of the dual PIR, $\hint$ and $\recon$. To prove this is correct, we show that the output distribution of our dual PIR is identical to the output distribution of a proper
execution of the underlying PIR using $\qrys'$ (except when the $\hint$ algorithm fails). Our proof shows that the dual PIR
samples from the same distribution as the underlying PIR in an inverted order: sampling communication transcripts first and client storage conditioned on the transcript.
As the
encoded database $\edb$ is then independent of $\Or$, a crypto oracle
resampled conditioned only on agreeing with $\Or$ on the server's query set
$S$ is distributed exactly as the real oracle conditioned on everything the
dual PIR has fixed (namely, $\edb$ and $\Or(S)$).

In the second step, we prove correctness of our dual PIR even when the $\hint$ and $\recon$ algorithms receive any query $\qrys \ne \qrys'$ different from $\qrys'$ used in $\advice$. If our dual PIR is incorrect
for some $\qrys \ne \qrys'$, we build a privacy adversary that distinguishes between transcripts generated from $\qrys$ and those from $\qrys'$.
At a high level, the adversary executes the last two steps of the dual PIR with the challenge transcript.
As the adversary chooses the database $\db$, the adversary may check
whether the dual PIR's output is correct. The adversary successfully distinguishes transcripts from the two query sequences by observing the difference in correctness between $\qrys$ and $\qrys'$.
Note that our PIR adversary 
makes a polynomial number of crypto oracle queries since its main cost is executing the underlying dual PIR.
In fact, our PIR adversary makes the same number of crypto oracle queries as an honest server makes during the $\query$ protocol of the underlying PIR.
As a note, our adversary may lazily sample a local random oracle that may be used the dual PIR (as it is
not used by the underlying PIR).
Thus, our lower bound even applies to weaker PIR schemes provide privacy only against
adversaries with a fixed polynomial number of queries.

\subsubsection{Communication Lower Bounds}

\heading{Small Server Cryptographic Operations.}
We note this setting is a direct corollary of our above computational lower bound. As a reminder, the $t$-bit advice in the dual PIR
consists of both the communication transcript and all server crypto oracle queries. If the server makes crypto oracle queries
on at most $o(n/s)$ amortized bits (or $o(n)$ total bits), this immediately implies that the communication must be
$\Omega(n)$ bits in total or $\Omega(n/s)$ amortized per retrieval.

\heading{Communication-Dependent Server Cryptographic Operations.}
We consider the setting where all the server crypto oracle queries depend only on the communication transcript. To our knowledge, this matches all OWF-based PIR schemes (in some cases trivially, when the server performs no cryptographic
operations). We
modify our dual PIR construction to omit all the server crypto oracle queries $S$ in the $t$-bit output of $\advice$.
Instead, the latter two steps
of our dual PIR, $\hint$ and $\recon$, may recompute the server crypto oracle queries $S$ using the communication transcript
and the crypto oracle $\Or$. So we obtain a $\Omega(n/s)$ amortized communication lower bound.

\heading{PIR with Perfect Privacy.}
Next, we consider schemes with perfect privacy that derive their
client query randomness from the random oracle and a seed in their state 
(making the online query algorithm deterministic for a fixed state and
oracle). At first glance,
this category of PIR may seem overly restrictive; however, our bound applies
to any such perfectly private PIR relative to an \emph{idealized} crypto oracle as achieved by prior works~\cite{CCS:RenMugSun24,EC:HPPY25}.
Perfect privacy enables a simpler dual PIR construction. If a PIR is so perfectly private, there
is no difference in how the dual PIR performs when $\advice$ is run with
the real query or a substitute query. Thus,
we can simply have $\hint$ output a random state $\st$ (which also
determines the client's query randomness) leading to a $2^{-s}$
loss in correctness probability. However, we can still apply the impossibility of
dual PIR when considering $k = \Omega(s)$.
In particular, this construction never resamples crypto oracles
($\recon$ replays the client with the real oracle $\Or$), so this is our one
lower bound where $\init$ may retain full crypto oracle access.

\subsubsection{Symmetric PIR without Online Public-Key Cryptography}

\heading{SPIR without Online Public-Key Cryptography.}
We present the first construction of SPIR that uses client-side preprocessing
to avoid any public-key operations during online queries.
Notably, our construction does not require per-client
storage by the server. If per-client storage is allowed, then
the problem is related to OT-extension~\cite{C:IKNP03}.
We start from recent PIR with client preprocessing
schemes~\cite{EC:CorKog20,EC:CorHenKog22,CCS:RenMugSun24,EC:HPPY25},
that preprocess a database
by storing parities of small random sets that we refer to as hints.
To add data privacy, we have the server one-time pad
each hint during offline preprocessing without learning the hints
(using generic 2PC protocol for example).
At query time, the client sends the hint alongside the normal
PIR query so the server can remove the one-time pad and 
return the correct answer.
There are additional complications around data privacy as
prior schemes use parallel repetition or require the server to return two XORs~\cite{CCS:RenMugSun24,EC:HPPY25}.
We resolve these issues with a more complex hint structure including backup hints~\cite{EC:CorHenKog22} and 
the client randomly permuting hints to maintain query privacy.

\heading{Limits of SPIR with Preprocessing.}
We also prove stronger limits of SPIR with preprocessing beyond our broad
PIR lower bounds. 
In the random oracle model, we prove that no perfectly private,
single-server SPIR relative with $s$ 
bits of client-side preprocessing 
can support batch queries of size $k > s$
regardless of online query communication costs, (even if the preprocessing
is done ``for free'' without the adversary seeing it).
Furthermore, we extend prior work~\cite{STOC:GIKM98} to prove that
offline and online phases of SPIR supporting $k$ queries require
$k$-bits of correlated randomness.
Our proofs utilize compression techniques as opposed to 
our more complex reductions using dual PIR.

\section{Related works}\label{sec:related}

We note that there is a long literature on single-server PIR without preprocessing (see~\cite{FOCS:KusOst97,EC:CacMicSta99,EC:Paillier99,PKC:DamJur01,STOC:IKOS04,ICALP:GenRam05,PKC:OstSke07,SP:ACLS18,TCC:GenHal19,ESORICS:ParTib20,CCS:MugCheRen21,SP:MenWu22,SP:MugRen23,USENIX:PatSeoYeo23,USENIX:BPSY24,CCS:BurMenWu24} and references therein as examples). In this section, we
mainly survey results concerning single-server PIR with preprocessing.

\heading{PIR with Server-Side Preprocessing.} The model was first proposed
by Beimel, Ishai and Malkin~\cite{C:BeiIshMal00} in the server-side preprocessing setting where the server can preprocess
the database in hopes of faster online queries. There have been several works in this line including~\cite{TCC:BIPW17,TCC:CanHolRic17,EC:HOWW19,TCC:BoyHolWei19,SODA:PerYeo22} leading
to the breakthrough work of Lin, Mook and Wichs~\cite{STOC:LinMooWic23} that obtained $\tilde{O}(1)$ online query communication
and computation in the public-key DEPIR with server preprocessing model from the ring LWE assumption. Recent work also ruled out the existence of restricted classes
of such schemes built from blackbox cryptography~\cite{EC:LinMooWic25}.
We note that the lower bounds of~\cite{EC:LinMooWic25} apply to DEPIR
schemes whose preprocessing may query the crypto oracle (under structural
restrictions on the query protocol), whereas our DEPIR result
(Section~\ref{sec:depir}) allows arbitrary query protocols but requires the
preprocessing to make no oracle queries, making
the two results incomparable.
Lastly, a recent work~\cite{eprint:CIMR25} presented a construction
of secret-key DEPIR
with $O(n^\varepsilon)$ communication from
the LPN assumption with parameters that are not known to imply public-key cryptography.

\heading{PIR with Client-Side Preprocessing.}
Using client-side preprocessing to reduce the amount of server computation was first studied by Patel, Persiano, and Yeo~\cite{CCS:PatPerYeo18} that showed how to obtain a sublinear number of server public-key operations in a single-server PIR.
The breakthrough work of Corrigan-Gibbs and Kogan~\cite{EC:CorKog20} showed that it is possible to even obtain
sublinear total computation during query time. There has been a long line of follow-up works
with improved constructions such as~\cite{USENIX:KogCor21,C:SACM21,EC:CorHenKog22,USENIX:MZRA22,C:LazPap23,C:Yeo23,EC:GhoZhoShi24,USENIX:LazPap24}.
By relying upon FHE during query time, it was shown that single-server schemes could obtain
both $\tilde{O}(1)$ communication and sublinear computation~\cite{EC:ZLTS23,TCC:LazPap23}.

More recent works starting from Piano~\cite{SP:ZPZS24} along with follow-up works~\cite{CCS:RenMugSun24,EC:HPPY25,EC:WanRen25} present single-server PIR with preprocessing schemes
using only one-way functions with $t = \tilde{O}(n/s)$ communication with $s$-bit client storage.
There also exist works that use heavy public-key cryptography during the offline phase to obtain sublinear communication
while the query algorithm only uses one-way functions (see~\cite{EC:CorKog20,EC:CorHenKog22,CCS:FLLP24}).
However, these constructions only support $\tilde{O}(s)$ queries per offline phase when using $s$-bit client storage. If they wish to support more queries, they must re-execute their offline phase that requires public-key operations during queries (or streaming the entire database).
Finally, we note that all prior single-server constructions (to our knowledge) may be modified to support
an unbounded number of queries in the following way (this is effectively the approach used in~\cite{SP:ZPZS24,CCS:RenMugSun24,EC:HPPY25,EC:WanRen25} as well).
Roughly speaking, in these constructions, a single offline phase supports $\tilde{O}(s)$ queries for $s$-bit client storage.
Over the $\tilde{O}(s)$ queries, the client can stream the entire database and perform the offline phase locally to compute
the next $s$-bit client storage. This avoids the usage of any public-key cryptography during online queries, but incurs $t = \Omega(n/s)$
amortized communication that is shown to be required by our lower bound.
One can view our lower bound as showing that this approach is optimal in terms of communication complexity.
Finally, prior works also consider information-theoretic security including~\cite{EC:CorKog20,C:IshShiWic24,TCC:SinWeiZik24}.

\heading{Lower Bounds.}
Di Crescenzo, Malkin, and Ostrovsky \cite{EC:DiCMalOst00} were the first to show that sublinear communication (i.e., non-trivial) single-server
PIR without any preprocessing implies oblivious transfer. As oblivious transfer may be used to build key agreement,
this may be combined with the well-known separation of Impagliazzo and
Rudich \cite{STOC:ImpRud89} (and improvement by
Barak and Mahmoody~\cite{C:BarMah09}), which shows that there is no key agreement
between unbounded computational parties with polynomial queries to a random oracle.
Furthermore, there are known server computation lower bounds for PIR without preprocessing.
Beimel, Ishai and Malkin~\cite{C:BeiIshMal00} show that PIR (both single- and multi-server) require
linear server computation.

For single-server PIR, Dujmovic and Hajiabadi~\cite{EC:DujHaj24} show
that a linear number of public-key operations is necessary without preprocessing.
In the preprocessing model, their lower bound shows that the server's total cryptographic operations (across
both the offline and online phases) must be $\Omega(n)$. 
Their result makes no claims about, for example,
the online server cryptographic operations for PIR with preprocessing in
which the server makes $\Omega(n)$ public-key operations.
In contrast, our lower bound is stronger as it shows that $\Omega(n/s)$ amortized cryptographic bit operations
must be performed by the server during online queries.

For the preprocessing model, there are known lower bounds on the trade-offs between
client storage and query time in both the client~\cite{EC:CorKog20,EC:CorHenKog22,EC:Yeo23} and server~\cite{C:BeiIshMal00,SODA:PerYeo22} preprocessing models. However, all of these lower
bounds only apply to the restricted class of constructions with very strong non-encoding server
requirements such that the server must use the $n$-bit database without modification to serve PIR queries.
In contrast, our computation lower bounds apply to schemes in which the
server may arbitrarily encode the database for use in answering queries, so
long as the encoding is computed without crypto oracle queries.
Furthermore, none of the techniques in this line of work yield any communication lower bounds.

Ishai, Shi, and Wichs \cite{C:IshShiWic24}
were the first to present communication lower bounds for single-server PIR with client preprocessing.
First, they showed that there exists no information-theoretic construction with $s$-bit client storage
and query communication $t = O(n/s)$. In particular, such a construction implies the existence of
a one-way function. For the restricted class of database-oblivious and single roundtrip query constructions,
they showed that a $s$-bit client scheme with $t = o(n/s)$ query communication would imply a
separation between complexity classes SZK and BPP that remains an open problem in complexity theory
in many natural settings (such as relative to a random oracle).
Our new result cannot be directly compared to~\cite{C:IshShiWic24}.
Notably, we show a stronger (unconditional) lower bound on the use of blackbox cryptography for a class of PIR that is broader in several respects (multi-round and database-dependent queries). However, our result does not have the same complexity theoretic implications. In particular,
if one were to show that another assumption $X$ (e.g., learning parity with noise, attribute-based encryption, etc.)
or a white-box reduction from one-way functions
implied a too-efficient PIR with preprocessing, then a corollary of the prior work~\cite{C:IshShiWic24} would be that $X$
implies a hard problem in SZK. This is something which our result does not comment on.
Additionally, both the barrier in~\cite{C:IshShiWic24} and
our communication lower bounds only apply to restricted classes of schemes. We note the imposed restrictions
are incomparable. The barrier in~\cite{C:IshShiWic24} required the query algorithm to be single roundtrip and ensure that all client requests to be database-oblivious (i.e., client requests may only depend on the queried index and internal randomness). In contrast, our work makes no restrictions
on the communication (the query algorithm may have multiple rounds and depend on the database). Instead, we
require one of the following: the server to make $o(n/s)$ cryptographic operations, all server cryptographic operations depend
only on the communication transcript or obtain perfect privacy in the idealized model.
In the first two cases, we additionally require the server's one-time
database encoding to be computed without cryptographic operations. The
perfect-privacy case does not need this restriction and instead applies to
schemes deriving their query randomness from the random oracle (see
Section~\ref{sec:comm-perfect}).

\heading{Symmetric PIR.}
Symmetric PIR (SPIR), where the contents of the database are also protected, has been studied
for decades (see~\cite{STOC:GIKM98,TCC:FIPR05,USENIX:ALPRSS21} and references therein).
There are ways to generically construct SPIR from (non-symmetric) PIR using, for example, oblivious
PRFs \cite{TCC:FIPR05,USENIX:ALPRSS21}, which works in the preprocessing setting for certain constructions.
To our knowledge, these generic transformations require public-key operations at query time.
One could also try to use these public operations to the offline phase (such as using
OT-extension \cite{C:IKNP03}), but this would require the server to store client-specific state that is undesirable as
the server storage would grow with the user base.

Notice that for SPIR (unlike PIR), we \emph{must} have public-key operations in the offline
or online settings, since it is equivalent to oblivious transfer. This means we cannot use
tricks like streaming the database to the client in the offline phase that were used in
prior PIR works. We discuss this further in Section~\ref{sec:spir}.

\section{Preliminaries}\label{sec:prelim}

\heading{Notation.}
Throughout our paper, we use $n$ to denote the size of a database $\db$, which we consider
to be a bitstring $\db \in \bits^n$. We additionally take $\secpar$ to be an integer which represents
the security parameter of the schemes in question. We also make use of the notation
$[n] = \{1,\ldots,n\}$. For a tuple of indices/queries $\qrys = (\qry_1,\ldots,\qry_k)\in [n]^k$,
we  use $\db[\qrys]$ to  denote $(\db[\qry_1],\ldots,\db[\qry_k])$. 

In this work, each algorithm denotes a possibly randomized (Turing machine) computation,
which takes zero or more inputs and has one or more outputs. A (two-party) protocol is
a pair of interactive computers each with their own zero or more inputs. The messages passed
between the computers define the transcript. At the end of the protocol,
the machines output their results. We denote a protocol $\Pi$ with inputs $x_i$, outputs $y_i$,
and transcripts $\tau$, as $\Pi(x_1;x_2) \to (y_1;y_2 \mid \tau)$. If one party has no output,
we denote its output as $\bot$. Sometimes, we ignore the transcript $\tau$ when it is
not relevant.
We say an algorithm or protocol is \emph{efficient} if there exists a polynomial $p$, such that
for any input(s) $x$, the runtime of algorithm or protocol is at most $p(|x|)$.

In our pseudocode, we use $x\gets y$ to denote assigning $x$ the value of $y$.
Likewise we use $x\getsr y$ to denote sampling (uniformly)
from a set $y$ and assigning the value to $x$ (e.g. $x\getsr [n]$ denotes selecting a random
integer between $1$ and $n$ inclusive). The item $y$ may also be a randomized
function, which could just be viewed as running the function with fresh randomness.
We use $\negl(\cdot)$ as a placeholder for some negligible function, i.e.,
one where for every polynomial $p$ there is some $n$ such that,
$\negl(n') < 1/p(n')$, for  every $n' \ge n$.
For conceptual clarity, we use \textbf{bold} font to denote random variables
when measuring probabilities or other information-theoretic quantities. For
example, we may have a variable $\db \in \bits^n$ which represents a
database, which corresponds to $\bdb$ to represent a uniformly random
choice of database.

\heading{Information theory.}
At various points in our results, we are concerned with the minimum entropy or
min-entropy that we define in \defref{def:min-entropy}. This quantity essentially
measures the predictability of a random variable. In particular, if a random variable $\mathbf{X}$
has min-entropy at least $h$, then an input-less algorithm can predict the
value of $\mathbf{X}$
only with probability at most $2^{-h}$. In particular, for all $\A$,
$\Pr_{\mathbf{X}}[\A = \mathbf{X}] \le 2^{-\minH(\mathbf{X})}$.

In our results, it is also useful for us to consider a slightly more general
version of min-entropy as well, called conditional min-entropy. Although there are few ways
to define this quantity, our choice is related to the predictability of a random variable
$\mathbf{X}$, given some side information (also a random variable) $\mathbf{Y}$, which may be correlated with $\mathbf{X}$.
Using our definition however, it is easy to see that for any algorithm $\A$,
$\Pr_{\mathbf{X},\mathbf{Y}} [\A(\mathbf{Y}) = \mathbf{X}] \le 2^{-\minH(\mathbf{X}\mid \mathbf{Y})}$.

\begin{definition}\label{def:min-entropy}
	Let $\mathbf{X},\mathbf{Y}$ be random variables supported on $\calX$ and $\calY$ respectively.
	The \emph{min-entropy} of $\mathbf{X}$ is defined as
	\[
		\minH(\mathbf{X}) := - \log \left(\max_{x\in \calX} \Pr_{\mathbf{X}} [\mathbf{X} = x]\right).
	\]
	We further define the \emph{conditional min-entropy} of  $\mathbf{X}$ given $\mathbf{Y}$ as
	\[
		\minH(\mathbf{X}\mid \mathbf{Y}) := - \log \left( \Ex_{y\sim \mathbf{Y}}
		\left[\max_{x\in \calX} \Pr_{\mathbf{X}} [\mathbf{X} = x \mid \mathbf{Y} = y]\right]\right).
	\]
\end{definition}

We use the following as a supporting lemma
in the proof of our main result. This lemma provides a concrete bound on the amount
that min-entropy can drop when given some additional side-information.

\begin{lemma}[\cite{EC:DodReySmi04}]\label{lem:cond-minh}
	If $\mathbf{X},\mathbf{Y},\mathbf{Z}$ are random variables and $\mathbf{Y}$ is supported on at most $k$ elements,
	then $$\minH(\mathbf{X} \mid \mathbf{Y}, \mathbf{Z}) \ge \minH(\mathbf{X}\mathbf{Y}\mid \mathbf{Z}) - \log k \ge \minH(\mathbf{X} \mid \mathbf{Z}) - \log  k.$$
\end{lemma}

In addition to min-entropy, we denote the statistical distance between two random
variables $\mathbf{X}$ and $\mathbf{Y}$ with the same support $\mathcal{U}$ as
$\Delta(\mathbf{X},\mathbf{Y}):= \frac12 \sum_{u\in \mathcal{U}} |\Pr[\mathbf{X} = u] - \Pr[\mathbf{Y} = u]|.$
We also use the fact that the best unbounded distinguisher between $\mathbf{X}$ and $\mathbf{Y}$
successfully identifies the correct random variable with probability
$1/2 + \Delta(\mathbf{X},\mathbf{Y})/2$.

\subsection{Crypto Oracles}\label{sec:prelim:oracles}

This work considers relativized worlds where some or all
parties have access to (possibly randomized) oracles. In particular,
we denote an algorithm $\A$ with access to an oracle $\Or$ as $\A^{\Or}$.
In relevant settings, the presence of a random
oracle requires us to change the sample space of a probability
experiment
to additionally consider the underlying randomness of the random oracle $\Or$
(for example, our definitions of PIR with preprocessing in~\defref{def:pir}).

We will utilize the notion of {\em crypto oracles}\footnote{In \apref{sec:ap-oracles}, we discuss how our results generalize beyond crypto oracles.} that encapsulates the usage of powerful blackbox cryptography primitives. We use the definition introduced by Lin, Mook, and Wichs~\cite{EC:LinMooWic25} that considers stateless simulators, but first considered
by Dujmovic and Hajibadi~\cite{EC:DujHaj24}.\footnote{Note, the definition in~\cite{EC:DujHaj24} considers stateful simulators, but there does not seem to be any meaningful cryptographic primitive not captured by stateless simulators as pointed out in~\cite{EC:LinMooWic25}.}

\begin{definition}[Due to \cite{EC:LinMooWic25}]
A crypto oracle is a function $\mathcal{B}^R$ where $B$ is a stateless, polynomial time, deterministic Turing machine with oracle access to a secret random function $R: \{0,1\}^* \rightarrow \{0,1\}$.
\end{definition}

In the same work~\cite{EC:LinMooWic25}, it was shown that crypto oracles may be used to implement both virtual blackbox obfuscation and the generic multilinear group model. We extend this result to (unsurprisingly) show that one can
implement public-key fully homomorphic encryption (with key and ciphertext
length independent of the size of the circuit evaluated), which can be used to
construct PIR directly. We refer the reader to~\cite{KatLin14}
for security definitions relevant for public-key encryption.

\begin{lemma}
\label{lem:fhe-oracle}
There exists a crypto oracle $\B^R$ that implements public-key fully
homomorphic encryption with $\mathsf{poly}(\lambda)$-size keys.
\end{lemma}
\begin{proof}
We construct a crypto oracle $\B^R$ providing interfaces
$(\mathsf{KG}^R, \mathsf{Enc}^R, \mathsf{Eval}^R, \mathsf{Dec}^R)$,
where $R : \{0,1\}^* \to \{0,1\}^{\lambda}$ is the underlying
secret random function.\footnote{We use notation with $\lambda$-bit outputs
to simplify our presentation. This can be implemented without loss of
generality by querying $R$ with distinct prefixes.}

Since $\B^R$ must be deterministic, the randomized algorithms
$\mathsf{KG}$ and $\mathsf{Enc}$ take their random coins
$\rho \in \bits^\lambda$ as an explicit input, and all indices are derived
using $R$ with the distinct prefixes $\mathsf{kg}$, $\mathsf{enc}$, $\mathsf{eval}$, and $\mathsf{pad}$.

\begin{itemize}
    \item $\mathsf{KG}^R(1^\lambda; \rho)$: Set $sk = \rho$,
    set $pk = R(\mathsf{kg} \| sk)$ and output $(pk, sk)$.

    \item $\mathsf{Enc}^R(pk, m; \rho)$: Compute the index
    $i = R(\mathsf{enc} \| pk \| \rho)$ and output
    $\mathsf{ct} = (i, R(\mathsf{pad} \| pk \| i) \oplus m)$.

    \item $\mathsf{Eval}^R(pk, f, \mathsf{ct}_1, \ldots, \mathsf{ct}_\ell)$:
    Parse each $\mathsf{ct}_j = (i_j, c_j)$, recover $m_j = c_j \oplus
    R(\mathsf{pad} \| pk \| i_j)$. Then, compute $m^* = f(m_1, \ldots,
    m_\ell)$ and the index $i^* = R(\mathsf{eval} \| pk \| f \| \mathsf{ct}_1 \| \cdots \| \mathsf{ct}_\ell)$,
    and output $(i^*, R(\mathsf{pad} \| pk \| i^*) \oplus m^*)$.

    \item $\mathsf{Dec}^R(sk, \mathsf{ct})$: Recover $pk = R(\mathsf{kg} \| sk)$,
    parse $\mathsf{ct} = (i, c)$, and output $c \oplus R(\mathsf{pad} \| pk \| i)$.
\end{itemize}

Both $pk$ and $sk$ have size $\mathsf{poly}(\lambda)$.
Standard correctness is immediate.
Encryption is public-key: $\mathsf{Enc}^R$ requires only $pk$ and makes no
use of $sk$. Note that IND-CPA security holds against any adversary making $q = \poly(\lambda)$
queries to $\mathcal{B}^R$. The challenge ciphertext is
$(i^*, R(\mathsf{pad} \| pk \| i^*) \oplus m_b)$ for $i^* = R(\mathsf{enc} \| pk \| \rho^*)$
with uniformly random coins $\rho^*$. The pad $R(\mathsf{pad} \| pk \| i^*)$ is revealed
to the adversary only through a ciphertext under the same index $i^*$ whose plaintext
the adversary knows. This requires an $\mathsf{Enc}^R$ or $\mathsf{Eval}^R$ query whose derived
index equals $i^*$, or a $\mathsf{Dec}^R$ query with a secret key $sk'$ such that
$R(\mathsf{kg} \| sk') = pk$. Since $R$ is a random function evaluated on distinct inputs
(the prefixes ensure that index derivations never coincide with pad evaluations) and $\rho^*$ is
hidden, each of these happens with probability at most $q / 2^\lambda$.
Otherwise, the pad is uniformly random and independent of the
adversary's entire view, so the challenge
ciphertext is uniformly random regardless of $b$, and the adversary
gains no advantage.
Homomorphic correctness follows because $\mathsf{Eval}^R$
decrypts and re-encrypts internally using $R$, without revealing any plaintext
to the caller.
\end{proof}

\begin{theorem}[Due to \cite{EC:LinMooWic25} and \lemref{lem:fhe-oracle}]
\label{thm:oraclecrypto}
There exists crypto oracles that implement the following:
\begin{itemize}[noitemsep]
\item (Trivially) random oracles,
\item Virtual blackbox obfuscation for Turing machines,
\item Generic multilinear groups.
\item Public-key fully homomorphic encryption.
\end{itemize}
\end{theorem}

The above is very powerful as it enables instantiating a wide range of powerful cryptographic primitives in a blackbox manner
including one-way functions, fully homomorphic encryption and obfuscation,
which will be relevant as they have been
used in past constructions of PIR with preprocessing.

We note that a single crypto oracle may be used to implement arbitrarily many different crypto oracles
that implement different cryptographic primitives and/or ideal models. In particular, a single random function $R$ may be used to derive an arbitrary number of independent random functions (for example,
by prepending with distinct prefixes). Therefore, without loss of generality, we can consider a single crypto oracle
that implements multiple (or all) of the above cryptographic primitives and idealized models.

We will need the following result that shows it is possible to sample a random crypto oracle such that a polynomial number
of queries agree with a previously sampled crypto oracle using a random oracle $\RO$.
In particular, we consider a crypto oracle $\Or$ along with a polynomial set of queries $S$. We want to sample another random crypto oracle $\Or'$ such that $\Or(q) = \Or'(q)$ for all $q \in S$. For convenience, we will denote the set of queries and results as $\Or(S) = \{(q, \Or(q)) \mid q \in S\}$.
The following Lemma~\ref{lem:oracle_resample} is inspired by a step at the 
heart of blackbox separations in the random oracle
model~\cite{STOC:ImpRud89,C:BarMah09}. After learning the oracle queries 
relevant to a transcript, an attacker samples a fresh random oracle that agrees 
with them and simulates the honest parties. For a random oracle this resampling 
is trivial, whereas we show it remains possible for crypto oracles, where the 
resampled oracle must again be of the form $\mathcal{B}^{R'}$.

\begin{lemma}
\label{lem:oracle_resample}
Let $\Or$ be a crypto oracle and let $S$ be a set of queries for $\Or$ of size $|S| = \poly(\lambda)$.
For any $\zeta = \poly(\lambda)$, there exists a deterministic algorithm $\mathcal{A}$ with access to a random oracle $\RO$ such that $\mathcal{A}^{\RO}(\Or(S))$ outputs a new uniformly random crypto oracle $\Or'$ with $\Or(q) = \Or'(q)$ for all $q \in S$, except with probability $e^{-\zeta}$ over the random choice of $\RO$.
Moreover, conditioned on the algorithm succeeding, the output $\Or'$ is
distributed exactly as a uniformly random crypto oracle conditioned on the
event that $\Or'(q) = \Or(q)$ for all $q \in S$.
\end{lemma}
\begin{proof}
As $\Or$ is a crypto oracle, there exists some stateless, deterministic simulator $\mathcal{B}^R$ that makes at most $t = \poly(\lambda)$ queries to the private random function $R$. Across all oracle queries $\Or(q)$ for all $q \in S$, we note that $\mathcal{B}^R$ makes
at most $|S| \cdot t$ queries to $R$. First, we use $\RO$ to generate $k = 2^{|S| \cdot t} \cdot \zeta$ different random $\RO_1,\ldots,\RO_k$ by
setting $\RO_i(\cdot) = \RO(i \mid\mid \cdot)$ where $i$ is represented using a $\ceil{\log_2 k}$-bit string.
Note that $\ceil{\log_2 k} = O(|S| \cdot t + \poly(\lambda)) = \poly(\lambda)$ since both $|S|$ and $t$ are polynomial-sized, so
all inputs to $\RO$ are poly-sized.
Afterwards, our algorithm constructs $\Or_i$ by executing $\mathcal{B}^{\RO_i}$ and checks whether $\Or_i(q) = \Or(q)$ for all $q \in S$. The final output oracle $\Or'$ is $\Or_i$ for the first index $i$ satisfying the above condition. Now, we bound the probability
that the algorithm finds such an $\Or_i$. First, we note that the probability that any $\Or_i$ constructed using $\RO_i$
will match $\Or$ at all queries in $S$ is at least $2^{-|S| \cdot t}$, which is the probability that $\RO_i$ exactly matches the underlying private random function $R$ used by the simulator $\mathcal{B}^R$ corresponding to $\Or$.
Note, there might exist more choices of random functions that result in oracles that agree on queries in $S$, but that would only increase
this probability.
Therefore, the probability that
all $k = 2^{|S| \cdot t} \cdot \zeta$ random oracles fail to produce a crypto oracle that matches the original oracle $\Or$ on all queries in $S$
is at most $(1 - 2^{-|S| \cdot t})^{k} \le e^{-\zeta}$ using the fact that $(1 - 2^{-|S| \cdot t})^{2^{|S| \cdot t}} \le e^{-1}$.
Finally, we argue the distribution of the output conditioned on success.
Each candidate oracle $\Or_i = \mathcal{B}^{\RO_i}$ for an independent
uniformly random function $\RO_i$ is an independent sample of a uniformly
random crypto oracle, and the algorithm outputs the first candidate agreeing
with $\Or$ on all of $S$. By the standard rejection sampling argument, the
first accepted sample among independent and identically distributed
candidates, conditioned on acceptance occurring, is distributed as a single
candidate conditioned on acceptance; that is, as a uniformly random crypto
oracle conditioned on agreeing with $\Or(S)$.
\end{proof}

Note, as the above algorithm $\mathcal{A}$ is deterministic, it will output the same newly sampled oracle $\Or'$ when given access to the same random oracle $\RO$ as well as input $\Or(S)$.

\subsection{PIR with Preprocessing}\label{sec:pir-def}

We focus on single-server PIR in the client preprocessing model, a notion
that we formally present in \defref{def:pir}.
In this model, the client is allowed to preprocess the database 
$\db\in\bits^n$ during the offline phase;
then, in the online phase, the client issues its set of queries $\qrys$.
In our definition, we choose a slightly weaker model that enables constructions with worse privacy and weaker functionality.
As we prove lower bounds, this makes our results stronger 
as they will apply to a wider range of schemes.
We remind the reader that the syntax $\db[\qrys]$ for a tuple $\qrys = (\qry_1,\ldots,\qry_k)$ is used to denote the tuple $(\db[\qry_1],\ldots,\db[\qry_k])$.

\begin{definition}[Private Information Retrieval]\label{def:pir}
A $k$-batch \emph{private information retrieval with preprocessing scheme (PIR)}
is a triplet of efficient two-party protocols
with access to an oracle $\Or$,
$\pir^{\Or} = (\init^{\Or},\setup^{\Or},\query^{\Or})$, parameterized
by a database size $n$, batch size $k$, total communication $t$,
and state size $s$ with the following syntax:\footnote{Note that $k,t,s$ are typically functions of
$\secpar$ and $n$ that we omit for simplicity.}
\begin{itemize}
\item $\init^{\Or}(1^{\secpar}, \db) \to \edb$, a (possibly randomized)
algorithm run by the server to encode the original database $\db \in \bits^n$
into some encoded database $\edb$.
At the end of the protocol, the client receives no output and the server receives an encoded database $\edb$.
\item $\setup^{\Or}(1^{\secpar} ; \edb) \to (\st; \bot)$, where the client receives
as input the security parameter $1^\secpar$ and the server receives the encoded database
$\edb$. At the end of the protocol, the client receives a state $\st\in \bits^s$ and the
server receives no output.
\item $\query^{\Or}(\st, \qrys ; \edb) \to (\ans; \bot \mid \tau)$, where the client
takes its precomputed state $\st$ and an ordered set of queries 
$\qrys\in [n]^k$ and the server
receives the encoded database $\edb$. At the end of the protocol,
the client receives $\ans \in \bits^k$
and the server receives no output.
We denote the messages sent between the client and server as $\tx \in \bits^t$. Where relevant,
we will also separate $\tx = (\tx_\client, \tx_\server)$ to distinguish 
the messages $\tx_\client$ sent by the client to the server and the messages
$\tx_\server$ sent by the server to the client.
\end{itemize}

A PIR is \emph{$\varepsilon$-correct} if for all sufficiently large $\secpar$,
$\qrys\in [n]^k$, and $\db \in \bits^n$,
\[
    \Pr%
    \left[\query^{\Or}(\bst, \qrys ; \bedb) = (\db[\qrys]; \bot)\Big|
    \substack{
    	\bedb \getsr \init^{\Or}(1^\secpar,\db)\\
        (\bst;\bot) \getsr \setup^{\Or}(1^{\secpar} ; \bedb)
    }
	\right]
    \ge \varepsilon.
\]

A PIR is \emph{$\delta$-private} if for all sufficiently large $\secpar$,
tuples $\qrys_0,\qrys_1 \in [n]^k$, $\db \in \bits^n$, and
adversaries $\A$ issuing $\poly(\secpar)$ queries to $\Or$,
\[
    \Pr
    \left[\A^{\Or}(1^{\secpar},\btx,\bedb) = b ~\Bigg|~
    \substack{
    	\bedb \getsr \init^{\Or}(1^\secpar,\db)\\
        (\bst;\bot) \getsr \setup^{\Or}(1^{\secpar} ; \bedb)\\
        b \getsr \bits\\
        (\bans;\bot \mid \btx) \getsr  \query^{\Or}(\bst,\bqrys_b;\bedb)
    }\right]
    \le \frac{1}{2} + \delta.
\]
When $\delta=0$ we say the PIR is \emph{perfectly private}.
\end{definition}

We note that the above %
definition is flexible enough to cover almost all prior studied PIR definitions.
For PIR with client preprocessing, we essentially allow the server to perform arbitrary encoding
of the database that may be used by the client to generate their private state.

While our definition places no restriction on the encoding algorithm,
our computation lower bounds will apply to the class of schemes whose
encoding is computed without access to the crypto oracle, which we
formalize as follows.

\begin{definition}[Oracle-Free Encoding]\label{def:oracle-free}
A PIR $\pir^{\Or} = (\init^{\Or},\setup^{\Or},\query^{\Or})$ has an
\emph{oracle-free encoding} if the encoding algorithm $\init$ makes no queries
to the oracle $\Or$. In this case, we drop the oracle superscript and simply
write $\init(1^\secpar, \db) \to \edb$.
\end{definition}

Our computation lower bounds and our first two communication lower
bounds apply to schemes with oracle-free encodings, while $\setup$ and
$\query$ may use the oracle $\Or$ arbitrarily. To our knowledge, this
captures every existing single-server PIR with client preprocessing
construction most of which have the server store the database unmodified
and all cryptographic work occurs inside the interactive
$\setup$ and $\query$ protocols, including FHE-based offline
phases~\cite{EC:CorKog20,EC:CorHenKog22,CCS:FLLP24}.
Our definition also encompasses public-key (and unkeyed)
doubly efficient PIR (DEPIR)~\cite{TCC:BIPW17,STOC:LinMooWic23}:
the server runs, inside $\init$, a key-generation step together with the
database encoding, and outputs both the public key and the encoded database as
$\edb$ (discarding the secret key), while the $\setup$ algorithm
is essentially ignored as the client receives no state (see
Section~\ref{sec:depir}).

On the other hand, our main lower bounds do not apply to schemes whose
database encoding queries the crypto oracle (see Remark~\ref{rem:init-ext}
for plausible relaxations). In particular, our lower bound aligns with the server's database encoding for the state-of-the-art public-key DEPIR~\cite{STOC:LinMooWic23} but not earlier schemes that use cryptographic preprocessing for the server's database encoding~\cite{TCC:BIPW17}. Although, we note none of these schemes make blackbox usage of cryptography and, thus, are outside the scope
of our lower bounds.
Similarly, our definition itself does not capture
secret-key DEPIR, where the server's encoded database may depend
on the client's privately generated state~\cite{TCC:BIPW17,TCC:CanHolRic17,EC:LinMooWic25,eprint:CIMR25}, as
our proofs require the server's encoding to be performed \emph{first} before generating the client's private state.

\heading{Offline Phase.}
One can note that our definition is {\em essentially agnostic} to
$\setup$. In particular,
the adversary is not given the transcript for the offline phase when
$\setup$ is executed.
As a result, our lower bounds
will apply to all choices of $\setup$ including the two popular choices of 
relying on fully-homomorphic encryption~\cite{EC:CorKog20,EC:CorHenKog22,CCS:FLLP24} 
and streaming the entire database~\cite{SP:ZPZS24,CCS:RenMugSun24,EC:HPPY25}.
In fact, our lower bound applies even
if we assume a setting where the offline phase is run securely and the adversary only appears during the online query portion.

\heading{Batch Queries.}
Our definition considers the setting of batch PIR
where it is assumed that all $k$ queried entries are given at once.
Again, this only strengthens our lower bound as it applies to both batch and standard, single-query PIR.
In particular, it is easy to see that it is straightforward to construct a $k$-batch query PIR from a single-query PIR by
performing $k$ queries in sequence.

\heading{Unrestricted Query Algorithm.}
Additionally, we note that our definition does not make any restrictions on the query algorithm of
the PIR with preprocessing scheme. In contrast, we note that the prior work~\cite{C:IshShiWic24}
considers a restricted class of query algorithms that are database-oblivious (client requests are entirely independent of the
database) and single roundtrip. Our lower bounds apply to schemes beyond these restrictions including those that are multiple roundtrips and may depend on the contents of the database.

\subsection{Subkey prediction}\label{sec:subkey-def}

We will rely upon a result about subkey prediction due to Bellare, Kane, and Rogaway~\cite{C:BelKanRog16}.
Subkey prediction considers the ability of an adversary to predict a uniformly random subkey of a large key, given some short leakage (side information).
In particular, the large key consists of $n$ independent and uniformly random bits and the adversary is interested in a subkey consisting of exactly $k$ uniformly random locations in the large key. 
Additionally, the %
adversary obtains $t$ bits of arbitrary leakage (side information)
of the large key, chosen independently of the subkey's uniformly random locations.
For example, one can imagine the adversary is able to preprocess the large key of $n$ bits into
any arbitrary $t$ bits of information before the indices of the subkey are chosen. Afterwards, the $k$ uniformly random locations of the subkey are provided and the adversary must predict the $k$ bits of the subkey
using the $t$ bits of leakage.

In more detail, suppose a random $n$-bit key $\bdb \in \{0,1\}^n$ is chosen. The adversary
is first able to output $t$-bit leakage denoted $\A_1(\bdb)$. Afterwards, $k$ random locations are chosen $\bqrys \in [n]^k$. Lastly,
the adversary is given the $t$-bit leakage and $k$ random locations and outputs its guess denoted by $\A_2(\bqrys, \A_1(\bdb))$.
We use the following lemma directly from~\cite{C:BelKanRog16}.
Before we present the lemma, we recall some notation
that will be important.
We denote by $B_n(r)$ the radius $r$ ball centered at the origin of
the $n$-dimensional boolean hypercube; that is,
all points of $\{0,1\}^n$
whose distance is at most $r$ from the all-zero vector.
Then, we see that  $|B_n(r)| = \sum_{i = 0}^r {n \choose i}$.
Additionally, we denote the radius
$\mathsf{rd}_n(N) :=
\max\{\hat{r}\in \mathbb{Z}_{\ge 0} : |B_n(\hat{r})| \le N\}$ 
as the radius of the largest ball centered at the all-zero vector 
of the boolean hypercube that contains at most $N$ points.

\begin{lemma}[Subkey Prediction \cite{C:BelKanRog16}]\label{lem:subkey}
For sufficiently large integers $n, t, k$ and for all
(computationally unbounded) adversaries $\A=(\A_1,\A_2)$,
\[
	\Pr[\A_2(\bqrys, \A_1(\bdb)) = \bdb[\bqrys]] \le \frac{1}{2^{n-t}} \cdot \sum\limits_{i = 0}^{ \min(1 + \mathsf{rd}_n(2^{n-t}), n)} {n \choose i} \cdot \left(1 - \frac{i}{n}\right)^k
\]
where $\bqrys \getsr [n]^k$ are the $k$ subkey locations, $\bdb \getsr \bits^n$ is the $n$-bit large key, adversary $\A_1$ outputs at most $t$ bits of leakage and $\A_2$ outputs $k$ bits.
\end{lemma}

\subsection{Blackbox constructions}
\label{sec:blackbox}
In our work, we will consider blackbox constructions. We borrow the following definitions directly from the
work of Reingold, Trevisan and Vadhan~\cite{TCC:ReiTreVad04}.

\begin{definition}[Blackbox Constructions]
Let $\mathcal{P}$ and $\mathcal{Q}$ be two primitives. A blackbox construction of primitive $\mathcal{Q}$
from primitive $\mathcal{P}$ consists of two PPT algorithms $(Q,S)$ satisfying the following:
\begin{itemize}[noitemsep]
\item For oracle $P$ that implements primitive $\mathcal{P}$, the algorithm $Q^P$ implements primitive $\mathcal{Q}$.
\item For any oracle $P$ implementing primitive $\mathcal{P}$ along with any computationally unbounded oracle
adversary $A$ successfully breaking the security of $Q^P$, then the oracle adversary $S^{P,A}$ breaks the security of $P$.
\end{itemize}
\end{definition}

Next, we will say that a primitive exists relative to an oracle using the following definition.

\begin{definition}
A primitive $\mathcal{P}$ exists relative to an idealized oracle $\mathcal{O}$ if there exists
an implementation $P$ for primitive $\mathcal{P}$ that is computable by a PPT algorithm
with access to idealized oracle $\mathcal{O}$. Furthermore, there exists no (computationally unbounded) adversary that makes at most polynomial number of queries to idealized oracle $\mathcal{O}$
that breaks the security of the implementation $P$.
\end{definition}

Finally, we will say that a blackbox construction relativizes using the following result from~\cite{TCC:ReiTreVad04}.

\begin{theorem}[Relativization]
If there is a blackbox construction of primitive $\mathcal{Q}$ from primitive $\mathcal{P}$, then the following two statements holds:
\begin{itemize}[noitemsep]
\item If primitive $\mathcal{P}$ exists relative to idealized oracle $\mathcal{O}$, then primitive $\mathcal{Q}$ also exists relative to idealized oracle $\mathcal{O}$.
\item If there exists a blackbox construction of primitive $\mathcal{P}$ from idealized oracle $\mathcal{O}$,
then there also exists a blackbox construction of primitive $\mathcal{Q}$ from idealized oracle $\mathcal{O}$.
\end{itemize}
\end{theorem}

Note, we can use the above to essentially show the following. If we prove that any cryptographic primitive cannot exist relative a crypto oracle $\Or$, then we prove that there cannot exist any blackbox construction
of the cryptographic primitive using any primitive that can be built from a crypto oracle $\Or$
(including random oracles, generic multilinear groups and virtual blackbox obfuscation as shown
in \thref{thm:oraclecrypto} taken from~\cite{EC:LinMooWic25}).

In the above definitions, we consider fully blackbox constructions from~\cite{TCC:ReiTreVad04}. One can also consider semi-blackbox constructions where the security reduction is not necessarily blackbox and may make use of the internal structure of the adversary. In this paper, we consider fully blackbox constructions for simplicity, but our results may be extended to rule out semi-blackbox constructions
using the same techniques explained in~\cite{EC:LinMooWic25}. In particular, most cryptographic
primitives such as encryption, obfuscation, and oblivious transfer enable the embedding property
introduced in~\cite{TCC:ReiTreVad04} which extends our results to semi-blackbox constructions for those primitives.

\section{Computation Lower Bounds}\label{sec:bb-sep}

We present our main result: an amortized computation lower bound of $\Omega(n/s)$
for any single-server PIR with client preprocessing with $s$-bit client storage supporting
$k = \Theta(s)$ queries.

\begin{theorem}\label{thm:main-lower-bound}
Relative to any crypto oracle $\Or$,\footnote{In \apref{sec:ap-oracles},
we explain how our results extend to \emph{any} ``consistently sampleable''
oracles, of which crypto oracles are a subset.} there does not exist a single-server
PIR with preprocessing scheme with an oracle-free encoding
(\defref{def:oracle-free}) that is $\varepsilon$-correct
and $\delta$-private
with $s$-bit client storage and total computation of
$t = o(n)$ %
supporting $k = \Omega(s)$ queries with $k \le 0.19n$
against computationally unbounded adversaries making $\poly(\lambda)$ queries to $\Or$
for any choice of $s = \omega(\log(n + \lambda))$, $\varepsilon = \Omega(2^{-s})$ and $\delta \le \varepsilon / 3$.
In more detail, there exists no such scheme satisfying both of the following efficiency measures:
\begin{itemize}[noitemsep]
\item Total communication of $o(n)$ bits.
\item Total of $o(n)$ input bits to the crypto oracle $\Or$ by the server during the online query protocol.
\end{itemize}
\end{theorem}

Our above result unconditionally rules out the existence of any such efficient single-server
PIR with preprocessing schemes with oracle-free encodings that make blackbox usage of any cryptography
that may be implemented using a crypto oracle (see Section~\ref{sec:blackbox} for more details on blackbox constructions). This includes random oracles (one-way functions),
virtual blackbox obfuscation, fully homomorphic encryption and idealized models such as generic multilinear groups (see Section~\ref{sec:prelim:oracles}).
As discussed in Section~\ref{sec:pir-def}, the oracle-free encoding
restriction is satisfied by every existing single-server PIR with client
preprocessing construction that we are aware of, most of which do not
encode the database at all along with the state-of-the-art public-key DEPIR.

We note our lower bound only applies to schemes that support sufficiently large $k = \Omega(s)$ queries. This turns out to be a necessary
restriction as there exist prior constructions that support a
smaller number of queries with $o(n/s)$ communication per
query (see Appendix~\ref{ap:small-queries}).

Finally, we remark on the wide applicability of our lower bound towards the parameterizations of $\varepsilon$ and $\delta$. In particular,
our lower bound rules out any scheme that is correct with exponentially small probability in $s$.
Note, the requirement that $\delta \le \varepsilon/3$ is essentially required as if $\delta > \varepsilon$ there are trivial algorithms
that essentially do plaintext retrievals with probability $\delta$ (sacrificing privacy), but obtaining sufficient correctness.
For more larger choices of correctness guarantees such as $\varepsilon = O(1)$ or $\varepsilon = 1/\poly(\lambda, n)$, we note
our lower bound applies to PIR with preprocessing schemes with weaker privacy guarantees with non-negligible choices of $\delta$.

\heading{Proof Overview.}
Our proof proceeds in three steps as follows. 
Each of the following sections is devoted to one step of the proof.

\begin{enumerate}
\item First, we present and define the new notion of dual PIR where the roles of
the online queries and offline preprocessing in PIR with preprocessing are essentially switched.
\item Next, we present a reduction that any PIR with preprocessing that is
secure relative to a crypto oracle $\Or$ implies that there is a dual PIR
relative to that same crypto oracle $\Or$ such that the dual PIR's efficiency guarantees directly
relate to the client storage size, communication and server cryptographic operations
in the original PIR with preprocessing scheme.
\item Finally, we prove information-theoretically (and therefore relative to any $\Or$) that no dual PIR can be both correct and efficient, by reducing to the subkey prediction problem from leakage-resilient cryptography.
\end{enumerate}

\subsection{Dual Private Information Retrieval}
We introduce a notion that we call \emph{dual} private information retrieval. 
At a high level, we reverse the order of the offline preprocessing and online querying 
that appears in the PIR with preprocessing setting.
Indeed, 
for dual PIR, there is first some $t$-bit {\em advice} that must be generated
without any knowledge (independently) of future queries by the client.
Afterwards, the client is given a batch of $k$ query indices $\qrys \in [n]^k$ to retrieve from the database $\db \in \{0,1\}^n$.
To answer their queries, the client may additionally
request an arbitrary $s$-bit {\em hint} about the database $\db$.
After receiving this information though, the client must be able to correctly retrieve the queried indices, 
$\db[\qrys] \in \{0,1\}^k$, using the $t$-bit advice and the $s$-bit hint.

The $t$-bit advice %
produced by $\advice$ plays the role of the transcript of the online queries, whereas 
the $s$-bit hint computed using $\hint$ 
is equivalent to the $s$-bit client storage from offline preprocessing in standard PIR with preprocessing.
Furthermore, the efficiency guarantees have been swapped in dual PIR where the client obtains $t$ bits in the first phase and $s$ bits in the second phase.
In contrast, a PIR with preprocessing scheme obtains an $s$ bits of client storage in the offline phase and $t$ bits of communication
during online queries. For this reason we shall sometimes call the advice {\em the transcript} and say that the
hint is the {\em secret}.

The definition of dual PIR that we present does not explicitly provide any privacy guarantees.
It turns out that privacy is essentially baked into the syntax of the algorithms. In the PIR with preprocessing model,
the privacy guarantees required by the query transcript viewed by the adversary during the online phase
should not reveal information about the queried indices $\qrys$.
In dual PIR, this privacy is essentially guaranteed since the client must generate the $t$-bit transcript
without any knowledge of the future query indices. Similarly, we allow the client to compute
$\hint$ outside of the adversary's view similar to our definition of $\setup$ in PIR with preprocessing (see Definition~\ref{def:pir}).

We formally present the definition of dual PIR below:

\begin{definition}[Dual private information retrieval]\label{def:dpir}
	A \emph{dual private information retrieval} for database size $n$, batch size $k$, 
	advice size $t$, and %
	hint size $s$, shortened to a $(n,k,t,s)$-dual PIR, is a tuple of algorithms
	$\dpir=(\advice,\allowbreak\hint,\allowbreak\recon)$ which satisfies the following syntax:
	\begin{itemize}
		\item $\advice(\db) \to \adviceout$, takes as input a database $\db\in \bits^n$ and outputs
		a string of advice $\adviceout \in \bits^t$.
		\item $\hint(\db,\qrys,\adviceout) \to \hintout$, takes as input a database $\db \in \bits^n$, a
		batch of queries $\qrys\in [n]^k$ as well as advice $\adviceout \in \bits^t$  and outputs a
		hint string $\hintout \in \bits^s$.
		\item $\recon(\qrys, \adviceout, \hintout) \to \ans$, takes as input queries $\qrys\in [n]^k$,
		advice $\adviceout \in \bits^t$ as well as hint string $\hintout \in \bits^s$ and outputs
		the answer $\ans \in \bits^k$.
	\end{itemize}
	We say that a $(n,k,t,s)$-\emph{dual PIR} is $\varepsilon$-correct if
	\[
	\Pr
	\left[\recon(\bqrys,\badviceout,\bhintout) = \bdb[\bqrys]
	\right]
		\ge \varepsilon,
	\]
	where $\bdb \getsr \bits^n$, $\bqrys \getsr [n]^k$ are uniform and independent random variables, and $\badviceout \getsr \advice(\bdb)$
	and $\bhintout \getsr \hint(\bdb,\bqrys,\badviceout)$ are the outputs of the respective algorithms on the random inputs.
\end{definition}

First, we note that dual PIR schemes are very easy in the case that $s \ge k$. In this setting, the client can simply ignore the
$t$-bit transcript from the
$\advice$ algorithm and simply privately retrieve the $k$ queried entries, $\db[\qrys]$ in $\hint$. Therefore, we focus
on the setting of $s < k$ where this construction is not possible.

One may wonder the formal relations between dual PIR and PIR with preprocessing. In a sense, dual PIR provides
the client both more and less information and flexibility at different steps of the protocol compared to PIR with preprocessing.
Nevertheless, we show that any PIR with preprocessing may be converted into a dual PIR scheme.

\subsection{Constructing a Dual PIR}
\label{sec:reduce-dual-pir}

This section proves that any PIR with preprocessing scheme $\Pi^{\Or} = (\init, \setup^{\Or}, \query^{\Or})$
with an oracle-free encoding\footnote{We discuss
plausible relaxations of this restriction in Remark~\ref{rem:init-ext}.}
(see \defref{def:oracle-free})
that has access to a crypto oracle $\Or$
that is sufficiently private and correct can be used to build
a correct dual PIR $\Pi' = (\advice^{\Or,\RO}, \hint^{\Or,\RO}, \recon^{\Or,\RO})$ with  access to both a crypto oracle $\Or$ as well as an additional independent random oracle $\RO$.
We make the following assumptions about $\Pi^{\Or}$ without loss of generality. First,
we will consider underlying PIR schemes where the client storage is
always exactly $s$ bits (one can always pad to make sure the client storage is $s$ bits).
Both $\init$ and $\setup^{\Or}$ will be randomized.\footnote{We note that randomized $\init$ enables us to capture schemes with randomized (but oracle-free) encodings, as well as the derandomization of the server described next.}
For $\query^{\Or}$, we will suppose
that the server's execution during $\query^{\Or}$ is a deterministic function of its input (including the crypto oracle $\Or$, encoded database $\edb$ and communication transcript $\tx$).
We note that assuming the server execution of $\query^{\Or}$ is deterministic is without loss of generality
as $\init$ is randomized and may attach the necessary server randomness
to respond to the $k$ queries into the encoded
database $\edb$ to be used during $\query^{\Or}$ by the server.
However, the client may use randomness in its execution of $\query^{\Or}$.

In our dual PIR construction, we will consider executing the client-side %
of the $\query$ algorithm of our underlying PIR scheme $\Pi$.
As a reminder,
the $\query$ algorithm is a two-party, interactive protocol (potentially consisting of multiple rounds) between
the client with input client state $\st$ and query sequence $\qrys$,
and a server with the encoded database $\edb$ 
(see~\defref{def:pir} for more details). We denote
the execution of this algorithm by $(\ans; \bot \mid \tx) \leftarrow \query^{\Or}(\st, \qrys; \edb)$.
The output consists of a private output $\ans$ to the client and a shared output $\tx$ consisting of the communication transcript
between both parties ($\bot$ denotes that there is no private output to the server). We will also consider the transcript $\tx$ into two components consisting of communication sent by the client $\tx_{\client}$ and the communication sent by the server $\tx_{\server}$.
We denote the client-side execution of the query algorithm by the $(\ans \mid \tx_\client) \leftarrow \query^{\Or}_{\client}(\st, \qrys; \tx_\server)$ where we
simulate the server-side computation by using its communication transcript denoted by $\tx_\server$. The output is both
the client's output $\ans$ along with the client's communication transcript $\tx_\client$.

\medskip\noindent{\bf Overview of Dual PIR Construction.} At a high level, $\advice^{\Or,\RO}(\db)$ will sample a communication transcript from the proper execution of underlying PIR construction. However, we note that $\advice^{\Or,\RO}(\db)$ only receives a database $\db$ and, in particular, does not receive the query $\qrys$. Instead, $\advice_{\qrys'}^{\Or,\RO}(\db)$ is parameterized by some fixed query sequence $\qrys' \in [n]^k$ and executes $\Pi$ using the fixed query sequence $\qrys'$ to obtain a communication transcript $\tx$.
Additionally, we will record all crypto oracle queries performed by the server in the set $S$. The final output of $\advice$ will be $\adviceout = (\tx, S)$ consisting of the sampled communication transcript $\tx$ and all crypto oracle queries performed by the server $S$
during the $\query$ algorithm.

In the next subroutine of dual PIR, $\hint^{\Or, \RO}(\db, \qrys, \adviceout = (\tx, S))$ receives the database $\db$, the real query sequence of interest $\qrys$ and the output of $\advice$ denoted by $\adviceout = (\tx, S)$.
As a note, $\hint$ accepts any query sequences $\qrys$ including those that are different from the fixed query sequence used in $\advice$ such that
$\qrys \ne \qrys'$.
With this input, $\hint^{\Or, \RO}(\db,\qrys,\adviceout = (\tx, S))$ will essentially re-randomize the client's execution of $\Pi$ using input query sequence $\qrys$ while fixing the server's
execution. In more detail, we will iterate through all possible choices of randomness (including randomness in $\setup$ and the choice of crypto oracle) to find some client state $\st$ such that the client's execution matches the communication transcript $\tx$.
To fix the server's execution, we sample crypto oracles $\Or_1,\ldots,\Or_z$ such that each sampled crypto oracle matches $\Or$ for all queries in $S$ for some choice of $z$ we will fix later.
For sampling each crypto oracle $\Or_i$ that is consistent with $\Or(S)$, we use
the algorithm described in Lemma~\ref{lem:oracle_resample} with the shared random oracle $\RO$.
As we assume without loss of generality that the server's algorithm is deterministic, fixing all its server crypto oracle results means that the server's execution is fixed now.
Note, $\hint$ may obtain the same encoded database $\edb$ by executing
$\init(\db)$ using the same randomness $\omega_\server = \RO(0)$ as used in
$\advice$; as $\init$ makes no oracle queries, this recomputed $\edb$ is
identical to the one in $\advice$ and, importantly, does not depend on the
crypto oracle at all.
Next, $\hint$ will attempt to sample a client state using the $z$ different crypto oracles.
For all $i \in \{1,\ldots,z\}$,
$\hint$ samples a client state $(\st_i; \bot) \leftarrow \setup^{\Or_i}(1^\lambda; \edb)$ and checks whether the execution of the query algorithm $(\ans; \bot \mid \tx_i') \leftarrow \query^{\Or_i}(\st, \qrys; \edb)$ produces the same communication transcript $\tx = \tx_i'$.
For the first such instance with index $i$, $\hint$
will encode the sampled client state $\st_i$ and crypto oracle $\Or_i$ in its output $\hintout$. We will revisit
this encoding in more detail later as encoding $\Or_i$ is clearly not efficient, but we will assume this is possible for now.

Finally, $\recon^{\Or,\RO}(\qrys, \adviceout, \hintout)$ executes the query algorithm
of the underlying PIR protocol using the encoded communication transcript $\tx$ from $\advice$ as well as the
client state $\st$ and crypto oracle $\Or_i$ from $\hint$. As a reminder, $\recon$ does not receive the database
$\db$ as input. Instead, it executes only the client-side of the query algorithm $(\ans \mid \tx_\client) \leftarrow \query^{\Or_i}_\client(\st, \qrys; \tx_\server)$
where the server's interaction is replicated using the server's communication in the transcript $\tx_\server$.
The output of $\recon$ will be the answer $\ans$ obtained by executing the client-side of the query algorithm.

\medskip\noindent{\bf Encoding Algorithm for $\hint$.} 
Now, we revisit the issue of encoding a crypto oracle $\Or_i$ for $\hint$.
While the client state $\st_i$ is $s$ bits,
the crypto oracle $\Or_i$ is much larger and cannot be encoded efficiently.
As a reminder, our high level goal is to show that there exists a 
dual PIR where the output of $\hint$ is $O(s)$ bits.
A first attempt might be for $\hint$ to simply encode the sampled client state $\st$. 
As the crypto oracles
$\Or_1,\ldots,\Or_z$ are sampled using a shared random oracle $\RO$, $\recon$ can also sample them without any additional information.
Afterwards, $\recon$ could attempt to execute the query algorithm using $\st$ for each of the crypto oracles and find the first
crypto oracle $\Or_j$ that results in the communication transcript $\tx$.
Unfortunately, there is no guarantee that $\recon$
ends up at the same oracle as $\hint$ such that $j = i$.
In particular, $\recon$ does not receive the database. Therefore, $\recon$ cannot determine
whether $\st_i$ is a possible output of $\setup^{\Or_j}(1^\lambda; \edb)$.
For example, it is a possible $\st_i$ somehow results in a communication matching $\tx$ for some
index $j < i$ where $\st_i$ cannot be output by $\setup^{\Or_j}(1^\lambda; \edb)$.
Another idea might be to instead encode the index $i$ directly. Unfortunately, it turns out that the index $i$ might need to be quite large
since $\hint$ might need to use $z = 2^{\Omega(|\tx|)}$ iterations to find a sampled client state and crypto oracle 
that results in the communication transcript $\tx$ output by $\advice$. 
This would be necessary if the distribution of the transcripts is uniformly random over all $|\tx|$-bit strings.

To resolve this, we use a combination of the above two ideas with additional changes to fix various issues that arise
in both approaches. To ensure that both $\hint$ and $\recon$ use the same client state in each execution, we will instead
sample client state and randomness $(\st_i, \omega_{\client,i}) \leftarrow \RO_0(i)$ to be used with crypto oracle $\Or_i$ using the shared random oracle $\RO$. We also modify
$\hint$ such that it will also honestly generate a client state $(\st_i'; \bot) \leftarrow \setup^{\Or_i}(1^\lambda; \edb)$
where $\edb$ is the encoded database from $\edb \leftarrow \init(1^\secpar, \db)$ using randomness $\omega_{\server} = \RO(0)$.
If $\st_i \ne \st_i'$, then $\hint$ will essentially skip past using the $i$-th client state and crypto oracle.
However, $\hint$ will still execute the query algorithm using the random client state $\st_i$, client randomness $\omega_{\client,i}$, and crypto oracle $\Or_i$
to obtain the communication transcript $(\ans; \bot \mid \tx_i) \leftarrow \query^{\Or_i}(\st_i, \qrys; \edb)$ using $\omega_{\client,i}$ as the client-side randomness.
$\hint$ will keep a counter $c$ of the number of times that $\tx_i = \tx$ matching the transcript output by $\advice$.
With this change, $\hint$ will find the
first index $i \in [z]$ such that the randomly sampled client state and randomness $(\st_i, \omega_{\client,i})$ using the random oracle $\RO$ matches
the honestly generated client state $\st_i' \leftarrow \setup^{\Or_i}(1^\lambda; \edb)$ and the resulting communication transcript matches
$\tx$. 
Instead of encoding the index $i$, 
$\hint$ will instead encode the counter $c$ keeping track of the number of times 
it has observed transcript matches (that is, $\query$ outputs the same communication transcript as $\tx$ received from $\advice$)
before
finding an iteration where the observed transcript matches and $\setup$ gives the same client state as the one sampled
by using $\RO$. 
We will later show that this counter $c$ only requires $O(s)$ bits except with negligible probability.
It turns out that this counter is sufficient for $\recon$ to reconstruct both the sampled client state $\st_i$ and
crypto oracle $\Or_i$ without needing the database $\db$. In particular, $\recon$ executes the client-side of
the query algorithm for each of the sampled client states and crypto oracles and observes whether the output
communication transcript matches $\tx$ output by $\advice$. On the $c$-th matching transcript occurrence, $\recon$ will 
successfully decode the same client state $\st_i$ and crypto oracle $\Or_i$ that $\hint$ needed to encode.

\medskip\noindent{\bf Our Dual PIR Construction.}
We formally present the Dual PIR below.
In the description below, $Q'$ is a fixed set of $k$ queries such as $Q'=(1,\ldots,k)$
whereas $Q$ is an input. For convenience, we will abuse notation and consider the outputs
of the random oracle $\RO$ to be variable length depending on the desired output.
For random oracles $\RO$ with fixed output lengths, one can always add unique prefixes/suffixes
to derive arbitrary amounts of randomness.
Similarly, we assume the inputs to $\RO$ are domain separated: the input
$0$ used to derive $\omega_\server = \RO(0)$ is encoded distinctly from every
input of the form $i \mid\mid x$ used by the virtual oracles
$\RO_0,\ldots,\RO_z$ (e.g., via a dedicated prefix), so the randomness of
$\omega_\server$ is independent of all other randomness derived from $\RO$.

\begin{enumerate}
\item $\advice_{\qrys'}^{\Or, \RO}(\db)$: %
    \begin{itemize}
    	\item Run $\edb \gets \init(1^{\secpar}, \db)$ using
		$\omega_{\server} = \RO(0)$ as the algorithm's randomness.
        \item Run $(\st;\bot)\gets\setup^{\Or}(1^\lambda ;\edb)$.
        \item Run $(\ans;\bot\mid \tx)\gets\query^{\Or}(\st,\qrys';\edb)$.
        \item Record all server-side oracle $\Or$ queries executed during the $\query$ algorithm as $S$.
        \item Output $\adviceout \leftarrow (\tx, S)$.
    \end{itemize}
\item
$\hint^{\Or, \RO}(\db, \qrys, \adviceout)$:
    \begin{itemize}
    	\item Parse $\adviceout = (\tx, S)$.
    	\item Compute $\Or(S) \leftarrow \{(q,\Or(q)) \mid q \in S\}$.
	\item Run $\edb \leftarrow \init(1^{\secpar}, \db)$ using
	$\omega_{\server} = \RO(0)$ as the algorithm's randomness.
	\item Set $z \leftarrow 2^{3n+2s}$.
    	\item Construct virtual random oracles $\RO_0, \RO_1,\ldots,\RO_z$ such that $\RO_i(\cdot) = \RO(i \mid\mid \cdot)$ and $i \in \{0,1,\ldots,z\}$ is represented using a string of $O(n+s)$ bit length.
	\item Instantiate counter $c \leftarrow 0$.
	\item For $i \in \{1,\ldots,z\}$ and while $c < 2^{3s}$:
	\begin{enumerate}
		\item Sample oracle $\Or_i \leftarrow \mathcal{A}^{\RO_i}(\Or(S))$ such that $\Or_i(S) = \Or(S)$ using Lemma~\ref{lem:oracle_resample} with parameter $\zeta = 6n$. If the algorithm fails to sample crypto oracle $\Or_i$, then output $\hintout \leftarrow \bot$.
		\item Compute $(\st_i, \omega_{\client,i}) \leftarrow \RO_0(i)$.
		\item Compute $(\st_i'; \bot) \leftarrow \setup^{\Or_i}(1^\lambda; \edb)$.
		\item Compute $(\ans; \bot \mid \tx_i) \leftarrow \query^{\Or_i}(\st_i, \qrys; \edb)$ using $\omega_{\client,i}$ as the randomness for the client side of the $\query$ protocol.
		\item If $\tx_i = \tx$:
		\begin{enumerate}
			\item If $\st_i = \st_i'$, output $\hintout \leftarrow c$.
			\item Increment $c \leftarrow c + 1$.
		\end{enumerate}
	\end{enumerate}
	\item Output $\hintout \leftarrow \bot$.
    \end{itemize}
\item
$\recon^{\Or, \RO}(\qrys,\adviceout,\hintout)$:
    \begin{itemize}
	\item Parse $\adviceout = (\tx, S)$.
	\item Parse $\hintout = c$.
    	\item Parse $\tx = (\tx_\client, \tx_\server)$ as client and server communication respectively.
	\item Compute $\Or(S) \leftarrow \{(q, \Or(q)) \mid q \in S\}$.
	\item Set $z \leftarrow 2^{3n+2s}$.
    	\item Construct virtual random oracles $\RO_0, \RO_1,\ldots,\RO_z$ such that $\RO_i(\cdot) = \RO(i \mid\mid \cdot)$ and $i \in \{0,1,\ldots,z\}$ is represented using string of $O(n+s)$ bits.
	\item Instantiate counter $c' \leftarrow 0$.
	\item For $i \in \{1,\ldots,z\}$:
	\begin{enumerate}
		\item Sample oracle $\Or_i \leftarrow \mathcal{A}^{\RO_i}(\Or(S))$ such that $\Or_i(S) = \Or(S)$ using Lemma~\ref{lem:oracle_resample} with parameter $\zeta = 6n$. If the algorithm fails to sample crypto oracle $\Or_i$, then output $\reconout \leftarrow \bot$.
		\item Compute $(\st, \omega_{\client}) \leftarrow \RO_0(i)$.
		\item Compute $(\ans \mid \tx_\client') \leftarrow \query^{\Or_i}_\client(\st, \qrys; \tx_\server)$ using $\omega_{\client}$ as the randomness for the client side of the $\query$ protocol.
		\item If $\tx_\client' = \tx_\client$:
		\begin{enumerate}
			\item If $c' = c$, output $\reconout \leftarrow \ans$.
			\item Increment $c' \leftarrow c' + 1$.
		\end{enumerate}
	\end{enumerate}
	\item Output $\reconout \leftarrow \bot$.
    \end{itemize}
\end{enumerate}

\subsection{Correctness Analysis of Dual PIR Construction}

This section will be devoted to proving the correctness of our dual PIR construction. In particular, we will prove the following theorem:

\begin{theorem}
\label{thm:dualpirmain}
Suppose the PIR construction $\Pi = (\init, \setup^{\Or}, \query^{\Or})$ has an oracle-free encoding (\defref{def:oracle-free}), supports batch queries of size $k$
for databases of size $n$ with client state of $s$ bits, the communication transcript and all server crypto oracle queries during the query algorithm are $t$ bits in total.
Furthermore, suppose $\Pi$ is $\varepsilon$-correct and $\delta$-private.
For any fixed query sequence $\qrys'$, our construction
is a $(n, k, t, 3s)$-dual PIR with $\varepsilon'$-correctness satisfying $\varepsilon' \ge \varepsilon - 2\delta - 2^{-s+1} - 2^{-n+1}$.
\end{theorem}

To prove the correctness of our dual PIR construction, we will use the following steps:
\begin{enumerate}
\item First, we analyze our dual PIR construction when we assume that the input query sequence $\qrys$ to both $\hint$ and $\recon$ is the same as the fixed query sequence $\qrys'$ used in $\advice$. To do this, we prove the following properties
about our dual PIR construction when $\qrys = \qrys'$:
\begin{enumerate}
\item $\hint$ successfully samples a client state $\st_i$ and crypto oracle $\Or_i$ resulting in a communication transcript matching the output of $\advice$ (and, thus $\hint$ outputs a counter $c \ne \bot$) except with probability exponentially small in $s$.
\item If $\hint$ outputs a valid counter $c \ne\bot$, then the resulting output distribution of $\recon$ is equivalent to a proper
execution of the underlying PIR on input database $\db$, query sequence $\qrys'$ and a properly sampled crypto oracle $\Or$.
As a result, $\recon$ successfully obtains the correct answer $\db[\qrys']$ with probability similar to the correctness probability of the
underlying PIR.
\end{enumerate}
\item
As the last step, we prove that our dual PIR must be correct for any query sequence $\qrys \ne \qrys'$. If this is not the case, we show there exists an adversary that breaks the underlying PIR privacy using the same number of oracle queries as an honest server.

\end{enumerate}

\medskip\noindent{\bf Analyzing Output of $\hint$.}
To start, we will prove several properties about the output of $\hint$.
First, we denoted the output of $\advice^{\Or,\RO}_{\qrys'}(D)$ by $\adviceout = (\tx, S)$.
Throughout this section, we fix the encoding randomness $\omega_\server
= \RO(0)$ and, thus, the encoded database $\edb = \init(1^\secpar, \db)$
computed using randomness $\omega_\server$. As $\init$ makes no oracle
queries, $\edb$ is a fixed value that is identical in $\advice$, $\hint$, and
the honest PIR execution, and is independent of the crypto oracle $\Or$. All
probability statements below hold for any fixed choice of $\omega_\server$
(equivalently, of $\edb$).

Consider the following \emph{resampled experiment} associated with the
iterations of $\hint$: first sample a crypto oracle $\Or_i$ that is
consistent with $\Or(S)$, then a client state
$\st' \leftarrow \setup^{\Or_i}(1^\lambda; \edb)$, and execute
$\query^{\Or_i}(\st', \qrys'; \edb)$ with fresh client randomness to obtain
a communication transcript $\tx_i$; finally, check whether $\tx_i = \tx$.
(The iterations of $\hint$ instead run $\query$ with a \emph{uniformly}
sampled state; as we show in Lemma~\ref{lem:hint-counter}, the experiment
above with an honestly sampled state characterizes both the distribution of
the advice and the acceptance rate of $\hint$.) For a set of query-answer pairs
$v = \{(\qry, v(\qry)) \mid \qry \in S\}$, we denote the probability of this
event by
$$p(\tx, S, v, \edb) = \Pr[\btx = \tx]$$
over the random choice of $\Or_i$ conditioned on $\Or_i(S) = v$ and the
internal randomness of $\setup$ and $\query$
(noting that the client-side randomness $\omega_{\client}$ for $\query$ is derived explicitly from $\RO_0$ in $\hint$, but $p$ averages over all such choices),
where $\btx \leftarrow \query^{\Or_i}(\bst, \qrys'; \edb)$ and $\bst \leftarrow \setup^{\Or_i}(1^\lambda; \edb)$.

Before analyzing $\hint$, we record two consequences of the fact that
the fixed encoded database $\edb$ is independent of the crypto oracle. The
first states that, once an oracle agrees with $\Or$ on $S$, the entire
server-side execution of $\query$ is replayed exactly. The second states that
the transcript distribution of an honest execution, conditioned on the
server's oracle queries and answers, is exactly the transcript distribution
of the resampled executions performed by $\hint$. The second part is
precisely where the oracle-free encoding assumption is used: if $\init$ could
query $\Or$, then $\edb$ would be correlated with $\Or$ outside of $S$, and
an honestly generated pair $(\edb, \Or)$ would not be distributed as the pair
$(\edb, \Or_i)$ used in the resampled executions of $\hint$.

\begin{clm}[Server Replay and Advice Consistency]\label{clm:replay}
Fix any database $\db$, query sequence $\qrys'$, encoding randomness
$\omega_\server$ (and thus $\edb$), transcript $\tx = (\tx_\client,
\tx_\server)$, query set $S$, and answers $v$ such that $(\tx, S)$ is a
possible output of $\advice^{\Or,\RO}_{\qrys'}(\db)$ when $\Or(S) = v$.
Then:
\begin{enumerate}
\item \emph{(Server replay.)} For every crypto oracle $\Or^*$ such that
$\Or^*(S) = v$, every client state $\st$, every client randomness
$\omega_\client$, and every query sequence $\qrys$: in the execution
$\query^{\Or^*}(\st, \qrys; \edb)$, for every round $r$, if the client's
first $r$ messages equal the first $r$ messages of $\tx_\client$, then the
server's oracle queries, answers, and messages through round $r$ equal the
corresponding prefix of the execution recorded by $\advice$ (in particular,
its queries lie in $S$ with answers given by $v$, and its messages equal
the corresponding prefix of $\tx_\server$).
Consequently, if the client sends all of $\tx_\client$, then the server
makes exactly the oracle queries $S$ (in the same order), receives the
answers $v$, and sends the messages $\tx_\server$; the execution produces
the full transcript $\tx$ if and only if
the client-side messages equal $\tx_\client$; and any execution under
$\Or^*$ with transcript $\tx$ has server query set exactly $S$.
\item \emph{(Advice consistency.)} For the honest execution inside
$\advice^{\Or,\RO}_{\qrys'}(\db)$ over a uniformly sampled crypto oracle
$\Or$,
$$
\Pr\left[(\btx, \bS) = (\tx, S) \wedge \Or(S) = v\right] =
\Pr\left[\Or(S) = v\right] \cdot p(\tx, S, v, \edb).
$$
\end{enumerate}
\end{clm}
\begin{proof}
For (1), recall that the server's execution during $\query$ is a
deterministic function of the encoded database $\edb$, the incoming client
messages, and its oracle answers. The server's execution recorded by
$\advice$ (with an oracle answering $v$ on $S$, encoded database $\edb$, and
incoming client messages $\tx_\client$) made the queries $S$ in some order,
received the answers $v$, and produced the messages $\tx_\server$. Consider
any execution under $\Or^*$ with the same encoded database $\edb$ (recall
that $\edb$ does not depend on the choice of oracle since $\init$ makes no
oracle queries) in which the client sends the messages $\tx_\client$. The
server's first oracle query is determined by $\edb$ and a prefix of
$\tx_\client$, so it is the same first query $\qry_1 \in S$, and the answer
$\Or^*(\qry_1) = v(\qry_1)$ also matches. By induction, every subsequent
query, answer, and outgoing message of the server matches the execution
recorded by $\advice$, proving the first statement. The remaining statements
follow immediately: if the client sends $\tx_\client$, the server sends
$\tx_\server$, so the transcript is $\tx$ if and only if the client-side
messages equal $\tx_\client$, in which case the query set is exactly $S$.

For (2), note that conditioned on the event $\Or(S) = v$, the uniformly
sampled crypto oracle $\Or$ is distributed exactly as the resampled oracles
$\Or_i$ (both are a uniformly random crypto oracle conditioned on answering
$v$ on $S$; see Lemma~\ref{lem:oracle_resample}). Moreover, the honest
execution inside $\advice$ touches the oracle only through $\setup$ and
$\query$ using the fixed $\edb$ (again, since $\init$ makes no oracle
queries), which is exactly the experiment defining $p(\tx, S, v, \edb)$,
except that the event of $p$ does not require the server query set to equal
$S$. By part (1), any execution under an oracle agreeing with $v$ on $S$
that produces transcript $\tx$ has server query set exactly $S$, so the two
events coincide and
$\Pr[(\btx,\bS) = (\tx,S) \mid \Or(S) = v] = p(\tx, S, v, \edb)$.
\end{proof}

We now show that $\advice$ will output a pair $\adviceout = (\tx, S)$ such that this probability $p$ is at least $2^{-3n}$ except with probability negligible in $n$.

\begin{lemma}
\label{lem:p}
Suppose the underlying PIR protocol has $o(n)$ bits of communication and the server performs blackbox crypto operations
on at most $o(n)$ bits. For any choice of database $D$ and query sequence $Q'$,
let $\mathbf{p} = p(\btx, \bS, \Or(\bS), \edb)$ be the above probability where $(\btx, \bS) \leftarrow \advice^{\Or,\RO}_{\qrys'}(\db)$.
Then,
$$
\Pr[\mathbf{p} \ge 2^{-3n}] \ge 1 - 2^{-n}.
$$
over the internal randomness of $\advice$ and random choice of $\Or$.
\end{lemma}
\begin{proof}
Consider the set $Z$ consisting of all triples $(\tx, S, v)$ such that
$p(\tx, S, v, \edb) < 2^{-3n}$. Using part (2) of Claim~\ref{clm:replay},
\begin{align*}
\Pr[\mathbf{p} < 2^{-3n}]
&= \sum_{(\tx, S, v) \in Z} \Pr[(\btx, \bS) = (\tx, S) \wedge \Or(S) = v]
\le \sum_{(\tx, S, v) \in Z} \Pr[\Or(S) = v] \cdot p(\tx, S, v, \edb)\\
&< 2^{-3n} \cdot \sum_{(\tx, S)} \sum_{v} \Pr[\Or(S) = v]
= 2^{-3n} \cdot \sum_{(\tx, S)} 1 \le 2^{-3n} \cdot 2^{o(n)} \le 2^{-n},
\end{align*}
where the first sum effectively ranges over the triples in $Z$ that arise
with positive probability as advice outcomes (to which part (2) of
Claim~\ref{clm:replay} applies; all other summands on the left vanish, and
enlarging the second sum to all of $Z$ only increases the bound). We also
use that $\sum_v \Pr[\Or(S) = v] = 1$ for every fixed query set $S$,
and that the total number of pairs $(\tx, S)$ is at most
$2^{|\tx| + |S|} = 2^{o(n)}$ since the underlying PIR protocol has
$|\tx| = o(n)$ communication and the server performs blackbox cryptography
operations on $|S| = o(n)$ bits.
\end{proof}

In other words, we have shown that by executing $1/p(\btx, \bS, \Or(\bS), \edb) \le 2^{3n}$ sampling experiments in $\hint$ will successfully sample a crypto oracle $\Or_i$ and client state $\st'_i$ that will result in this communication transcript.
Next, we show that $\hint$ outputs $\hintout = c$ such that $c \ne\bot$ except with negligible probability.
To do this, we essentially show the number of erroneous times that the communication transcript matches
the one output by $\advice$ is essentially equivalent to $1/p(\btx, \bS, \Or(\bS), \edb)$.

\begin{lemma}
\label{lem:hint-counter}
Suppose the underlying PIR protocol has $o(n)$ bits of communication and the server performs blackbox crypto operations
on at most $o(n)$ bits. For any choice of database $\db$ and query sequence $\qrys'$,
$$
\Pr[\hint^{\Or,\RO}(\db, \qrys', (\btx, \bS)) = \bot \mid (\btx, \bS) \leftarrow \advice^{\Or,\RO}_{\qrys'}(\db)] \le 2^{-s+1} + 2^{-n+1}
$$
over the random choices of $\Or$ and $\RO$ as well as the internal randomness of $\advice$.
\end{lemma}
\begin{proof}
We start by analyzing the probability that $\hint$ outputs $\bot$ due to the failure of sampling a crypto oracle.
As a reminder, $\hint$ samples $2^{3n + 2s}$ crypto oracles using the algorithm in Lemma~\ref{lem:oracle_resample}
with parameter $\zeta = 6n = \poly(\lambda)$ since $n = \poly(\lambda)$. Therefore, a single crypto oracle
sample fails with probability $e^{-6n}$. By a union bound, the probability that any crypto oracle sample fails is at
most $2^{3n+2s} \cdot e^{-6n} \le 2^{5n-6n} = 2^{-n}$.
For the remainder of the proof, we condition on all crypto oracle
samples succeeding; by Lemma~\ref{lem:oracle_resample}, each $\Or_i$ is then
an independent, uniformly random crypto oracle conditioned on
$\Or_i(S) = \Or(S)$.

Next, we identify the two remaining ways in which $\hint$ can output $\bot$.
We will consider the first
$z' = \lfloor 2^{2s}/p(\btx, \bS, \Or(\bS), \edb) \rfloor \le 2^{2s} \cdot 2^{3n} = z$
iterations of the loop, where $(\btx, \bS)$ is the output of
$\advice^{\Or,\RO}_{\qrys'}(\db)$ and we assume
$p(\btx, \bS, \Or(\bS), \edb) \ge 2^{-3n}$, adding the failure probability
$2^{-n}$ from Lemma~\ref{lem:p} at the end. Say that iteration $i$ is a
\emph{match} if $\tx_i = \tx$, and an \emph{acceptance} if additionally
$\st_i = \st_i'$. If some iteration $i \le z'$ is an acceptance and the
number of matches strictly before it is fewer than $2^{3s}$, then $\hint$
outputs a counter $c \ne \bot$. Therefore, denoting by $\mathbf{c}$ the
number of matches in the first $z'$ iterations,
$$
\Pr[\hint^{\Or,\RO}(\db, \qrys', (\btx, \bS)) = \bot] \le \Pr[\mathbf{c} \ge 2^{3s}] +
\Pr[\text{no acceptance in } z' \text{ iterations}]
+ 2^{-n} + 2^{-n}.
$$

We first bound $\Ex[\mathbf{c}]$. Fix an advice output and oracle answers
$(\tx, S, v)$ where $v = \Or(S)$. Each iteration $i$ independently samples
$\Or_i$ (uniform conditioned on $\Or_i(S) = v$) and a uniform pair
$(\st_i, \omega_{\client, i}) \leftarrow \RO_0(i)$, all independent of the
advice execution (recall that $\RO$ is independent of $\Or$ and, by domain
separation, the values $\RO_0(i)$ and the randomness of each virtual oracle
$\RO_i$ are independent of $\omega_\server$). Therefore, every iteration is
a match independently with the same probability
$$
q(\tx, S, v) := \Pr\left[\query^{\Or_i}(\st_i, \qrys'; \edb)
\text{ produces transcript } \tx\right]
$$
over the choice of $\Or_i$ conditioned on $\Or_i(S) = v$ and uniform
$(\st_i, \omega_{\client,i})$. Note that $q$ does not depend on $i$ and differs from $p$ only in
that the client state is uniformly random instead of honestly sampled using
$\setup^{\Or_i}$. Using part (2) of Claim~\ref{clm:replay} and linearity of
expectation,
\begin{align*}
    \Ex[\mathbf{c}] &\le \sum_{(\tx, S, v)}
    \Pr[(\btx, \bS) = (\tx, S) \wedge \Or(S) = v] \cdot
    \frac{2^{2s}}{p(\tx, S, v, \edb)} \cdot q(\tx, S, v)\\
    &= 2^{2s} \sum_{(\tx, S, v)} \Pr[\Or(S) = v] \cdot q(\tx, S, v)
    \le 2^{2s},
\end{align*}
where both sums effectively range over the triples $(\tx, S, v)$ arising
with positive probability as advice outcomes (all other summands in the
first line vanish), to which Claim~\ref{clm:replay} applies, and the final
inequality holds as follows. Note, sampling a uniformly random
crypto oracle conditioned on answering $v$ on $S$ and multiplying by
$\Pr[\Or(S) = v]$ is equivalent to sampling an unconditioned uniformly
random crypto oracle and restricting to the event $\Or_i(S) = v$. Combined with part (1) of
Claim~\ref{clm:replay} (any execution under an oracle agreeing with $v$ on
$S$ with transcript $\tx$ has server query set exactly $S$), we get that
$$
\Pr[\Or(S) = v] \cdot q(\tx, S, v) =
\Pr\left[\query^{\Or_i}(\st_i, \qrys'; \edb) \text{ produces } \tx
\text{ with query set } S \text{ and } \Or_i(S) = v\right]
$$
where the probability is now over an \emph{unconditioned} uniformly random
crypto oracle $\Or_i$ and uniform $(\st_i, \omega_{\client, i})$. A single
execution determines a unique triple (its transcript, its server query set,
and the oracle's answers on that query set), so these events are disjoint
across distinct triples $(\tx, S, v)$ and their probabilities sum to at most
$1$. By Markov's inequality, $\Pr[\mathbf{c} \ge 2^{3s}] \le 2^{-s}$.

Finally, we bound the probability that no acceptance occurs in the first
$z'$ iterations. In each iteration, the state $\st_i$ is uniform over
$\bits^s$ (recall client states are padded to exactly $s$ bits) and is
independent of the honestly sampled state
$\st_i' \leftarrow \setup^{\Or_i}(1^\lambda; \edb)$. Therefore, each
iteration is an acceptance independently with probability exactly
\begin{align*}
\Ex_{\Or_i}\left[\sum_{\st \in \bits^s} 2^{-s} \cdot
\Pr\left[\setup^{\Or_i}(1^\lambda; \edb) = \st\right] \cdot
\Pr\left[\query^{\Or_i}(\st, \qrys'; \edb) \text{ produces } \tx\right]\right]
= 2^{-s} \cdot p(\tx, S, v, \edb),
\end{align*}
since requiring the uniform state $\st_i$ to equal an independent honest
sample $\st_i'$ is equivalent to using the honest sample directly, at an
exact multiplicative $2^{-s}$ loss in probability. Over the
$z' = \lfloor 2^{2s}/p \rfloor$
independent iterations, the probability that no acceptance occurs is at most
$$
\left(1 - 2^{-s} \cdot p\right)^{z'} \le
\left(1 - 2^{-s} \cdot p\right)^{2^{2s}/p - 1} \le 2 \cdot e^{-2^{s}}
\le 2^{-2^s+1} \le 2^{-s},
$$
using $(1 - 2^{-s} \cdot p)^{-1} \le 2$ and $2^s \ge s + 1$ for $s \ge 1$.
Altogether, $\hint$ outputs $\bot$ with probability at most
$2^{-s} + 2^{-s} + 2^{-n} + 2^{-n} = 2^{-s+1} + 2^{-n+1}$.
\end{proof}

\medskip\noindent{\bf Correctness of $\recon$.} Finally, we show that  $\recon$ outputs the correct answer of $\db[\qrys']$ with high probability assuming that $\hint$ outputs a counter $c \ne \bot$. To do this, we essentially show that whenever $c \ne \bot$,
the distribution of the output of $\recon$ is identical to a proper execution of the underlying PIR.
As a result, the correctness probability of $\recon$ will be the same as the $\varepsilon$-correctness of the underlying PIR.

\begin{lemma}
\label{lem:output-dist}
Fix any database $\db$, query sequence $\qrys'$, and encoding randomness
$\omega_\server$ (thus, the encoded database $\edb$). For every transcript,
query set, and answer assignment $(\tx, S, v)$ occurring with positive
probability in the honest execution below, the following two
distributions are identical:
\begin{enumerate}
\item The distribution $\mathcal{D}$ of $(\bst, \bans)$ resulting from the execution of the dual PIR where
$(\btx, \bS) \leftarrow \advice_{\qrys'}^{\Or, \RO}(\db)$, $\mathbf{c} \leftarrow \hint^{\Or, \RO}(\db, \qrys', (\btx, \bS))$ and
$\bans \leftarrow \recon^{\Or, \RO}(\qrys', (\btx, \bS), \mathbf{c})$,
conditioned on the events $(\btx, \bS) = (\tx, S)$, $\Or(S) = v$, and
$\mathbf{c} \ne \bot$,
where $\bst$ denotes the client state $\st_{i}$ of the accepting iteration
of $\hint$.
\item The distribution $\mathcal{D'}$ of $(\bst', \bans')$ resulting from the proper execution of the underlying PIR where
$\bst' \leftarrow \setup^{\Or}(1^\lambda; \edb)$, $(\bans' \mid \btx') \leftarrow \query^{\Or}(\bst', \qrys'; \edb)$ and $\bS'$ is the set of all server oracle $\Or$ queries performed during $\query$,
conditioned on the event $\{(\btx', \bS') = (\tx, S) \wedge \Or(S) = v\}$.
\end{enumerate}
Moreover, conditioned on the event in the first distribution, $\recon$
outputs exactly the answer $\bans$ of the client in the accepting iteration
of $\hint$. This implies that
\begin{align*}
\Pr[\bans = \db[\qrys'] &\mid (\btx,\bS) = (\tx,S), \Or(S) = v, \mathbf{c} \ne \bot]\\
&= \Pr[\bans' = \db[\qrys'] \mid (\btx',\bS') = (\tx,S), \Or(S) = v].
\end{align*}
\end{lemma}
\begin{proof}
We first show that $\recon$ recovers the state, randomness, and oracle
of the accepting iteration of $\hint$, and then show that the accepting
iteration's client execution is distributed as an honest client execution
conditioned on $(\tx, S, v)$.

We first claim that for every iteration index $i$, iteration $i$ of $\hint$ is a
match (i.e., $\tx_i = \tx$) if and only if iteration $i$ of $\recon$ is a
client-side match (i.e., $\tx'_\client = \tx_\client$). First, note that
$\hint$ and $\recon$ construct identical iterations: the resampling algorithm
of Lemma~\ref{lem:oracle_resample} is deterministic given $\RO_i$ and
$\Or(S)$, and both parties compute $\Or(S)$ by querying the real oracle on
the set $S$ contained in the advice, so iteration $i$ of both algorithms uses
the same oracle $\Or_i$, state $\st_i$, and client randomness
$\omega_{\client,i}$ (and identical resampling failures).
For the forward direction, if $\hint$'s iteration $i$ produces the full
transcript $\tx$, then the client's view in that execution consists of
$(\st_i, \omega_{\client,i}, \Or_i)$ and incoming server messages
$\tx_\server$; the client-side execution
$\query^{\Or_i}_\client(\st_i, \qrys'; \tx_\server)$ run by $\recon$ has the
identical view and therefore produces the identical client messages
$\tx_\client$ and answer. For the reverse direction, suppose $\recon$'s
iteration $i$ produces client messages $\tx_\client$ when fed the server
messages $\tx_\server$. Consider $\hint$'s interactive execution
$\query^{\Or_i}(\st_i, \qrys'; \edb)$. We argue by induction over rounds that
it reproduces the transcript $\tx$: if the first $r-1$ rounds match $\tx$,
then the client's message in round $r$ equals that of $\recon$'s replay
(same state, randomness, oracle, and incoming messages), which is the
corresponding message of $\tx_\client$; and the server's message in round
$r$ equals the corresponding message of $\tx_\server$ by part (1) of
Claim~\ref{clm:replay}, using $\Or_i(S) = \Or(S) = v$ and the identical
encoded database $\edb$ (again, since $\init$ makes no oracle queries).
Hence $\hint$'s iteration $i$ is a full match. It follows that the sequences
of matching iterations in $\hint$ and $\recon$ coincide, so on input the
counter $c$, $\recon$ stops at exactly the accepting iteration of $\hint$
and outputs the answer of that iteration's client execution.

Next, we show that the accepting iteration is distributed honestly.
Condition on $(\btx, \bS) = (\tx, S)$ and $\Or(S) = v$ (and, as
in Lemma~\ref{lem:hint-counter}, on all resamplings succeeding). The
iterations of $\hint$ are independent and identically distributed: iteration
$i$ samples $\Or_i$ uniformly conditioned on $\Or_i(S) = v$, uniform
$(\st_i, \omega_{\client,i})$, and fresh internal coins for
$\setup^{\Or_i}(1^\lambda; \edb)$, all independent of the conditioning event.
Iteration $i$ is an acceptance if and only if the transcript matches $\tx$
and $\st_i = \st'_i$. For any oracle $\Or^*$, state $\st^*$, and randomness
$\omega^*$, the probability that iteration $i$ is an acceptance with
$(\Or_i, \st_i, \omega_{\client,i}) = (\Or^*, \st^*, \omega^*)$ is
\begin{align*}
&\Pr[\Or_i = \Or^*] \cdot 2^{-s} \cdot \Pr[\bomega_{\client,i} = \omega^*] \cdot
\Pr\left[\setup^{\Or^*}(1^\lambda; \edb) = \st^*\right]\\
&\qquad \cdot \1\left[\query^{\Or^*}(\st^*, \qrys'; \edb) \text{ with } \omega^*
\text{ produces } \tx\right],
\end{align*}
where the factor
$\Pr[\setup^{\Or^*}(1^\lambda;\edb) = \st^*]$ arises from the requirement
$\st_i = \st'_i$ over $\setup$'s fresh coins. On the other hand, in the
honest execution, the oracle $\Or$ conditioned on $\Or(S) = v$ is
distributed identically to $\Or_i$, and by part (1) of
Claim~\ref{clm:replay} the event
$\{(\btx', \bS') = (\tx, S)\}$ is exactly the event that the transcript is
$\tx$. Hence the conditional distribution of $(\Or, \bst', \bomega'_\client)$
given $\{(\btx', \bS') = (\tx, S) \wedge \Or(S) = v\}$ assigns each
$(\Or^*, \st^*, \omega^*)$ probability proportional to
\begin{align*}
&\Pr[\Or_i = \Or^*] \cdot \Pr[\bomega_{\client,i} = \omega^*] \cdot
\Pr\left[\setup^{\Or^*}(1^\lambda; \edb) = \st^*\right]\\
&\qquad \cdot \1\left[\query^{\Or^*}(\st^*, \qrys'; \edb) \text{ with } \omega^*
\text{ produces } \tx\right].
\end{align*}
The two expressions differ only by the constant factor $2^{-s}$ (here we use
that client states are padded to exactly $s$ bits), so the
accepted triple $(\Or_i, \st_i, \omega_{\client, i})$, conditioned on
acceptance, is distributed identically to $(\Or, \bst', \bomega'_\client)$
conditioned on the above event. Finally, since the iterations are
independent and identically distributed, the triple at the \emph{first}
acceptance has this same distribution, and further conditioning on the events
defining $\mathbf{c} \ne \bot$ (that the first acceptance occurs within the
loop bounds with fewer than $2^{3s}$ preceding matches) does not change this
distribution, as those events depend only on the other (independent)
iterations and on whether the accepting iteration accepts.

Combining the above, conditioned on
$\{(\btx,\bS) = (\tx,S), \Or(S) = v, \mathbf{c}\ne\bot\}$, the accepting
iteration's triple $(\Or_i, \st_i, \omega_{\client,i})$ is distributed as
$(\Or, \bst', \bomega'_\client)$ in the honest execution conditioned on
$\{(\btx',\bS') = (\tx,S), \Or(S) = v\}$. In both executions, the client's
answer is the same deterministic function of the triple (together with the
fixed $\edb$ and $\qrys'$): the full interactive execution of $\query$ is
determined by the triple and $\edb$ (the server being deterministic).
As shown above, $\recon$ outputs exactly the accepting iteration's answer.
Therefore, $(\bst, \bans) \sim \mathcal{D}$ and
$(\bst', \bans') \sim \mathcal{D}'$ are identically distributed, which
implies the final claim.
\end{proof}

Using this fact, we can deduce that our dual PIR has the same correctness probability as the underlying PIR
except with probability corresponding to the event that the counter $\mathbf{c}$ output by $\hint$ is $\bot$.

\begin{lemma}
\label{lem:same-query}
Fix any database $\db$ and query sequence $\qrys'$ and suppose that the underlying PIR is $\varepsilon$-correct.
Consider a proper execution of the dual PIR to obtain $(\btx, \bS) \leftarrow \advice_{\qrys'}^{\Or,\RO}(\db)$,
$\mathbf{c} \leftarrow \hint^{\Or,\RO}(\db, \qrys', (\btx, \bS))$ and $\bans \leftarrow \recon^{\Or,\RO}(\qrys', (\btx, \bS), \mathbf{c})$. Then,
$$
\Pr[\bans = \db[\qrys']] \ge \varepsilon - 2^{-s+1} - 2^{-n+1}
$$
over the random choice of oracles $\Or$ and $\RO$ as well as the internal randomness of the algorithms.
\end{lemma}
\begin{proof}
Since $\advice$ performs an honest execution of the underlying PIR
using the fixed encoded database $\edb$, the joint distribution of
$(\btx, \bS, \Or(\bS))$ in the dual PIR execution is identical to the joint
distribution of $(\btx', \bS', \Or(\bS'))$ in a proper execution of the
underlying PIR. Then, for every outcome $(\tx, S, v)$,
\begin{align*}
\Pr[\bans = \db[\qrys'] \wedge \mathbf{c} \ne \bot &\mid (\btx, \bS, \Or(\bS)) = (\tx, S, v)]\\
&\ge \Pr[\bans' = \db[\qrys'] \mid (\btx', \bS', \Or(\bS')) = (\tx, S, v)]\\
&\phantom{\ge} - \Pr[\mathbf{c} = \bot \mid (\btx, \bS, \Or(\bS)) = (\tx, S, v)],
\end{align*}
where we use Lemma~\ref{lem:output-dist} for the equality of the two
conditional probabilities of outputting $\db[\qrys']$ and that
probabilities are at most $1$.
Averaging over the outcomes $(\tx, S, v)$ and the encoding randomness
$\omega_\server$ (recall that $\varepsilon$-correctness averages over the
coins of $\init$ as well), and applying the
$\varepsilon$-correctness of the underlying PIR together with
Lemma~\ref{lem:hint-counter}, we conclude
$$
\Pr[\bans = \db[\qrys']] \ge
\Pr[\bans = \db[\qrys'] \wedge \mathbf{c} \ne \bot] \ge
\varepsilon - 2^{-s+1} - 2^{-n+1}.
$$
\end{proof}

\medskip\noindent{\bf Dual PIR Correctness for All Queries.} As the last step of the proof, we now show that our dual PIR construction is correct
even when the input query sequence $\qrys \ne \qrys'$ that is given to both $\hint$ and $\recon$ does not match the fixed query sequence $\qrys'$
used in $\advice_{\qrys'}$. At a high level, we essentially show that, if the dual PIR is incorrect when given input query sequences $\qrys \ne \qrys'$,
then we can construct an adversary that can break the privacy of the underlying PIR $\Pi$. In particular, the adversary may directly use the dual PIR's
correctness to distinguish transcripts produced from $\qrys$ and those transcripts produced using $\qrys'$.

As a note, the privacy adversary to the underlying PIR $\Pi$ only has access to a crypto oracle $\Or$
(used by $\Pi$ to generate the challenge transcript) whereas the dual PIR requires both a crypto
oracle $\Or$ and a random oracle $\RO$. As $\Pi$ only uses $\Or$ and does not ever access
$\RO$, we will have our adversary lazily sample the random oracle $\RO$ locally that will be used
by the dual PIR construction.

\begin{lemma}
\label{lem:dualpiradv}
Suppose the underlying PIR $\Pi$ is $\varepsilon$-correct and $\delta$-private.
Furthermore, we consider our dual PIR construction $\Pi' = (\advice^{\Or,\RO}_{\qrys'}, \hint^{\Or,\RO}, \recon^{\Or,\RO})$ for any fixed query sequence $\qrys'$. For any query sequence $\qrys \ne \qrys'$ and database $\db$,
consider an execution of the dual PIR to obtain $(\btx, \bS) \leftarrow \advice_{\qrys'}^{\Or,\RO}(\db)$,
$\mathbf{c} \leftarrow \hint^{\Or,\RO}(\db, \qrys, (\btx, \bS))$ and $\bans \leftarrow \recon^{\Or,\RO}(\qrys, (\btx, \bS), \mathbf{c})$.
Then,
$$
\Pr[\bans = \db[\qrys]] \ge \varepsilon - 2^{-s+1} - 2^{-n+1} - 2\delta
$$
over the random choice of oracles $\Or$ and $\RO$ as well as the internal randomness of the algorithms.
\end{lemma}
\begin{proof}
Towards a contradiction, suppose there exists a query sequence $\qrys \ne \qrys'$ for database $\db$ such that
$$
\Pr[\bans = \db[\qrys]] < \varepsilon - 2^{-s+1} - 2^{-n+1} - 2\delta.
$$
We construct the following adversary $\mathcal{A}$ for the privacy game of the
underlying PIR $\Pi$ that has access to the crypto oracle $\Or$. The adversary
is explicitly defined for $\qrys_0 = \qrys'$ and $\qrys_1 = \qrys$ as the
two query sequences and $\db$ as the challenge database.

\medskip\noindent{$\mathcal{A}^{\Or}(1^\lambda, \tx, \edb)$:}
\begin{enumerate}
\item Parse $\tx = (\tx_\client, \tx_\server)$ as the client and server communication respectively.
\item Execute the server-side execution of the query algorithm $\query_\server^{\Or}(\tx_\client; \edb)$ where we simulate the client using client's communication in the transcript. Record all server performed queries to the crypto oracle $\Or$ as $S$.
\item Lazily sample a local random oracle $\RO$. For any query $x \in \{0,1\}^*$, check if $\RO(x)$ is defined. If not, set $\RO(x)$ to be a uniformly random string. Then, return $\RO(x)$.
\item Execute $\hintout \leftarrow \hint^{\Or,\RO}(\db, \qrys, (\tx, S))$, where the step of $\hint$ computing $\edb \leftarrow \init(1^\secpar, \db)$ using $\omega_\server = \RO(0)$ is omitted and the encoded database $\edb$ received as input is used in its place.
\item Execute $\reconout \leftarrow \recon^{\Or,\RO}(\qrys, (\tx, S), \hintout)$.
\item If $\reconout = \db[\qrys]$, output $1$.
\item If $\reconout \ne \db[\qrys]$, output $0$.
\end{enumerate}

First, we claim that the execution of $\mathcal{A}$ in the privacy game
with challenge bit $b$ perfectly simulates an execution of the dual PIR in
which $\advice$ is parameterized by $\qrys_b$ while $\hint$ and $\recon$
receive the query sequence $\qrys$. To see this, view the dual PIR's random
oracle $\RO$ as follows: the value $\RO(0)$ is set to the challenger's
internal randomness $\omega_\server$ used by $\init$, while $\RO$ at every
other point takes the values lazily sampled by $\mathcal{A}$. This view
is consistent since $\init$ makes no oracle queries and, by domain separation,
the construction uses the value $\RO(0)$ \emph{only} as the encoding
randomness: $\hint$ uses it only to recompute $\edb$ (the step that
$\mathcal{A}$ replaces with the challenge $\edb = \init(1^\secpar, \db)$
computed with randomness $\omega_\server$, which is by definition the same
value), and $\recon$ never queries $\RO(0)$ at all. With this view, the
challenger's execution is exactly
$(\btx, \bS) \leftarrow \advice^{\Or,\RO}_{\qrys_b}(\db)$: note that
$\mathcal{A}$'s step 2 recovers exactly the set $\bS$ recorded by $\advice$,
since the server's execution during $\query$ is deterministic given
$(\Or, \edb, \tx_\client)$. Hence steps 4--5 of $\mathcal{A}$ produce
$(\bhintout, \breconout)$ distributed exactly as in the dual PIR execution.

Next, we prove that our adversary $\mathcal{A}$ performs a polynomial number of
queries to the crypto oracle $\Or$.
Note that indeed $\mathcal{A}$ queries crypto oracle $\Or$ only on queries in $S$ corresponding to crypto oracle queries
performed by an honest server during the online query phase: step 2
re-executes the deterministic server, $\hint$ and $\recon$ query $\Or$ only
to compute $\Or(S)$, and all executions under resampled oracles $\Or_i$ are
simulated locally using only $\RO$, since the sampling algorithm of
Lemma~\ref{lem:oracle_resample} is deterministic with oracle access only to
$\RO$ given the input $\Or(S)$.
Next, we prove that $\mathcal{A}$ wins the privacy game for the underlying PIR $\Pi$.
Consider the case when $\tx$ is output from executing PIR $\Pi$ with query sequence $\qrys_0 = \qrys'$.
We note that $\mathcal{A}$ outputs $0$ with probability at least
$$
1-\Pr[\bans = \db[\qrys]] > 1 - \varepsilon + 2^{-s+1} + 2^{-n+1} + 2\delta
$$
since $\Pr[\bans = \db[\qrys]] < \varepsilon - 2^{-s+1} - 2^{-n+1} - 2\delta$
by our assumption towards a contradiction.
Now, consider the case when $\tx$ is output from executing PIR $\Pi$ with query sequence $\qrys_1 = \qrys$.
In this case, the adversary is executing the dual PIR where $\advice$ is parameterized by $\qrys$ and both $\hint$ and $\recon$ receive
the same query sequence $\qrys$ as input.
By Lemma~\ref{lem:same-query}, we know that the dual PIR is correct with probability at least
$\varepsilon - 2^{-s+1} - 2^{-n+1}$. Therefore, $\mathcal{A}$ outputs $1$ with probability at least
$$
\Pr[\bans = \db[\qrys]] \ge \varepsilon - 2^{-s+1}-2^{-n+1}.
$$
Averaging the two cases over the uniform challenge bit, $\mathcal{A}$
outputs the correct bit with probability
$$
\frac{1}{2}\left(1 - \Pr[\bans = \db[\qrys] \mid b = 0]\right) +
\frac{1}{2}\Pr[\bans = \db[\qrys] \mid b = 1] >
\frac{1}{2} + \frac{2\delta}{2} = \frac{1}{2} + \delta,
$$
contradicting the $\delta$-privacy of PIR.
\end{proof}

The above essentially shows that the dual PIR construction is correct even if we consider any query sequence $\qrys \ne \qrys'$ where $\qrys'$ is the fixed query sequence parameterized in $\advice_{\qrys'}$ of the dual PIR construction. We can now prove our main result as follows:

\begin{proof}[Proof of Theorem~\ref{thm:dualpirmain}]
First, we note that Lemma~\ref{lem:dualpiradv} states that the dual PIR construction is correct when considering any query sequence $\qrys$ with probability at least $\varepsilon - 2^{-s+1} - 2^{-n+1} - 2\delta$. As this holds for any fixed query sequence $\qrys$ and fixed database $\db$, we note it also holds over uniformly randomly chosen
$\qrys$ and database $\db$ to complete the proof.
\end{proof}

\begin{remark}[Extending Beyond Oracle-Free Encodings]\label{rem:init-ext}
Beyond our main claim, our proof technique extends straightforwardly
to rule out encodings which use oracles in particular restricted ways:
\begin{itemize}[noitemsep]
    \item First, if $\init^{\Or}$ makes only $o(n)$ queries to its oracle, then
    we can include these queries along with the transcript $\tx$ output by
    $\advice$. This is analogous to the result in \secref{sec:comm-sublinear}.
    \item Second, if the queries made by $\init^{\Or}$ are deterministic and
    do not depend on the database, or can be predicted from the transcript
    $\tx$ alone, then $\recon$ can compute these queries directly
    and resample only oracles consistent with them. This is analogous to the result in \secref{sec:comm-det}.
\end{itemize}
However, it is unclear how to extend our result to any encoding which depends
on its oracle. This is the most natural extension left open by our work.
\end{remark}

\heading{PIR with Weaker Privacy Guarantees.}
In our above proof, we note that the adversary performs at most the same number of queries as an honest server. In particular,
our adversary makes all crypto oracle queries as performed by an honest server in the underlying PIR during the online $\query$ protocol. Note,
the adversary does not need to make any crypto oracle queries corresponding to $\init$ (which makes none for schemes with oracle-free encodings) nor server crypto oracle queries performed during $\setup$. As a result, we immediately
rule out the existence of more efficient weak PIR schemes that only obtain privacy guarantees against adversaries
that may perform at most a fixed polynomial number of queries. We point readers to Section~\ref{sec:weakpir} for more details.

\subsection{Impossibility of Dual PIR}

In the next lemma, we rule out the existence of dual PIR which are both correct and too efficient.
In particular, we show that
any dual PIR scheme that supports $k = \Omega(s)$ queries will be correct with probability at most
$2^{s-\Omega(k)}$ %
assuming the upfront $t$-bit advice is less than half
of the database size, $t < n/2$. We will prove this as an information-theoretic result and, thus, it holds even in relative worlds.

\begin{lemma}\label{lem:dpir}
	If $n, k, t, s$ are integers with $t < n/2$, $n$ sufficiently large, and
	$1000\log n \le k \le 0.19n$,
	then there is no $(n,k,t,s)$-dual PIR that is $\varepsilon$-correct for
	$\varepsilon > 2^{s - 0.15k}$. Moreover, this is true
	 relative to any oracle $\Or$ that is independent of the database and queries.
\end{lemma}

By setting $k \ge 14s$, we immediately obtain that it is impossible
for a dual PIR to be correct with probability $2^{-s}$. %
In other words, there are no interesting constructions of dual PIR when the final $s$-bit hint string is small compared
to the number of queries $k$ such as the setting of $s = o(k)$.
Given the prior section showing that any efficient PIR with preprocessing scheme may be transformed into
an efficient dual PIR, we immediately can prove our main lower bound.

\begin{proof}[Proof of \thref{thm:main-lower-bound}]
Towards a contradiction, suppose there exists a too-efficient-to-be-true
PIR with preprocessing scheme with $s$-bit client storage supporting $k = \Omega(s)$ queries with $k \le 0.19n$
where the total bits of communication and server cryptographic bit operations is $t$
that obtains
$\varepsilon$-correctness and $\delta$-privacy
for any choice of $\varepsilon \ge 2^{-s+4}$, $\delta \le \varepsilon / 3$ and $t \le n/2$.
By \thref{thm:dualpirmain}, this implies a
$(n, k, t, 3s)$-dual PIR with $\varepsilon'$-correctness
where
$$\varepsilon' \ge \varepsilon - 2^{-s+1} - 2^{-n+1} - 2\delta  \ge 2^{-s}$$
using the fact that $s \le n$, which holds as we take $k \ge 27s$ below and $k \le 0.19n$ (so that $2^{-n+1} \le 2^{-s+1}$, giving
$\varepsilon' \ge \varepsilon/3 - 2^{-s+2} \ge (16/3 - 4) \cdot 2^{-s} \ge
2^{-s}$).
Next, we can apply \lemref{lem:dpir} showing that there cannot exist
any $(n, k, t, 3s)$-dual PIR that is $\varepsilon'$-correct where
$\varepsilon' \ge 2^{-s}$ to achieve a contradiction, which
is true whenever $k \ge 27s$ (as the dual PIR's hint size is
$3s$, we require $2^{-s} > 2^{3s - 0.15k}$, i.e., $0.15k > 4s$).
In other words, this implies that it must be that $t > n/2$ meaning that the
total bits in the online communication and the server cryptographic bit operations must
be at least $\Omega(n)$. This immediately implies that the total online computation is $\Omega(n)$
and the amortized online computation is $\Omega(n/k)$, which is $\Omega(n/s)$ for $k = \Theta(s)$.
\end{proof}

The remainder of this section is proving the impossibility of efficient dual PIR.
To prove this impossibility, we will relate the correctness properties of dual PIR with the
subkey prediction problem studied in~\cite{C:BelKanRog16} (we also refer readers back to Section~\ref{sec:subkey-def}).
As a reminder, the subkey prediction considers the probability that an adversary is able to correctly guess a subkey consisting of $k$ random locations
of a large $n$-bit key when given $t$ bits of leakage. We note the collision of parameters is not coincidental as there are strong
similarities between the two problems. Interestingly, subkey prediction is used as a tool to prove the security of a construction
of symmetric encryption (upper bounding adversarial advantage) in the bounded retrieval model.
In contrast, we will use the subkey prediction problem as a way to prove lower bounds and impossibilities
about dual PIR constructions (upper bounding dual PIR correctness).

Prior work \cite{C:BelKanRog16} presented an upper bound for subkey prediction and later use approximations to obtain concrete bounds. However, they omit a proof of asymptotic bounds for the
parameters necessary in our proof. We will
provide the relevant analysis below using \lemref{lem:subkey} proven in~\cite{C:BelKanRog16}.
Before we present the lemma, we recall some notation
that will be important.
As a reminder, $B_n(r)$ ball of radius $r$ centered at the origin of
the $n$-dimensional boolean hypercube (all points with distance at most $r$ from $0^n$.
We use
$\mathsf{rd}_n(N) :=
\max\{\hat{r}\in \mathbb{Z}_{\ge 0} : |B_n(\hat{r})| \le N\}$ 
as the radius of the largest ball with at most $N$ points.

\begin{lemma}\label{lem:subkey-ours}
	For sufficiently large integers $n, t, k$ such that $t \le n/2$ and
	$1000\log n \le k \le 0.19 n$ and for all (computationally unbounded) adversaries
	$\A=(\A_1,\A_2)$,
	\[
	\Pr\big[\A_2(\bqrys, \A_1(\bdb)) = \bdb[\bqrys]\big] \le 2^{-0.15 k}
	\]
	where $\bqrys$, the $k$ subkey locations, is uniform over $[n]^k$,
	$\bdb$, $n$-bit large key, is uniform over $\bits^n$,
	adversary $\A_1$ outputs at most $t$ bits of leakage, and
	adversary $\A_2$ outputs $k$ bits.
\end{lemma}
\begin{proof}
We start with the adversarial advantage proven by Bellare, Kane,
and Rogaway \cite{C:BelKanRog16} (see \lemref{lem:subkey}) that
we present below where we denote $N = 2^{n - t} \ge 2^{n/2}$ and
$r = \min(1 + \mathsf{rd}_n(N), n)$:
$$
\Pr[\A_2(\bqrys, \A_1(\bdb)) = \bdb[\bqrys]] \le \frac{1}{N} \cdot \sum\limits_{i = 0}^{r} {n \choose i} \cdot \left(1 - \frac{i}{n}\right)^k.
$$
For convenience, we rewrite this bound using conditional probabilities as follows:
$$
\Pr[\A_2(\bqrys, \A_1(\bdb)) = \bdb[\bqrys]] \le \frac{|B_n(r)|}{N} \cdot \sum\limits_{i = 0}^{r} \Pr[|\mathbf{X}| = i \mid \mathbf{X} \in B_n(r)] \cdot \left(1 - \frac{i}{n}\right)^k.
$$
for $\mathbf{X}$ uniformly random over $\bits^n$.

To start, we show that $|B_n(r)| \le n \cdot N$. First, note that $r \le \mathsf{rd}_n(N) + 1$. By definition of  $\mathsf{rd}_n(N)$, we know that $|B_n(r - 1)| \le N$. Also, $r \ge 2$ since $|B_n(1)| = n + 1 \le 2^{n/2} \le N$ for sufficiently large $n$.
Every point of Hamming weight $r$ is obtained from exactly $r$ points of Hamming weight $r - 1$ by flipping one of at most $n$ coordinates, so ${n \choose r} \le (n/r) \cdot {n \choose r - 1} \le (n/2) \cdot |B_n(r - 1)|$. Therefore, we can see that $|B_n(r)| = |B_n(r - 1)| + {n \choose r} \le n \cdot |B_n(r - 1)| \le n \cdot N$. As a result, we see that $|B_n(r)| / N \le n$.

Next, we split the points into two disjoint sets of Hamming weights: those with at most $n/10$ Hamming weight and the remainder with strictly more than $n/10$ Hamming weight.
We can rewrite the adversary's guessing probability as:
\begin{align*}
\Pr[&\A_2(\qrys, \A_1(\db))= \db[\qrys]]\\
&\le n \cdot \left(\sum\limits_{i = 0}^{n/10} \Pr[|x| = i \mid x \in B_n(r)] \cdot \left(1 - \frac{i}{n}\right)^k + \sum\limits_{i = n/10 + 1}^{r} \Pr[|x| = i \mid x \in B_n(r)] \cdot \left(1 - \frac{i}{n}\right)^k\right)\\
&\le n \cdot \sum\limits_{i = 0}^{n/10} \Pr[|x| = i \mid x \in B_n(r)] \cdot \left(1 - \frac{i}{n}\right)^k  + n \cdot \left(\frac{9}{10}\right)^{k}\\
&\le n \cdot \sum\limits_{i = 0}^{n/10} \Pr[|x| = i \mid x \in B_n(r)]  + 2^{-0.151 k + \log n}\\
&= n \cdot \Pr[x \in B_n(n/10) \mid x \in B_n(r)] + 2^{-0.15k}.
\end{align*}
where we use the fact that $\log_2 (9/10) < -0.152$ and that $k \ge 1000 \log n$.

It remains to bound the first term of the summand. To do this, we will bound the conditional probability
that a random element $x$ from $B_n(r)$ has Hamming weight at most $n/10$. That is, we bound the probability
$\Pr[\mathbf{X} \in B_n(n/10) \mid \mathbf{X} \in B_n(r)]$. We first re-write this probability as
$$
\Pr[\mathbf{X} \in B_n(n/10) \mid \mathbf{X} \in B_n(r)] = \frac{|B_n(n/10)|}{|B_n(r)|}.
$$
Next, we show that $|B_n(r)| \ge |B_n(0.11n)|$.
To see this, we consider the size of the ball $B_n(0.11 n)$ that can be upper bounded as
$$
|B_n(0.11 n)| = \sum\limits_{i = 0}^{0.11 n} {n \choose i} \le 2^{h(0.11)n} < 2^{n/2} \le N
$$
using the fact that $N = 2^{n - t} \ge 2^{n/2}$ since $t \le n/2$. We use the well-known bounds $2^{h(b/a)a}/(a+1)  \le \sum_{i=0}^b {a \choose i} \le 2^{h(b/a)a}$ for $b \le a/2$, where $h(p) = -p\log_2 p - (1-p)\log_2 (1-p)$ is the binary entropy function. We also use the fact that $h(0.11) < 1/2$ and, thus, $2^{h(0.11)n} < 2^{n/2}$
in the second inequality.
This immediately implies that $|B_n(r)| \ge |B_n(0.11n)|$ since $|B_n(r)| \ge N$ since $r = \min(\mathsf{rd}_n(N) + 1, n)$ and we note that $\mathsf{rd}_n(N) + 1 < n$ since $N \ge 2^{n/2}$.
Therefore, we can get that 
\begin{align*}
\Pr[\mathbf{X} &\in B_n(n/10) \mid \mathbf{X} \in B_n(r)] \le \frac{|B_n(n/10)|}{|B_n(r)|} \\
&\le \frac{|B_n(n/10)|}{|B_n(0.11n)|} \le (n+1) 2^{(h(0.1) - h(0.11))n} \le n2^{-0.03n}. %
\end{align*}
where again we use the well-known bounds
$2^{h(b/a)a}/(a+1)  \le \sum_{i=0}^b {a \choose i} \le 2^{h(b/a)a}$ and the fact that $h(0.1) - h(0.11) < -0.0309$, so the last inequality holds for sufficiently large $n$.
Putting it altogether, we obtain the following:
\begin{align*}
\Pr[&\A_2(\bqrys, \A_1(\bdb))= \bdb[\bqrys]]
\le n^2 \cdot 2^{-0.03n} + 2^{-0.151k} \le 2^{-0.15k}
\end{align*}
using the fact that $k \le 0.19n$ and that $n$ is sufficiently large.
\end{proof}

Next, we observe also that any efficient dual PIR construction could equivalently
be considered as an extension of constructing a good adversary for subkey prediction.
In fact, the subkey prediction game is essentially the same, except that
the adversary is not allowed to receive any additional secret function in the last phase of dual PIR.
We use this reduction to complete the proof of \lemref{lem:dpir} showing that any dual PIR
may be used to construct an adversary in the subkey prediction game. Thus, there does not exist
a dual PIR that is both correct and efficient.

\begin{proof}[Proof of \lemref{lem:dpir}]
	Consider $\bqrys \getsr [n]^k$ to be the uniformly random $k$ key locations and
	$\bdb \getsr \bits^n$ to be the uniformly random $n$-bit key.
	We prove the lemma by considering the min-entropy of the random variable
	$\bdb[\bqrys]$ given the outputs of the algorithms $\badviceout \gets \advice(D)$ and $\bhintout \gets \hint(\bdb, \bqrys, \btx)$
	in the DPIR correctness game (\defref{def:dpir}).
	In particular, we will use the fact that
	\[
	\Pr[\recon(\bqrys, \badviceout, \bhintout) = \bdb[\bqrys]]
	\le 2^{-\minH(\bdb[\bqrys] \mid \bqrys, \badviceout,\bhintout)}
	\]
	to prove the bound. So, all that remains is to show that
	\[
	\minH(\bdb[\bqrys] \mid \bqrys,\badviceout,\bhintout) \ge \Omega(k) - s.
	\]
	For this, we can first appeal to \lemref{lem:subkey-ours}, which shows that
	\[
		\minH(\bdb[\bqrys] \mid \bqrys, \badviceout) \ge 0.15 k,
	\]
	for $t < n/2$ and $n$ sufficiently large.
	Finally, we observe that $\bhintout$ is $s$ bits and, thus, supported on at most $2^{s}$ elements.
	This means that $\minH(\bdb[\bqrys] \mid \bqrys,\badviceout,\bhintout) \ge 0.15k - s$ via
	\lemref{lem:cond-minh},
	and the claimed lemma follows immediately.
\end{proof}

\subsection{Impossibility of (Public-Key) Doubly Efficient PIR}
\label{sec:depir}

In this section, we prove that public-key (and therefore unkeyed)
doubly efficient PIR (DEPIR) with an oracle-free encoding built from blackbox cryptography is impossible.
In particular, we show that any such DEPIR constructions requires linear
server time per query (which shows that DEPIR is impossible since it is
defined to use sublinear time).
As a reminder, prior work~\cite{EC:LinMooWic25} ruled out the existence of a
restricted class of secret-key DEPIR
constructions (thus, also public-key/unkeyed DEPIR) in the crypto oracle model.
However, their lower bounds required the DEPIR to consist of a single
roundtrip query, along with the restriction of a passive server that simply
acts as memory without performing any computation. In this work, we rule out the
existence of any construction of any public-key DEPIR in the crypto oracle
model whose preprocessing (key generation and database encoding) makes
no queries to the crypto oracle. Our lower bound makes no assumptions on the
online protocol, which may consist of arbitrarily many rounds with an
actively computing server making arbitrary use of the crypto oracle. The two
results are therefore incomparable:~\cite{EC:LinMooWic25} allows
oracle-dependent encodings but restricts the query structure, whereas we
allow arbitrary query protocols but restrict the preprocessing.

We note that the oracle-free preprocessing restriction is substantive
for public-key DEPIR, where known constructions use cryptography heavily
inside their encodings (such as the virtual blackbox obfuscation used
in~\cite{TCC:BIPW17}); ruling out blackbox DEPIR with oracle-dependent
preprocessing remains an interesting open problem (see also
Remark~\ref{rem:init-ext}). For this reason, our result is most interesting
for \emph{unkeyed} DEPIR, where the encoding is a public procedure applied
to the database. In other words, our result shows that a cryptographic
database encoding is necessary: blackbox cryptography used only during the
online phase cannot yield sublinear-time PIR.

Note that our syntax allows a PIR to encode an $n$-bit database $\db$
as $\edb$. This maps directly onto the \emph{unkeyed} DEPIR syntax. However,
we note that our result below additionally rules out \emph{public-key} DEPIR
as well. This is because without loss of generality, we can assume that the
public-key is included in the encoding $\edb$ and then sent as the first
message of the query algorithm. In more detail, $\init(1^\secpar, \db)$
runs the (oracle-free) key generation and encoding of the public-key DEPIR,
outputs $\edb$ consisting of the public key and encoded database, and
discards the secret key, which the online protocol of a public-key DEPIR
never uses~\cite{TCC:BIPW17}. Therefore, any public-key DEPIR with a
$o(n)$ size public-key\footnote{Since the public key is generated
independently of the database, requiring a $o(n)$ size public key is a very
mild condition; still, our results do not rule out public-key DEPIR with
public keys of $\Omega(n)$ size.} can be transformed into an unkeyed DEPIR with the same parameters at the cost of an additional round and $o(n)$ communication. Since
we have no restrictions on the rounds or communication beyond that it is
$o(n)$, our bound will apply to both public-key and unkeyed
DEPIR.

\heading{Improving \thref{thm:main-lower-bound} for DEPIR.}
The main difference between DEPIR and the PIR with preprocessing setting of \thref{thm:main-lower-bound} is
that the $\setup$ algorithm (see Definition~\ref{def:pir}) is a no-op
as the client receives no private state. Note, we could try to apply \thref{thm:main-lower-bound} to obtain a lower bound
on the query computation immediately. However, this result only applies for $k = \Omega(s)$ queries where $s = \omega(\log(n + \lambda))$.
To apply \thref{thm:main-lower-bound}, we could pad the public-key DEPIR with a dummy $\omega(\log(n + \lambda))$-bit client
storage and obtain that the total query computation over $\omega(\log(n + \lambda))$ queries must be $\Omega(n)$.
Instead, we will refine our proof to show that there does not exist any public-key DEPIR that obtains sublinear query computation
for even a single query. This proof will still follow the same proof framework of building a dual PIR and then showing that a dual PIR cannot be both correct and efficient.

\begin{lemma}
\label{lem:depir-dual}
Consider any public-key DEPIR construction with an oracle-free encoding
for databases of size $n$ where the communication transcript and all server crypto oracle queries during the query algorithm are $t$ bits in total.
Furthermore, suppose $\Pi$ is $\varepsilon$-correct and $\delta$-private.
Then, there exists a $(n, 1, t, 0)$-dual PIR with $\varepsilon'$-correctness satisfying $\varepsilon' \ge \varepsilon - 2\delta - 2^{-n+2}$.
\end{lemma}
\begin{proof}
We build our dual PIR using the public-key DEPIR construction in a blackbox manner following \thref{thm:dualpirmain}
with two modifications.
First, the $\hint$ algorithm is modified to perform no operations and always output nothing.
Secondly, the $\recon$ algorithm is modified such that it replaces the counter returned by $\hint$ with $c = 0$.
The resulting dual PIR has the same correctness analysis with the
following differences. Since $s = 0$, the client state is empty, so the
uniformly sampled state always equals the honestly sampled one: every
transcript match is an acceptance, and $\recon$ (using $c = 0$) outputs the
answer of the first matching iteration. Hence the counter-overflow failure
of Lemma~\ref{lem:hint-counter} disappears, and the remaining failure modes
are the resampling failures (probability $2^{-n}$ as before) and the event
that no match occurs among the $z = 2^{3n}$ iterations. For the latter,
rerunning the computation of Lemma~\ref{lem:p} with threshold $2^{-2n}$
shows that $p \ge 2^{-2n}$ except with probability
$2^{-2n} \cdot 2^{o(n)} \le 2^{-n}$; conditioned on $p \ge 2^{-2n}$, no
match occurs within the $z = 2^{3n}$ iterations with probability at most
$(1 - 2^{-2n})^{2^{3n}} \le e^{-2^n}$. Altogether, the failure probability
is at most $2^{-n+1} + e^{-2^n} \le 2^{-n+2}$.
\end{proof}

Next, we prove a more direct impossibility for the dual PIR construction for a single query $k = 1$ with no additional encoded information after receiving the query since $s = 0$. This is actually equivalent to the well-known one-way
communication complexity of the \textsc{Index} problem, which was shown
to require $\Omega(n)$ communication originally by Kremer et al.~\cite{STOC:KreNisRon95,CC:KreNisRon99}. For clarity and completeness,
we recount a direct proof below with explicit constants
using commonly known information theory measures and inequalities
(see,~\cite{CoverThomas06} for the relevant definitions details).

\begin{lemma}[Due to~\cite{CC:KreNisRon99}]
\label{lem:dualpir-single-query}
For any $t < n/100$, then there is no $(n, 1, t, 0)$-dual PIR that is $\varepsilon$-correct for $\varepsilon \ge 0.6$. Moreover,
this is true relative to any oracle $\Or$ that is independent of the database and queries.
\end{lemma}
\begin{proof}
Let $\badviceout$ be the $t$ bit output of $\advice(\bdb)$ for a uniformly
random $n$-bit database $\bdb$. Then the output
$\breconout \gets \recon(\bqrys,\badviceout)$ over the choice of a uniformly
random index $\bqrys \in [n]$ (we drop the notation of
$\hint$ and its output, since we care about $0$-bit outputs).

Now, we consider the success probability of any dual PIR subject to our
constraints, which, expanding with the law of total probability over
all values $\bqrys$ can take is
\[
  \Pr\big[\breconout = \bdb[\bqrys]\big]
  = \frac{1}{2}
    + \frac{1}{n}\sum_{\qrys=1}^{n}
        \Ex_{\badviceout}\!\left[\max_{b \in \bits}\Pr[\bdb[\qrys] = b \mid \badviceout] - \tfrac{1}{2}\right].
\]
The quantity in the expectation is exactly $\Delta(\bdb[\qrys], \mathbf{b} \mid \badviceout)$, the
statistical distance of $(\bdb[\qrys] \mid \badviceout)$ and a uniformly
random bit $\mathbf{b}$.
Below, we use
$D_\mathsf{KL}(\bdb[\qrys], \mathbf{b} \mid \badviceout)$ denotes the KL divergence
between the conditional distribution of $\bdb[\qrys]$ given $\badviceout$ and
the uniform distribution on $\bits$.
Applying Pinsker's inequality pointwise and then
Jensen's inequality
(using concavity of $\sqrt{\cdot}$) yields
\[
  \Ex_{\badviceout}\!\left[\Delta(\bdb[\qrys], \mathbf{b} \mid \badviceout)\right]
  \leq \Ex_{\badviceout}\!\left[\sqrt{\frac{D_\mathsf{KL}(\bdb[\qrys] \mid \badviceout)}{2}}\right]
  \leq \sqrt{\frac{\Ex_{\badviceout}\!\left[D_\mathsf{KL}(\bdb[\qrys] \mid \badviceout)\right]}{2}}
  = \sqrt{\frac{I(\bdb[\qrys]\,;\,\badviceout)}{2}},
\]
where the last equality uses the fact that $\bdb[\qrys]$ is uniform and the
definition of mutual information.
Next, since each $\bdb[\qrys]$ are independent and
$|\badviceout| = t$, the chain rule gives
$\sum_{\qrys} I(\bdb[\qrys] ; \badviceout) \leq I(\bdb ; \badviceout)
\leq t \ln 2$.
Finally, we can apply Cauchy--Schwarz, which shows that
$\frac{1}{n}\sum_{\qrys=1}^{n} \sqrt{I(\bdb[\qrys];\badviceout)/2}
\leq \sqrt{\sum_{\qrys} I(\bdb[\qrys];\badviceout) / (2n)}
\leq \sqrt{t\ln 2/(2n)}$.
This finally gives us that
\[
  \Pr\big[\breconout = \bdb[\bqrys]\big]
  \leq \frac{1}{2} + C\sqrt{\frac{t}{n}},
\]
with explicit constant $C = \sqrt{(\ln 2)/2} \approx 0.589$.

Finally, taking $t < n / 100$, we see that no dual PIR can exist that is
correct with probability better than $\varepsilon = \frac{1}{2} + \frac{1}{10}\sqrt{\frac{\ln 2}{2}} \approx 0.559$.
\end{proof}

\begin{theorem}
\label{thm:public-key-depir}
Relative to any crypto oracle $\Or$, there does not exist a single-server public-key doubly efficient PIR
with an oracle-free encoding (\defref{def:oracle-free})
that is $\varepsilon$-correct and $\delta$-private with sublinear query computation $t = o(n)$ against computationally unbounded
adversaries making $\poly(\lambda)$ queries to $\Or$ for any choice of $\varepsilon \ge 9/10$ and $\delta \le 1/20$.
\end{theorem}
\begin{proof}
First, we apply \lemref{lem:depir-dual} to obtain a $(n,1,t,0)$-dual PIR that is $\varepsilon'$-correct where $\varepsilon' \ge \varepsilon - 2\delta - 2^{-n+2} \ge 9/10 - 1/10 - 2^{-n+2} > 3/4$ for sufficiently large $n$.
This immediately contradicts \lemref{lem:dualpir-single-query} to show that such a public-key DEPIR is impossible.
\end{proof}

\subsection{Impossibility of Weak PIR (with Preprocessing)}
\label{sec:weakpir}

In this section, we consider the notion of a weak PIR where the privacy guarantees only hold for weaker adversaries
that are limited to a fixed polynomial number of queries to the crypto oracle. We note such weaker cryptographic primitives have been
studied in other contexts. The most famous example are Merkle puzzles~\cite{merkle1978secure} that
enable building weak key exchange where honest parties perform $q$ random oracle queries and provides privacy guarantees against
any adversary that performs $o(q^2)$ random oracle queries as shown in~\cite{TCC:BihGorIsh08}.

To our knowledge, prior works have not ruled out the existence of weak single-server PIR with sublinear computation even without any preprocessing. For example, both of the previous results ruling out sublinear communication single-server
PIR~\cite{EC:DiCMalOst00} and sublinear public-key operations~\cite{EC:DujHaj24} reduce to the impossibility of oblivious transfer
in the random oracle model. However, replicating the same reductions in~\cite{EC:DiCMalOst00,EC:DujHaj24} would not rule out weak PIR as there do exist weak oblivious transfer
schemes using one-way functions~\cite{TCC:BihGorIsh08}.

We show that our prior lower bounds immediately rule out the existence of such weak PIR built from any blackbox cryptography.
In particular, our reduction constructs a weak adversary that performs crypto oracle queries
that is at most the number of crypto oracle queries performed by an honest server during the online query protocol of a proper execution of the PIR
(see the adversary in the proof of Lemma~\ref{lem:dualpiradv}). Therefore, our lower bounds immediately apply
to weak PIR where the honest parties make $q$ crypto oracles that are secure against any adversaries that perform
$O(q)$ crypto oracle queries.

\begin{theorem}
\label{thm:weak-pir}
Relative to any crypto oracle $\Or$, there does not exist a single-server weak PIR
with an oracle-free encoding (which is trivially satisfied by any PIR
whose server stores the database unmodified)
that is $\varepsilon$-correct and $\delta$-private with sublinear communication, sublinear server cryptographic operations such that honest parties make at most $q$ queries to $\Or$ against computationally unbounded
adversaries making $O(q)$ queries to $\Or$ for any choice of $\varepsilon \ge 9/10$ and $\delta \le 1/20$.
\end{theorem}
\begin{proof}
This immediately follows from \thref{thm:public-key-depir} where we use the fact that the adversary constructed in Lemma~\ref{lem:dualpiradv} makes at most $O(q)$ crypto oracle queries.
\end{proof}

Note, the above applies to even public-key DEPIR where the server is able to perform arbitrary oracle-free encoding of the database. Note,
we can also rule out the existence of single-server weak PIR with preprocessing by applying the same observation to \thref{thm:main-lower-bound}.

\begin{theorem}
Relative to any crypto oracle $\Or$, there does not exist a single-server
PIR with preprocessing scheme with an oracle-free encoding that is $\varepsilon$-correct
and $\delta$-private
with $s$-bit client storage and total communication and server cryptographic operations of
$t = o(n)$ bits
supporting $k = \Theta(s)$ queries
such that honest parties make at most $q$ queries to $\Or$
against computationally unbounded adversaries making $O(q)$ queries to $\Or$
for any choice of $s = \omega(\log(n + \lambda))$, $\varepsilon = \Omega(2^{-s})$ and $\delta \le \varepsilon / 3$.
\end{theorem}

\section{Communication Lower Bounds}
\label{sec:comm-lb}

This section presents some notable corollaries and
variations of our main result that yield \emph{communication} lower bounds.
In contrast to our computation bound, each of the variations applies to a
somewhat restricted class of PIR with preprocessing. As
discussed in Section~\ref{sec:bb-sep} (and we iterate here),
there are good reasons to believe
these communication bounds cannot apply in full generality.

\subsection{PIR with Sublinear Server Cryptographic Operations}
\label{sec:comm-sublinear}

This communication lower bound is a direct corollary of
Theorem~\ref{thm:main-lower-bound}. Consider a correct PIR protocol with
an oracle-free encoding with a batch
size $k \ge \Omega(s)$ that requires
that a server only makes $o(n)$ total queries to its crypto oracle during
the online query protocol. Then, our
main theorem shows that this would imply a dual PIR with $o(n)$ advice and
a $O(s)$-bit hint that contradicts information-theoretic bounds.
We state this formally below.

\begin{corollary}\label{cor:comm-sublinear}
Relative to any crypto oracle $\Or$, any single-server
PIR with preprocessing scheme with an oracle-free encoding
(\defref{def:oracle-free}) that is $\varepsilon$-correct
and $\delta$-private
with $s$-bit client storage and total of $t=o(n)$ server-side queries
to $\Or$,
supporting $k = \Omega(s)$ queries with $k \le 0.19n$
against computationally unbounded adversaries making $\poly(\lambda)$
queries to $\Or$, requires total communication $\Omega(n)$.
This holds for any choice of $s = \omega(\log(n + \lambda))$,
$\varepsilon = \Omega(2^{-s})$ and $\delta \le \varepsilon / 3$.
\end{corollary}

We find it important to consider this lower bound separately from the main
result, since it illustrates the limitations of blackbox constructed PIR
when the server doesn't use a significant amount of cryptography.
Notice that this bound applies to PIR which run in $\Omega(n)$ time but which
do not use the oracle more than $o(n)$ times. And in fact, many of the OWF-based
PIR with preprocessing require the server issue no OWF evaluations,
e.g., see~\cite{EC:HPPY25}. This bound shows that these approaches are at their
communication limit, when supporting an unbounded number of queries. This
is because they use $O(n)$ communication across
$k = \Theta(s)$ queries to refresh a client's state, which is optimal given
the above corollary.

Perhaps more importantly, this corollary shows significant limitations
of what can be achieved by combining server-side and client-side preprocessing
(since our definition allows database encoding).
It has been an open question, whether one could decrease the (amortized)
communication of OWF-based, single-server PIR by combining server and client
preprocessing. The above shows that this is not possible from blackbox
cryptography without the server using at least $\Omega(n)$ calls to its crypto
oracle (which is achievable via FHE). However, it remains open whether one
could improve communication using $\Omega(n)$ server-side
evaluations of a weaker oracle like a random oracle,
which may be more practically efficient than the
FHE-based construction.

\subsection{PIR with Communication-Determined Server Cryptographic Operations}
\label{sec:comm-det}

Our next communication lower bound follows from similar techniques used
in our proof of Theorem~\ref{thm:main-lower-bound}. We observe that
for a specific class of PIR with preprocessing, we can actually simplify the
PIR to dual PIR reduction. In particular, we define the following class of
PIR which are ``communication-determined,'' meaning that the server-side
queries to $\Or$ depend only on the transcript $\tx$ and not on the database
itself. More formally, we provide the following definition.

\begin{definition}
\label{def:determined}
A PIR protocol $\pir^{\Or}=(\init^{\Or},\setup^{\Or},\query^{\Or})$ is
\emph{communication-determined} if the server-side queries on the
$\query^{\Or}$ protocol are a function of the online transcript $\tx$ and
$\Or$. In other words, there exists an oracle algorithm
$\mathsf{SvrQ}^{\Or}(\tx)$, issuing
$\poly(\lambda)$ queries to $\Or$, which outputs $S$, the $\Or$ queries
by the server. (More formally, it must return the server queries
with probability 1 for every $\qrys$ and $\db$ over an honestly
sampled transcript $\tx$ and oracle $\Or$.)
\end{definition}

Next, we present the dual PIR used for the following theorem. We highlight
the differences from the version used in Theorem~\ref{thm:dualpirmain} with
old code in {\color{gray} gray} and new code in {\color{blue} blue}. We
then sketch the proof of the reduction theorem and state the
corresponding corollary. For brevity, we leave out the intermediate
lemmas as they are nearly identical to those used to prove
Theorem~\ref{thm:dualpirmain} with the exception that the server side
queries are computed by each party. Roughly speaking, the old difference is that
the set $S$ of server-side queries is computed by each party instead of being
sent along with the transcript.\footnote{For our
reduction, we could allow that $\mathsf{SvrQ}$ takes $\qrys$ as input, but this
is essentially equivalent to the given definition. If the output depended
on $Q$ in a (computationally) noticeable way, then a server would be able
to break the privacy of PIR.}

\begin{enumerate}
\item $\advice_{\qrys'}^{\Or, \RO}(\db)$: %
    \begin{itemize}
    	\item Run $\edb \gets \init(1^{\secpar}, \db)$ using
		$\omega_{\server} = \RO(0)$ as the algorithm's randomness.
        \item Run $(\st;\bot)\gets\setup^{\Or}(1^\lambda ;\edb)$.
        \item Run $(\ans;\bot\mid \tx)\gets\query^{\Or}(\st,\qrys';\edb)$.
        \item {\color{gray} Record all server-side oracle $\Or$ queries executed during the $\query$ algorithm as $S$.}
        \item {\color{blue} Output $\adviceout \leftarrow \tx$.} {\color{gray} Output $\adviceout \leftarrow (\tx, S)$.}
    \end{itemize}
\item
$\hint^{\Or, \RO}(\db, \qrys, \adviceout)$:
    \begin{itemize}
    	\item {\color{blue} Parse $\adviceout = \tx$.} {\color{gray} Parse $\adviceout = (\tx, S)$.}
    	\item {\color{blue} Compute $S \gets \mathsf{SvrQ}^{\Or}(\tx)$.}
    	\item Compute $\Or(S) \leftarrow \{(q,\Or(q)) \mid q \in S\}$.
	\item Run $\edb \leftarrow \init(1^{\secpar}, \db)$ using
	$\omega_{\server} = \RO(0)$ as the algorithm's randomness.
	\item Set $z \leftarrow 2^{3n+2s}$.
    	\item Construct virtual random oracles $\RO_0, \RO_1,\ldots,\RO_z$ such that $\RO_i(\cdot) = \RO(i \mid\mid \cdot)$ and $i \in \{0,1,\ldots,z\}$ is represented using a string of $O(n+s)$ bit length.
	\item Instantiate counter $c \leftarrow 0$.
	\item For $i \in \{1,\ldots,z\}$ and while $c < 2^{3s}$:
	\begin{enumerate}
		\item Sample oracle $\Or_i \leftarrow \mathcal{A}^{\RO_i}(\Or(S))$ such that $\Or_i(S) = \Or(S)$ using Lemma~\ref{lem:oracle_resample} with parameter $\zeta = 6n$. If the algorithm fails to sample crypto oracle $\Or_i$, then output $\hintout \leftarrow \bot$.
		\item Compute $(\st_i, \omega_{\client,i}) \leftarrow \RO_0(i)$.
		\item Compute $(\st_i'; \bot) \leftarrow \setup^{\Or_i}(1^\lambda; \edb)$.
		\item Compute $(\ans; \bot \mid \tx_i) \leftarrow \query^{\Or_i}(\st_i, \qrys; \edb)$ using $\omega_{\client,i}$ as the randomness for the client side of the $\query$ protocol.
		\item If $\tx_i = \tx$:
		\begin{enumerate}
			\item If $\st_i = \st_i'$, output $\hintout \leftarrow c$.
			\item Increment $c \leftarrow c + 1$.
		\end{enumerate}
	\end{enumerate}
	\item Output $\hintout \leftarrow \bot$.
    \end{itemize}
\item
$\recon^{\Or, \RO}(\qrys,\adviceout,\hintout)$:
    \begin{itemize}
    \item {\color{blue} Parse $\adviceout = \tx$.} {\color{gray} Parse $\adviceout = (\tx, S)$.}
    \item {\color{blue} Compute $S \gets \mathsf{SvrQ}^{\Or}(\tx)$.}
	\item Parse $\hintout = c$.
    	\item Parse $\tx = (\tx_\client, \tx_\server)$ as client and server communication respectively.
	\item Compute $\Or(S) \leftarrow \{(q, \Or(q)) \mid q \in S\}$.
	\item Set $z \leftarrow 2^{3n+2s}$.
    	\item Construct virtual random oracles $\RO_0, \RO_1,\ldots,\RO_z$ such that $\RO_i(\cdot) = \RO(i \mid\mid \cdot)$ and $i \in \{0,1,\ldots,z\}$ is represented using string of $O(n+s)$ bits.
	\item Instantiate counter $c' \leftarrow 0$.
	\item For $i \in \{1,\ldots,z\}$:
	\begin{enumerate}
		\item Sample oracle $\Or_i \leftarrow \mathcal{A}^{\RO_i}(\Or(S))$ such that $\Or_i(S) = \Or(S)$ using Lemma~\ref{lem:oracle_resample} with parameter $\zeta = 6n$. If the algorithm fails to sample crypto oracle $\Or_i$, then output $\reconout \leftarrow \bot$.
		\item Compute $(\st, \omega_{\client}) \leftarrow \RO_0(i)$.
		\item Compute $(\ans \mid \tx_\client') \leftarrow \query^{\Or_i}_\client(\st, \qrys; \tx_\server)$ using $\omega_{\client}$ as the randomness for the client side of the $\query$ protocol.
		\item If $\tx_\client' = \tx_\client$:
		\begin{enumerate}
			\item If $c' = c$, output $\reconout \leftarrow \ans$.
			\item Increment $c' \leftarrow c' + 1$.
		\end{enumerate}
	\end{enumerate}
	\item Output $\reconout \leftarrow \bot$.
    \end{itemize}
\end{enumerate}

\begin{theorem}\label{thm:comm-computable}
Suppose there is a PIR $\Pi = (\init, \setup^{\Or}, \query^{\Or})$ with
an oracle-free encoding (\defref{def:oracle-free}) supporting batch queries of size $k$
for databases of size $n$ with client state of $s$ bits, and the communication transcript is $t$ bits.
Furthermore, suppose $\Pi$ is $\varepsilon$-correct and $\delta$-private and
communication-determined.

For any fixed query sequence $\qrys'$, our construction
is a $(n, k, t, 3s)$-dual PIR with $\varepsilon'$-correctness satisfying $\varepsilon' \ge \varepsilon - 2\delta - 2^{-s+1} - 2^{-n+1}$.
\end{theorem}
\begin{proof}[Proof sketch]
The proof of this theorem follows the same line of reasoning as
Theorem~\ref{thm:dualpirmain} (in particular, the oracle-free encoding
assumption is used in the same places; see Claim~\ref{clm:replay}). Each of
Lemmas~\ref{lem:p},~\ref{lem:hint-counter},~\ref{lem:output-dist},
\ref{lem:same-query}, and~\ref{lem:dualpiradv} can be analogously
shown for the construction outlined above. Most of these do not change
significantly, such as bounding the counter, number of oracles needed, and
switching from one query to another. Lemma~\ref{lem:output-dist} changes
the most as we now must argue that the output distribution of the two
games are the same where $\bS$ is now computed from $\btx$ instead of output
by $\advice$. However, this change adds a single step to the proof to establish
that if $\btx$ is distributed identically to $\btx'$ then $(\btx,\bS)$ is
identically distributed to $(\btx',\bS')$. This follows from the straightforward
observation that $\bS \gets \mathsf{SvrQ}^{\Or}(\btx)$ is deterministic
and correct with probability 1.
\end{proof}

Next, we state the corollary, which is most analogous to
Theorem~\ref{thm:main-lower-bound} and follows from the above
Theorem~\ref{thm:comm-computable}
combined with Lemma~\ref{lem:dpir}.

\begin{corollary}\label{cor:comm-computable}
Relative to any crypto oracle $\Or$, any single-server
PIR with preprocessing scheme with an oracle-free encoding that is $\varepsilon$-correct, $\delta$-private,
and communication-determined
with $s$-bit client storage,
supporting $k = \Theta(s)$ queries
against computationally unbounded adversaries making $\poly(\lambda)$
queries to $\Or$, requires total communication $\Omega(n)$.
This holds for any choice of $s = \omega(\log(n + \lambda))$,
$\varepsilon = \Omega(2^{-s})$ and $\delta \le \varepsilon / 3$.
\end{corollary}

\heading{Comparison to ``database oblivious/independent'' PIR.}
Prior work~\cite{C:IshShiWic24} considered a definition of
``database-oblivious'' PIR. This is arguably a more restricted definition 
than Definition~\ref{def:determined} above though technically incomparable. 
Database-oblivious PIR requires both that the database is not encoded
and that all client messages depend only on their randomness and not on server
responses. Our definition makes neither of these restrictions. Instead, we only
require that the server's use of cryptography must be fully determined by the
transcript itself and therefore cannot depend on the entire database.

In fact, under our definition the server crypto oracles may depend on
the database if they appear in the communication transcript (so, they do not
need to be completely database-independent). In fact, this lower bound may
be generalized to the case where the server crypto oracles
depend on a sublinear subset of the database or to a function of the database
with output size $o(n)$. To our knowledge, this
generality does not widen the applicability of our lower bound
as, all known PIR schemes either make the server crypto
oracles completely independent of the database (such as those based on OWFs)
or the server crypto oracle queries depend on the entire database (such
as schemes using FHE).

\subsection{PIR with Perfect Privacy}
\label{sec:comm-perfect}
Our next communication lower bound applies to any PIR which is
\emph{perfectly private} against all adversaries with $\poly(\lambda)$ queries
to a specific crypto oracle $\Or$ and whose client-side $\query$
algorithm is deterministic given its state, its queries, the incoming server
messages, and $\Or$ (e.g., the state contains a seed $\mathsf{sd}$ and the
client derives its query randomness as
$\Or(\mathsf{sd} \mid\mid \cdot)$).\footnote{For example, the schemes
of~\cite{CCS:RenMugSun24,EC:HPPY25} derive all query randomness from the
oracle and key material in their state. More generally, any scheme can store
its query randomness in its state, preserving privacy exactly with a larger
state.}
Recall that perfect privacy is equivalent to our
traditional privacy when all adversaries have advantage $\delta=0$.
This assumption simplifies many of the
difficulties that arise in the proof of Theorem~\ref{thm:main-lower-bound}.
Despite this simplification, perfect privacy is not as strong of an assumption
as one might originally think. In particular, we require perfect privacy
only with respect to an \emph{idealized crypto oracle}. For example,
recent works, like~\cite{EC:HPPY25}, are perfectly secret relative to
a random oracle, derive their client query randomness from the random
oracle and key material in their state, and therefore this bound applies to
them. We also note that, in contrast to our previous lower bounds, this
result requires no oracle-free encoding restriction on $\init^{\Or}$, as the
dual PIR below never resamples crypto oracles.

The simplified dual PIR for this result is defined as follows, where
all algorithms have access to the crypto oracle $\Or$ but share no other
randomness.

\begin{enumerate}
\item $\advice_{\qrys'}^{\Or}(\db)$: %
    \begin{itemize}
    	\item Run $\edb \gets \init^{\Or}(1^{\secpar}, \db)$.
        \item Run $(\st;\bot)\gets\setup^{\Or}(1^\lambda ;\edb)$.
        \item Run $(\ans;\bot\mid \tx)\gets\query^{\Or}(\st,\qrys';\edb)$.
		\item Parse $\tx = (\tx_\client, \tx_\server)$
        \item Output $\adviceout \leftarrow \tx_\server$.
    \end{itemize}
\item
$\hint^{\Or}(\db, \qrys, \adviceout)$:
    \begin{itemize}
	\item Sample $\st \getsr \bits^s$
	\item Output $\hintout \leftarrow \st$.
    \end{itemize}
\item
$\recon^{\Or}(\qrys,\adviceout,\hintout)$:
    \begin{itemize}
	\item Parse $\adviceout = \tx_\server$.
	\item Parse $\hintout = \st$.
	\item Compute $(\ans \mid \tx_\client') \leftarrow \query^{\Or}_\client(\st, \qrys; \tx_\server)$,
    where the client derives its randomness from $\Or$ and the seed in $\st$.
	\item Output $\reconout \leftarrow \ans$.
    \end{itemize}
\end{enumerate}

Notably, this dual PIR only requires that the advice protocol output the
server-side transcript, meaning that we obtain a stronger lower bound in the
sense that it is a lower bound on the \emph{download} communication online,
regardless of how much upload communication the PIR uses. This also allows
our hint algorithm to avoid searching for consistent states and instead just
try and guess the state that was used in the initial phase directly. The
following theorem and subsequent corollary remove many of the more
technical parts of our main result.

\begin{theorem}\label{thm:dpir-perfect-privacy}
Suppose there is a PIR $\Pi = (\init^{\Or}, \setup^{\Or}, \query^{\Or})$ whose client-side $\query$ algorithm is deterministic given its state as above, supporting batch queries of size $k$
for databases of size $n$ with client state of $s$ bits, and the
download communication transcript is $t$ bits.
Furthermore, suppose $\Pi$ is $\varepsilon$-correct and perfectly private.

For any fixed query sequence $\qrys'$, our construction
is a $(n, k, t, s)$-dual PIR with $\varepsilon'$-correctness satisfying $\varepsilon' \ge \varepsilon \cdot 2^{-s}$.
\end{theorem}
\begin{proof}
To show the above dual PIR is correct, we can argue much more directly compared
to our previous Theorem~\ref{thm:main-lower-bound}. In particular, we
first observe that perfect privacy implies that for any database $\db$
and two query sets $\qrys,\qrys'$, the probability that the above dual
PIR outputs the correct answer $\db[\qrys]$ is identical whether using $\btx$ output
from $\advice_{\qrys'}^{\Or}$ or from $\advice_{\qrys}^{\Or}$.
Indeed, a privacy adversary given $(\tx, \edb)$ can sample a
random state $\st \getsr \bits^s$, run the deterministic client-side replay
$\query^{\Or}_\client(\st, \qrys; \tx_\server)$ itself, and check whether
the output equals $\db[\qrys]$; this adversary makes $\poly(\secpar)$
queries to $\Or$, its execution with challenge bit $b$ is exactly the dual
PIR with $\advice$ parameterized by $\qrys_b$, and applying $\delta = 0$ to
it and its complement forces the two probabilities to be equal. (This is a
simplification of the argument given in the proof of
Lemma~\ref{lem:dualpiradv}.) In particular, we consider
the following two games:
\begin{enumerate}
    \item $\Game_0$: The dual PIR correctness game.
    \item $\Game_1$: The dual PIR correctness game replacing
    $\advice_{\qrys'}^{\Or}(\bdb)$ with $\advice_{\bqrys}^{\Or}(\bdb)$.
\end{enumerate}

Now, consider that the $\ans_{\advice}$ computed in $\advice_{\qrys}^{\Or}$ is
computed just as in the regular PIR game and therefore the probability
that it is equal to $\db[\qrys]$ is at least $\varepsilon$ by assumption. Further,
notice that $\query^{\Or}_{\client}$ is deterministic given
the state $\st$ (which contains the seed determining the client's randomness),
queries $\qrys$, oracle
$\Or$, and transcript $\tx_{\server}$. Further, conditioned on the event
that the state sampled in $\hint$, $\st_{\hint} = \st_{\advice}$, the output
of the $\recon$, $\ans_{\recon} = \ans_{\advice}$.
Let $p_i$ denote the probability that the event
\[
    \recon^{\Or}(\bqrys,\badviceout,\bhintout) = \bdb[\bqrys]
\]
occurs in $\Game_i$, then we have
\[
    p_0 = p_1 \ge \Pr[\st_{\hint} = \st_{\advice}] \cdot
    \Pr[\ans_{\recon} = \bdb[\bqrys] \mid \st_{\hint} = \st_{\advice}] \ge  \varepsilon\cdot 2^{-s},
\]
where the probabilities on the right are taken in $\Game_1$,
which establishes the claim that the proposed dual PIR is correct
with probability
$\varepsilon' \ge \varepsilon\cdot  2^{-s}$.
\end{proof}

We note that the proof degrades gracefully for $\delta$-private
schemes: the adversary above then has advantage
$|p_0 - p_1|/2 \le \delta$, yielding a dual PIR with
$\varepsilon' \ge \varepsilon \cdot 2^{-s} - 2\delta$. Hence, the corollary
below also holds for $\delta$-private schemes with
$\delta \le \varepsilon \cdot 2^{-s-2}$, though this tolerance is
necessarily exponentially small.

Next, we state the corollary, which is most analogous to
Theorem~\ref{thm:main-lower-bound} and follows from the above
Theorem~\ref{thm:dpir-perfect-privacy}
combined with Lemma~\ref{lem:dpir}.

\begin{corollary}\label{cor:perfect-privacy}
Relative to any crypto oracle $\Or$, any single-server
PIR with preprocessing scheme with a deterministic client-side $\query$ algorithm as above that is $\varepsilon$-correct and perfectly private,
with $s$-bit client storage,
supporting $k = \Omega(s)$ queries
against computationally unbounded adversaries making $\poly(\lambda)$
queries to $\Or$, requires total download communication $\Omega(n)$.
This holds for any choice of $s = \omega(\log(n + \lambda))$
and $\varepsilon = \Omega(2^{-s})$.
\end{corollary}

\section{Symmetric PIR without Online Public-Key Cryptography}\label{sec:spir}
In this section, we study symmetric PIR 
(SPIR) with client-side preprocessing focused on online efficiency,
which to our knowledge has only been studied in one concurrent work
\cite{EPRINT:LWZW25}.\footnote{The work of \cite{EPRINT:LWZW25}
requires two servers. Moving our main construction's 
offline computation to
a second server can match their construction's parameters
up to polylog factors, but will only obtain computational
rather than statistical security.}
We begin with an overview of how this direction relates to prior works on 
SPIR, and we extend our lower bound to this setting. 
We then give both simple and complex constructions of SPIR with
no online public-key operations,
the best of which match the online communication and time of the best PIR constructions.

Recall that SPIR is equivalently formulated as efficient 
$1$-out-of-$n$ Oblivious Transfer (OT). 
And, it is well known from~\cite{STOC:ImpRud89} that
OT cannot be built from one-way functions in a blackbox way. 
Regardless of the efficiency limitations,
this separation also applies to any SPIR with preprocessing as well,
since we know that across the whole protocol an OT has been performed.
This seemingly makes it impossible to base SPIR on one-way functions.

However, there is still an asymmetry between the offline and online phases
of a PIR/SPIR protocol. 
And, since using public-key operations is generally more expensive than symmetric-key ones (like OWFs),
we would ideally only rely on heavy computation in the offline phase.
This is a problem that has been theoretically and practically considered in the
context of OT-extensions,
which allow two parties to run a few OTs offline and then
use symmetric-key primitives online to quickly perform even more OTs~\cite{STOC:Beaver96a,C:IKNP03,CCS:ALSZ13,C:KolKum13,EC:ALSZ15}.

\subsection{Formal model}
The primary difference from other 
versions of secure two-party computation (2PC) with preprocessing is that
we wish to build  SPIR with \emph{no per-client storage}
and \emph{no online public-key operations}.

\heading{Shared offline/online server randomness.}
In our SPIR syntax, the server obtains a random string
$R$ that is the same between the offline and online phases.
As we show later in \secref{sec:spir-limits}, this shared randomness is
necessary. In particular, any perfectly private
SPIR with preprocessing supporting $k$ queries
must allow the server to maintain at least $k$-bits of shared randomness
between the offline and online phases.\footnote{For computational privacy,
the servers can share $O(\secpar)$ bits of randomness which
can be expanded with a pseudorandom generator.}
If we allow randomized encodings, we could have instead fit this into an
$\init$ procedure, but for conceptual simplicity, we remove the $\init$
function. Fortunately, our lower bounds apply to any PIR whose encoding
procedure makes no crypto oracle queries, and the encoding induced by a
SPIR (the database $\db$ stored together with the plain uniform randomness
$R$) requires no oracle access. Hence, our lower bounds will similarly apply to any SPIR with preprocessing.

\heading{Per-client storage.}
Beyond a shared random string or encoding, we could allow the server to store
its own state from interacting with the client in the offline phase.
However, in a setting with many clients this
would be impractical as the server storage grows with the size of the user base.
Additionally, if both the client and server are allowed a state,
then the (computational version) of the
problem is exactly efficient oblivious transfer with preprocessing (or 2PC with preprocessing
more generally), which has been studied in prior works~\cite{C:CDGGMP17,C:DotGar17,EC:DKLLMR23}.
But, those works do not capture the same asymmetric preprocessing
between the sender and receiver that is present between the client and server in SPIR.

\heading{Adversarial model.}
We only prove our constructions secure against semi-honest servers and
clients.
We make this simplifying
assumption because our schemes are primarily proofs of concept,
relying heavily on generic 2PC. 
Using some generic 2PC which is secure against malicious adversaries,
we could have maliciously secure SPIR with client-side preprocessing (for at least our toy constructions);
however, there are known barriers for generic maliciously-secure 2PC
(see \cite{STOC:GolMicWig87,EC:CanKusLin03,EC:LinPin07,C:HazVen16} and the references therein).
Therefore, we leave maliciously-secure SPIR with preprocessing as an interesting
open problem for future work.

\begin{definition}[Symmetric Private Information Retrieval]
\label{def:spir}
A $k$-batch 
\emph{symmetric private information retrieval with preprocessing scheme (SPIR)}
is a pair of efficient two-party protocols $\spir=(\setup,\query)$ parameterized
	by a database size $n$, batch size $k$, download $t$,
	state size $s$, and randomness $r$ with the following syntax:\footnote{Note that $k,t,s,r$
		may be functions of $\lambda$ and $n$; however, we omit this
		notation for simplicity.}
	\begin{itemize}
		\item $\setup(1^\secpar;1^\secpar,\db,R)\to(\st;\bot\mid\txoff)$,
        where the client and server receive as input the security parameter $\secpar$,
        and the server additionally receives a database $D\in\bits^n$ 
        and private randomness $R\in\bits^r$.
		At the end of the protocol,
        the client outputs a state $\st\in\bits^s$,
        the server receives no output,
        and $\txoff$ is the transcript of the interaction.
        
		\item $\query(\st,\qrys;\db,R)\to(\ans;\bot\mid\txon)$,
        where the client takes its precomputed state $\st$ and 
        an ordered set of queries
		$\qrys\in [n]^k$, and the server receives 
        a database $\db$ and private randomness $R\in \bits^r$.
		At the end of the protocol, the client outputs $\ans \in \bits^k$
		and the server outputs nothing.
        We denote the client messages $\tau_u \in \bits^*$ and the
		server messages $\tau_d\in \bits^t$, so the transcript of this entire protocol
		is $\txon=(\tau_u,\tau_d)$.
	\end{itemize}
	
	A SPIR is \emph{$\varepsilon$-correct} if for all sufficiently large $\lambda$,
	tuples $\qrys\in [n]^k$, and $\db \in \bits^n$,
	\[
	\Pr%
	\left[\query(\bst, \qrys ; \db, \bR) = (\db[\qrys]; \bot) ~\bigg|~
	    \substack{
	    	\bR \getsr \bits^r\\
			(\bst;\bot\mid\txoff) \getsr \setup(1^{\lambda} ; 1^{\lambda},\db,\bR)
		}
	\right]
	\ge \varepsilon.
	\]
	
	A SPIR is \emph{$\delta$-private} if for all sufficiently large $\lambda$,
	tuples $\qrys_0,\qrys_1 \in [n]^k$, efficient
	adversaries $\A$, and $\db \in \bits^n$,
	\[
	\Pr
	\left[\A(1^{\lambda},\db,\bR,\btxoff,\btxon) = b~\bigg|~
	\substack{
		\bR \getsr \bits^r ~;~ b \getsr \bits\\
		(\bst;\bot\mid \btxoff) \getsr \setup(1^{\lambda} ; 1^{\lambda},\db,\bR)\\
		(\bans;\bot \mid \btxon) \getsr  \query(\bst,\qrys_b;\db,\bR)
	}\right]
	\le \frac{1}{2} + \delta.
	\]
	
	A SPIR is \emph{$\gamma$-data private}
	if for all sufficiently large $\lambda$,
	tuples $\qrys \in [n]^k$,
	databases $\db_0, \db_1 \in \bits^n$ with $\db_0[\qrys] = \db_1[\qrys]$,
	efficient adversaries $\A$,
	\[
	\Pr
	\left[\A(1^{\lambda},\qrys,\bst, \btxoff, \btxon) = b~\bigg|~
	\substack{
		\bR \getsr \bits^r ~;~b \getsr \bits\\
		(\bst;\bot \mid\btxoff) \getsr \setup(1^{\lambda} ; 1^{\lambda},\db_b,\bR)\\
		(\bans;\bot \mid\btxon) \getsr  \query(\bst,\qrys;\db_b,\bR)
	}\right]
	\le \frac{1}{2} + \gamma.
	\]
	
One may also consider SPIR relative to an oracle $\Or$, as
was done in Definition~\ref{def:pir}. We omit this notation to
avoid redundancy. Notably, security in such a model
only requires that adversaries $\A$ submit only polynomially-many
queries to $\Or$.

A SPIR scheme  is called \emph{perfectly} correct, private, and data private when
$\varepsilon=1$, $\delta=0$, and $\gamma=0$, respectively.
\end{definition}

\subsection{Limitations of SPIR without Public-Key Cryptography}
\label{sec:spir-limits}

In this section, we briefly discuss limitations of SPIR with preprocessing.
Obviously, the bounds in previous sections will still apply to SPIR with
preprocessing, since they apply to any PIR with an oracle-free encoding
procedure and the encoding induced by a SPIR
similarly makes no oracle queries (i.e., sharing
plain randomness for the server doesn't circumvent our
lower bounds).

Beyond these lower bounds, however, the introduction of
data privacy allows us to prove some limitations on the capability of SPIR
with no online public-key operations.
First, we show that any scheme with $s$-bit client storage may support at
most $k \le s$ queries in the random oracle model. Secondly,
we show that SPIR supporting $k$ queries must allow the server at least
$k$-bits of shared randomness between the offline and online phases.
These are presented in the context of perfectly private and
data private SPIR to simplify the proofs as they are meant just to better
explain the limitations of our own construction.

\heading{Limits on Number of Queries.}
Here, we show that one can use data privacy to simplify and potentially
improve over our general lower bounds in previous sections.
In particular, we show a simple proof that perfectly private and data
private SPIR with $s$ bits of preprocessing can only correctly query
$k \le s$.
Specifically,
Theorem~\ref{thm:spir} below shows that there is no 
SPIR relative to a random oracle $\RO$ which is perfectly private and
data private against adversaries with $\poly(\lambda)$ queries
to $\RO$, correct, and supports
$k>s$ queries — no matter how much communication is allowed.

Unfortunately, this result shows that it is impossible to construct SPIR
with preprocessing for an unbounded number of queries without introducing
online public-key operations and using blackbox OWFs.
Even if one were to allow $\Omega(n)$ communication!
This is a key difference from the PIR setting, where there
are constructions which can support an unbounded number of queries 
by using $\Omega(n)$ communication to continuously ``refresh'' hints without
public-key operations.
In our main construction, we show how to actually match this bound
up to $\polylog(\lambda,n)$ factors with sublinear per-query computation.

\begin{theorem}[Batch SPIR limitations]\label{thm:spir}
	Let $s < k \le n$ and $\varepsilon > 2^{s-k}$. In the random
	oracle model,
	there does not exist a $k$-batch SPIR scheme with $s$-bit client storage that is
	$\varepsilon$-correct, perfectly private, and perfectly data private.
\end{theorem}
\begin{proof}
	We could prove this result using our dual PIR paradigm,
	but in fact, we can provide the proof more directly.
	From perfect privacy and data privacy,
	we know that the
	messages between the client and server in the online phase are
	independent of the indices
	queried $\qrys$  and of the database $\db$. This means any party can
	sample a random
	variable $\txon$ from this distribution for a specific random oracle $\RO$
	without knowledge of either $\qrys$ nor $\db$.
	
	This gives a natural compression scheme for $k$ bits into $s$ bits given some common random
	string, which we interpret as a random oracle $\RO$, client query randomness
	and a sample of $\txon$ (for $\RO$, without knowledge of $\db$ nor $\qrys$).
	Specifically, a sender can set the bits of an otherwise
	random database $\db$ at a canonical set $\qrys$. Then, the sender will sample some state $\st$ 
	such that $\query^{\RO}(\st, \qrys) = \db[\qrys]$ when using the transcript $\txon$ for the server's
	replies. Finally, they send $\st$ ($s$ bits) to the other party. If no such state exists, then fail.
	
	Then, the receiver will run $\query^{\RO}(\st, \qrys)$ themselves to recover the $k$ bits.
	Here we assume that the client side of $\query$ runs with the same random coin tosses
	as the sender, which can be done
	without loss of generality by putting the client's randomness in the common random string.
	Notice that, if the underlying
	SPIR scheme is correct with probability $\varepsilon$, then this scheme will also be correct
	with probability at least $\varepsilon$ (over the choice of client randomness, $\txon$, and $\RO$).
	
	Note there are at most $2^{s}$ possible outputs of the first phase. For any fixed randomness,
	the receiver can only output at most $2^{s}$ strings, so at most $2^{s}$ of the $2^{k}$ possible
	values of $\db[\qrys]$ are recovered correctly. Since correctness must hold for every database $\db$,
	averaging over all $2^k$ values of $\db[\qrys]$ shows that the scheme is correct with probability
	at most $2^{s}/2^{k} = 2^{s-k}$.
\end{proof}

\noindent{\bf Necessity of Shared Randomness.}
We observe that any SPIR with client-side preprocessing can be transformed into a
(non-simultaneous) 2-server
SPIR by performing the preprocessing at one server and the only queries at the other.
Furthermore, in \cite{STOC:GIKM98}, a short argument is given that any multi-server SPIR system,
with information-theoretic data privacy, must have shared randomness. 
We review and strengthen this argument to show that any 2-server
SPIR requires some minimum amount of randomness, 
with an argument that fully relativizes.
Therefore, any SPIR with preprocessing construction necessitates some amount of shared randomness to avoid
public-key operations in the online phase.
We start by extending the argument in~\cite{STOC:GIKM98} as follows:

\begin{lemma}\label{lem:spir}
	There is no $2$-server SPIR (without preprocessing) which is perfectly correct, private, and data private
	and supports $k < n$ queries where the random string $\mathbf{R}$ shared by the two servers (beyond the database $\db$) has entropy $H(\mathbf{R}) < k$.
\end{lemma}
\begin{proof}
	Fix any set $\qrys$ of $k$ queries, let $\bdb$ be a uniformly random database,
	let $\bomega$ denote the client's private randomness, and let $\btx_1$ and $\btx_2$
	denote the transcripts with the first and second server, respectively.
	Each server's messages are a function of its own transcript so far, $\bdb$, $\mathbf{R}$,
	and its private randomness, while each client message is a function of $\bomega$
	and the messages the client has received so far. The argument below holds
	even relative to an independent oracle $\Or$.

	First, we claim that each $\btx_i$ is independent of $\bdb$.
	By perfect privacy, the distribution of $\btx_i$ does not depend on the queried indices $\qrys$ and,
	by perfect data privacy, it depends on $\bdb$ only through $\bdb[\qrys]$.
	Hence, for any two databases that differ in a single position $x$, choosing a set of $k < n$ queries
	that avoids $x$ shows that $\btx_i$ has the same distribution under both databases, and
	chaining single-bit changes shows that the distribution of $\btx_i$ is the same for every database.
	Moreover, $\btx_1$ remains independent of $\bdb$ conditioned on $\bomega$: the probability of a
	transcript $\btx_1$ factors into a client part that depends only on $\bomega$ and a server part that
	depends only on $\bdb$, $\mathbf{R}$, and the first server's randomness, and the independence above
	shows that the server part does not depend on $\bdb$ for any transcript in the support.

	Second, by perfect correctness the client recovers $\bdb[\qrys]$ from $(\bomega, \btx_1, \btx_2)$,
	and $\bdb[\qrys]$ is uniform over $\bits^k$, so $I(\btx_1\btx_2 ; \bdb \mid \bomega) \ge H(\bdb[\qrys]) = k$.
	Combining this with the first claim,
	\begin{align*}
		k &\le I(\btx_1\btx_2 ; \bdb \mid \bomega)\\
		&= I(\btx_1 ; \bdb \mid \bomega) + I(\btx_2 ; \bdb \mid \btx_1, \bomega)\\
		&= I(\btx_2 ; \bdb \mid \btx_1, \bomega).
	\end{align*}

	Finally, we show that $I(\btx_2 ; \bdb \mid \btx_1, \bomega) \le H(\mathbf{R})$, which completes the proof.
	Let $\mathbf{m}_2$ and $\mathbf{a}_2$ denote the client's messages and the second server's responses in $\btx_2$.
	Since $\mathbf{m}_2$ is determined by $(\bomega, \btx_1)$ and the earlier responses in $\mathbf{a}_2$,
	we have $I(\btx_2 ; \bdb \mid \btx_1, \bomega) = I(\mathbf{a}_2 ; \bdb \mid \btx_1, \bomega)$.
	We give the argument for a single round of interaction with the second server, so that $\mathbf{m}_2$
	is a function of $(\bomega, \btx_1)$; the general case follows by applying the same argument to each
	response in turn using the chain rule. We have
	\begin{align*}
		I(\mathbf{a}_2 ; \bdb \mid \btx_1, \bomega)
		&\le I(\mathbf{a}_2 ; \bdb, \btx_1, \bomega \mid \mathbf{m}_2)\\
		&= I(\mathbf{a}_2 ; \bdb \mid \mathbf{m}_2) + I(\mathbf{a}_2 ; \btx_1, \bomega \mid \bdb, \mathbf{m}_2)\\
		&\le I(\mathbf{a}_2 ; \mathbf{R} \mid \bdb, \mathbf{m}_2)\\
		&\qquad + I(\mathbf{a}_2 ; \btx_1, \bomega \mid \bdb, \mathbf{m}_2, \mathbf{R})\\
		&\le H(\mathbf{R}).
	\end{align*}
	The first inequality uses that $\mathbf{m}_2$ is determined by $(\btx_1, \bomega)$, and the second
	uses that $I(\mathbf{a}_2 ; \bdb \mid \mathbf{m}_2) = 0$, since the second server's view
	$(\mathbf{m}_2, \mathbf{a}_2)$ is independent of $\bdb$ by the first claim.
	The last inequality uses that $I(\mathbf{a}_2 ; \btx_1, \bomega \mid \bdb, \mathbf{m}_2, \mathbf{R}) = 0$,
	because given $\mathbf{m}_2$, $\bdb$, and $\mathbf{R}$, the response $\mathbf{a}_2$ depends only on the
	second server's private randomness, which is independent of everything else.
	Therefore, $H(\mathbf{R}) \ge k$.
\end{proof}

Finally, we prove that this implies shared randomness
in the single-server SPIR with preprocessing setting. Roughly speaking, this follows as it is
possible to convert any single-server SPIR with preprocessing into a two-server SPIR.

\begin{theorem}
	There is no single-server SPIR with preprocessing which is perfectly correct, private, and data private
	and supports $k < n$ queries where the server's random string $\mathbf{R}$ shared between the offline and online phases has entropy $H(\mathbf{R}) < k$.
\end{theorem}
\begin{proof}
We can convert any single-server SPIR with preprocessing into a two-server SPIR as follows.
We have one server to execute the offline phase and compute the $s$-bit hint. We can do this in plaintext since
this server will not receive the online query transcript and the preprocessing is done independent
of the queried indices. Using the $s$-bit hint, the client runs the query
algorithm with the second server where query privacy holds as the second server does not observe
the $s$-bit hint. For data privacy, we note that the client can only learn one entry per query as the
original single-server SPIR scheme obtained data privacy. The only difference is that the client receives
the $s$-bit hint in plaintext (as opposed to potentially using more complex cryptography such as MPC). Note,
the client learns less information about the database in this two-server SPIR scheme as it only learns the $s$-bit hint
output of the offline preprocessing phase. As a result, we obtain a two-server SPIR with the same correctness that is
both private and data private while supporting $k$ queries.
Finally, we can apply \lemref{lem:spir} that rules out the existence of such a single-server SPIR.
\end{proof}

\subsection{Toy Constructions}

Before giving our main construction in Section~\ref{se:mainspir},
we go over two simple constructions
of SPIR using no public-key cryptography in the online phase.
For these constructions, we will assume that we have access to some generic
secure two-party computation (2PC) \cite{FOCS:Yao86,STOC:GolMicWig87},
which can semi-honestly compute the functionality that our $\setup$ protocol calls for.
As these constructions are proofs of concept, we leave it to future work to optimize the
setup functionality for a specific scheme.

\heading{$n$-bit client storage, $O(k\log n)$-bit communication.}
Unlike PIR, notice that SPIR is non-trivial even when the client has $n$ bits of storage
because the client cannot store the database in plaintext, as this would violate data privacy.
However, we can still construct client-side preprocessing SPIR as follows:
\begin{itemize}
	\item $\setup(1^{\secpar} ; 1^\secpar,\db,R)$:
the client privately chooses a random permutation $\pi: [n] \to [n]$ 
and the client and server use 2PC so that the client learns
	and outputs $\st = \db \oplus (\pi \circ R)$ and the server learns nothing. 
Here, we use $n$-bit of randomness $R$ and the
	notation $\pi\circ R = (R[\pi(1)],\ldots,R[\pi(n)])$.
	
\item $\query(\st, \qrys ; \db, R)$: on input $\qrys = (\qry_1,\ldots,\qry_k)$ consisting of $k$ distinct indices, the client sends
	$\tau_u = (\pi(\qry_1),\ldots,\allowbreak\pi(\qry_k))$ and
	the server responds with $\tau_d = R[\tau_u]$. Then, the client outputs
	$\ans = \st[Q] \oplus \tau_d$.
\end{itemize}
The proofs of correctness, privacy, and data privacy are straightforward. Note that the indices in a batch must be distinct, as a repeated index would reveal the equality pattern of the batch to the server.

\heading{$O(k)$-bit client storage, $O(kn)$ communication.}
Next, we consider SPIR at the other end of the spectrum and show %
a SPIR with preprocessing which has optimal client storage but huge online
communication.
\begin{itemize}
	\item $\setup(1^{\secpar} ;1^\secpar, \db, R\in \bits^k)$:
the client privately chooses a random
	PRF key $\key\in\bits^\lambda$ for a pseudorandom function $\prf(\key, \cdot) : [k] \to \bits^n$
	and the client and server use 2PC, so that the client learns
	the $k$-bit string
	$C$ with $C[i]=\langle\prf(\key,i),\db\rangle\oplus R[i]$,
    where $\langle\cdot,\cdot\rangle$ denotes the inner product.
	Finally, the client outputs $\st = (\key, C)$.
	
	\item $\query(\st, \qrys ; \db, R)$: on input $\qrys = (\qry_1,\ldots,\qry_k)$, the client parses
	$(\key, C) \gets \st$ and uses each $e_{\qry_i} = 0^{\qry_i-1}\| 1 \| 0^{n - \qry_i}$ to send
	$\tau_u = (\prf(\key, i) \oplus e_{\qry_i})_{i\in [k]}$ to the server, and
	the server responds with $\tau_d = R \oplus (\langle \tau_u[i], \db \rangle)_{i\in [k]}$.
	Then, the client outputs $\ans = C \oplus \tau_d$.
\end{itemize}
We omit the proof that the scheme is correct, private, and data private, but it is clear by construction
that the sum of the inner products of vectors that are off by one will reveal the correct entry.

\heading{Existing 2-server SPIR.}
We additionally observe that we can convert any
information-theoretic 2-server SPIR into a an offline-online SPIR without any public-key
operations online, assuming the existence of generic 2PC. In particular, a client could
obtain $k$ independent responses from the first server under 2PC in the offline phase. Then,
the client could make online second server queries (which do not require public-key operations,
as they're information theoretically secure).  This scheme would require $O(k\cdot C_1)$
state and $k\cdot C_2$ online communication, where $C_i$ is the communication between the
client and the $i$-th server for single query. However, this doesn't leverage the asymmetry
between the offline and online phases — since most schemes to date have $C_1 = C_2$.
Also, as far as the authors are aware,
the best information-theoretic SPIR has
$C_1 = C_2 = \tO(n^{1/3})$~\cite{STOC:GIKM98}.

If we instead consider computational 2-server SPIR, we could adapt approaches using 2-server
computational PIR, e.g., those based on distributed point functions \cite{EC:GilIsh14}.
This approach is similar to our second toy construction, as we could one-time pad queries to
the first server and then query the second server, who additionally adds the same pad to their
responses. This will achieve better concrete efficiency than our toy construction above reducing
the communication to be $k\cdot \polylog(n)$.

However, all of these approaches require linear online server time. 
And, as it was done in
\cite{EC:CorKog20} for PIR, we wish to show how to use offline preprocessing to achieve
online \emph{sublinear query time} for SPIR.

\subsection{SPIR with Sublinear Online Computation}
\label{se:mainspir}
Our main construction adapts the recent works of \cite{CCS:RenMugSun24,EC:HPPY25}
to the SPIR setting.
This construction achieves both online sublinear communication and
computation per query — substantially improving over the efficiency of our
simple constructions above. Additionally, it
supports the maximum number of queries possible per offline phase (see Theorem~\ref{thm:spir})
as well as it meets an optimal state/computation trade-off from 
prior work~\cite{EC:Yeo23}, up to logarithmic factors.
Our approach is not entirely generic, but it could possibly be applied to other PIR with
client-side preprocessing; however, we critically use that the scheme does not require
running parallel repetition to achieve a negligible failure rate.

For our high-level overview, we recall the way that recent PIRs with preprocessing schemes
work at high level
\cite{EC:CorKog20,EC:CorHenKog22,CCS:RenMugSun24,SP:ZPZS24,EC:HPPY25,EC:WanRen25}.
We will then, give an overview of the difficulties and solutions that come up when augmenting these
schemes with data privacy. Finally, we will give our more detailed construction.

\heading{PIR construction review.}
Simplifying a bit, the way that these schemes preprocess for their PIR queries is to select
$h = \tO(\sqrt{n})$ random sets $S_1,\ldots,S_{h}$ each of size (about) $\sqrt{n}$.
Then, the client will learn the XOR of each $S_i$ for the database, i.e.,
$p_i \gets \bigoplus_{x\in S_i} \db[x]$, and finally the client then stores each $(S_i,p_i)$.

In the online phase, to query index $\qry$, the client will find some set $S_i$ with $\qry\in S_i$,
which exists with high probability based on the choice of $h$. Then, the client sends
$S' \gets S_i \setminus \qry$ to the server, who returns $p' \gets \bigoplus_{x \in S'} \db[x]$ to the client.
Finally, the client recovers $\db[\qry] \gets p' \oplus p_i$ correctly. Notice that this scheme
requires large client storage; however, we can easily compress the representation of the
random sets with a PRF, which we do below.

As written, the above has $\Theta(1/\sqrt{n})$ privacy for single query
and no clear way to support multiple queries. But, all schemes to date operate in this overall paradigm of
preprocessing by finding the XOR of small random sets, choosing a set containing their query, and
recovering the database element by asking the online server for the XOR of a modified set. The details
of how they support multiple queries and achieve stronger privacy are scheme specific, and we will
only illuminate these details later for our specific scheme.

\heading{Transforming PIR to SPIR.}
In order to transform this scheme from a PIR to a SPIR, we first need to modify the state so
that it reveals no information about the database $\db$. Since the sets $S_1,\ldots,S_{h}$ are
already generated independently of the database, we just need to mask the parity bits
$p_1,\ldots,p_h$. For all schemes like this, we can just add on a single one-time pad for each
of these, so that they are actually computed as $p_i \gets R[i] \oplus \bigoplus_{x\in S_i} \db[x]$
for a uniform and independent $R[i]$. Additionally, symmetric PIR restricts \emph{how}
we can run the offline preprocessing, which is often either done by streaming or homomorphic
encryption. So, we assume that the entire functionality of the streaming is run inside of a
generic 2PC protocol, which provides privacy of the database (since the state is now
independent of the database).

From this modification, it is straightforward to modify the online phase. All that is required
is that the client, in addition to sending the set $S'$, sends the index $i$ for the hint that
they used for the query. Then, the server can return $p' \gets R[i] \oplus \bigoplus_{x \in S'} \db[x]$,
and the client can learn $\db[\qry] \gets p_i \oplus p'$ as before.
This transformation will provide data privacy, because the client only downloaded one entry of
information! So, even a malicious client will be unable to infer any more than a single database
entry.

\heading{Dealing with complications.}
The three main details left out of the above transformation are: (1) how to compress set
representations, (2) how to achieve negligible privacy, and (3) how to support multiple queries.
In the next subsection, we give the details on how these are overcome with specific pseudocode.
Here we provide the sketch of how we can lift the same ideas that were used in recent
works \cite{CCS:RenMugSun24,EC:HPPY25} from PIR to SPIR.

First, we modify the sets that are used for the hints. Instead of choosing random
sets of size $\sqrt{n}$, we break up $[n]$ into $\sqrt{n}$ blocks of size $\sqrt{n}$ and
restrict our sets so that that have a at most a single element from each block. We do
this by choosing $\sqrt{n}/2 + 1$ random blocks and sampling a random offset (in $[\sqrt{n}]$)
for each one. These elements constitute a hint set $S$ and require us to store one PRF key for
the partition sampling and one invertible PRF key per block to sample offsets, as was done in
\cite{EC:HPPY25}.\footnote{As observed in prior works, it is only necessary
for the client to choose a single master key for the entire state — from which they can derive
other implicit keys.} In the offline phase, the client will learn the (one-time padded) XOR of these
sets just as before.

Now, to query the index $\qry$, the client again finds a hint set $S_i$ which contains $\qry$.
The client will remove $\qry$ from $S_i$ to get $S'$ and sample fresh random offsets for
the blocks which are not covered by $S'$ (including the block that $\qry$ belongs to).
Then, the client will upload all of these offsets ordered by the blocks, together with
the blocks covered by $S'$ and the index $i$. %
Then, the server will send back the XOR of the values in these blocks as $p_1'$ in addition to
the XOR of the values in the other blocks as $p_2'$,  each padded with $R[i]$. Then, the
first of which can be used with the
precomputed parity to learn $\db[\qry] \gets p_1' \oplus p_i$.

Unfortunately, as described, the above as two issues. First, it is not private if we reveal the
blocks that $S'$ covers to the server, since it allows the server to learn which blocks our query
was certainly not in. Additionally, if we pad both $p_1'$ and $p_2'$ with the same value, then
the data privacy is no longer preserved! To fix both of these problems, our scheme will
with probability $1/2$ send the blocks that are not covered by $S'$ rather than the blocks
covered by $S'$, which is effectively just permuting whether $p_1'$ or $p_2'$ is the relevant
value to recover $\db[\qry]$.

Additionally, we want to modify the online server so that it pads the $p_1'$ and $p_2'$ with
different bits, so that
the irrelevant $p_j'$ is uniformly random and unknown to the client but that
the relevant $p_j'$ needs to remain hidden to the server for the client's privacy.
To handle this tension, we can modify preprocessing so that, for each hint, the client privately
selects whether it will send the real or fake blocks. Then, in the 2PC functionality, we have
$p_i$ set so that it is padded with either the left or right random bit depending on the choice
(so that the other bit remains independent of the client state).  Then online, the client will send
send the real or fake blocks based on its offline choice rather than a new random choice. This way
both the client's privacy and the server's data privacy are achieved.

\heading{Multiple queries.}
Fortunately, we can modify the techniques from prior work to achieve privacy even for adaptive 
queries, even though our definitions are written as batch queries. To do this, the key insight is
to specify a way to \emph{refresh hints}. At a high level, we cannot reuse hints while preserving
privacy and removing hints from the pool of hints biases the distribution of queries in a noticeable 
way. To rectify this, in addition to preprocessing the hints, we compute a collection of ``backup''
hints.

Following prior work, these backup hints have a similar but different structure than the normal
hints. Specifically, a backup hint consists of a set of size $\sqrt{n}/2$, based on a random partition
of size $\sqrt{n}/2$ blocks and $\sqrt{n}/2$ offsets and
the XOR of that set (padded with one of two random bits so that it gives no information
about the database).
After a client learns the value $\db[\qry]$, we can refresh the hint that
was used by finding a backup hint that does not cover the partition $\qry$.
Then, we can take the compressed set representation (i.e., PRF key) $S$, and store
$(S,\qry,p\oplus \db[\qry])$, which represents the hint set $S\cup \{\qry\}$ and its padded
parity as the refreshed hint. Then, this can be used and a normal hints and preserves the
distribution of the preprocessed hints (since we removed one hint conditioned on it containing
$\qry$ and added back a hint conditioned on it containing $\qry$).

The only remaining issue with our scheme, is that we have to set the index of the hint that
we send to the server. This is not an issue when making a single query, since it is distributed
identically for any query $\qry$. However, for multiple queries, the hint index leaks information
about the queries that have been made.\footnote{In particular, if a user queries for
an index $\qry$ which has already appeared in a previous query, then the index sent
to the server is more likely to be larger (since it is more likely that the user must use a backup hint).}
As a simple fix to this issue, the client can choose a secret random permutation $\pi$ over their
hints and backup hints. This way, every query just sends a distinct random index to the server
no matter what the client queries.

This is essentially the entire sketch of our scheme, which can continue processing
querying adaptively until it eventually runs out of backup hints. Next, we give a more detailed
description of the above and give the code to support any specified $k$ number of queries.

\subsection{Detailed Construction}
Here, we provide our detailed construction and theorem, with
the appropriate details. We
show how to handle $k$ adaptive queries for any $k$, taking
$h = k \cdot \polylog(\lambda,n)$, so that with high probability the online phase
can find valid hints for the online phase with all but negligible probability. This also involves
modifying the sizes of the hint sets (to $n/k$ rather than $\sqrt{n}$) for better efficiency.
We additionally assume that $t$ (and therefore $k$) divides $n$ without loss of generality,
since one could always pad the database up to the nearest multiple of $t$.

\begin{figure}
	\begin{center}
		\fbox{
			\begin{minipage}{4cm}
				\begin{tabbing}
					123\=123\=123\=123\=123\=\kill
					\underline{$\mathcal{F}(\key_p, \key_1,\ldots,\key_{n/k}, \pi, B; \db, R)$}\\
					Parse $(r_{i,0}, r_{i,1})_{i\in [h+t]} \gets R$\\
					For each $i = 1,\ldots,h$:~~ \texttt{// construct regular hints}\\
					\> $P \gets \prf_1(\key_p, i)$ ; $o_j \gets  \prf_2(\key_j, i)$ for each $j \in [n/k]$\\
					\> $X \gets \{(p-1)\cdot k + o_p : p\in P\}$\\
					\> $H[i] \gets \bigoplus_{x \in X} \db[x] \oplus r_{\pi(i),B[i]}$\\
					For each $i = h+1,\ldots,h+t$:~~\texttt{// construct backup hints}\\
					\> $P \gets \prf_1(\key_p, i)$ ; $o_j \gets  \prf_2(\key_j, i)$ for each $j \in [n/k]$\\
					\> $d \getsr P$ ; $P \gets P \setminus\{d\}$\\
					\> $X \gets \{(p-1)\cdot k + o_p : p\in P\}$\\
					\> $H[i] \gets (d, \bigoplus_{x \in X} \db[x] \oplus r_{\pi(i),B[i]})$\\
					Return $(H ; \bot)$
				\end{tabbing}
			\end{minipage}
		}
		\caption{Functionality for the setup phase of SPIR based on
			\cite{CCS:RenMugSun24,EC:HPPY25}.}
		\label{fig:hppy-func}
	\end{center}
\end{figure}

\begin{itemize}[leftmargin=*]
	\item $\setup(1^{\secpar} ;1^\secpar, \db\in \bits^n, R\in \bits^{2(h+t)})$: 
the client privately samples keys $\key_p\from\bits^\secpar$ 
for PRF $\prf_1(\key_p,\cdot):[h+t] \to {n/k \choose n/2k + 1}$,
keys $\key_1,\ldots,\key_{n/k}\from\bits^\secpar$ 
for invertible PRF\footnote{See \cite{EC:HPPY25} for definitions and construction of invertible PRFs.}
$\prf_2(\key,\cdot) : [h+t] \to [k]$,
	a secret permutation $\pi$ over $[h+t]$, and $B \getsr \bits^{h+t}$. Run
	secure MPC for functionality $\mathcal{F}$ specified in \figref{fig:hppy-func}, so the client
	learns $H$ and then outputs 
$\st = (\key_p, \key_1,\ldots,\key_{n/k}, \pi, B, H)$.

	\item $\query(\st, \qrys ; \db, R)$: on input $\qrys = (\qry_1,\ldots,\qry_k)$, the client 
	parses $(\key_p, \key_1,\ldots,\key_{n/k}, \pi, B, H)\gets\st$.
	Then, iteratively, the client runs the following query protocol for each $i \in [k]$:
	
    \begin{itemize}
    \item The client picks a random unused (regular or refreshed) hint $H[i^*]$ which
    contains $\qry_i$ (in block $j$), sets
	$(o_1,\ldots,o_{n/k})\gets (\prf_2(\key_1,i^*),\ldots, \prf_2(\key_{n/k},i^*))$,
	partition set $P \gets \prf_1(\key_p,i^*) \setminus \{j\}$, and resamples $o_j \getsr [k]$.

	The client sends $(o_1,\ldots,o_{n/k})$, $\pi(i^*)$, and $P$
	(if $B[i^*] = 0$) or its complement $\overline{P} = [n/k] \setminus P$ (if $B[i^*] = 1$).

    \item The server receives offsets $o_1,\ldots, o_{n/k}$, index $\iota$,
    and partition set $P'$ and computes
	$X_0 \gets \{(p-1)\cdot k + o_p : p\in P'\}$ and $X_1 \gets \{(p-1)\cdot k + o_p : p\in [n/k] \setminus P'\}$.

    The server
	returns $s_b = \bigoplus_{x \in X_b} \db[x] \oplus r_{\iota,b}$ for $b = 0,1$.
	
    \item 
	The client computes $\db[\qry_i] = H[i^*] \oplus s_{B[i^*]}$.
	For the next query, the client uses this result to refresh a
	backup hint. 
    In particular, the client will find a backup hint $(d,x) \gets H[\hat{i}]$ such that
	$j\not\in \prf_1(\key_p,\hat{i})$ or $j = d$.  Then, the client stores $H[\hat{i}] \gets (j, d, x \oplus \db[\qry_i])$.
	For this hint, the client will construct the set $P$ as $(\prf_1(\key_p,\hat{i}) \setminus \{d\}) \cup \{j\}$.\footnote{We omit this in the above pseudocode for simplicity, see \cite{EC:HPPY25} for detailed
	pseudocode about managing backup hints.}
	
	At the end, the client outputs $(\db[\qry_1], \ldots, \db[\qry_k])$.
\end{itemize}
\end{itemize}

We write the construction in the form of $k$ rounds to illustrate that our construction can
be used without batching. If instead one is specifically concerned about batching, one
could figure out the offsets and partitions for refreshed backup hints before learning the
results. Then, one can continue the loop before ever querying the server and send all
$k$ queries at once. Then, as long as the database entries are inferred in the correct order,
the client can figure out all $k$ values correctly.

\begin{theorem}[SPIR construction]\label{thm:spir-hppy}
Assuming the existence of one-way functions and semi-honest generic 2PC, there exists a SPIR with
preprocessing scheme for $k$ queries on an $n$-bit database which:
\begin{itemize}[noitemsep]
	\item uses no online public-key operations,
	\item operates in a single round-trip,
	\item $\tO(k)$ state size,
	\item $\tO(n/k)$ per-query upload, $O(1)$ per-query download,
	\item and $\tO(n/k)$ per-query computation.
\end{itemize}
\end{theorem}
\begin{proof}
The construction we analyze is described by the $\setup$ and $\query$ algorithms earlier in this section.
Each of the claims about the scheme's efficiency can be verified by direct inspection of the scheme
and the observation that the keys $\key_1,\ldots,\key_{n/k}$ can stored in $O(1)$ space
and derived from a single master key in the client's memory. For claims about invertible PRF efficiency
and security, see \cite{EC:HPPY25}. Throughout, we additionally assume that
the 2PC protocol used is secure and correct with all but negligible probability.

For the purposes of this proof, we will call the hints which can be selected for a query ``active'' hints.
For this first query, this is exactly the regular hints. When a hint is used, it is no longer considered active,
and the backup hint that is refreshed is considered active. We will use the fact that for any of the $k$
queries, the distribution of active hints is distributed identically with all but negligible probability. In particular,
as long as the query function finds a suitable hint for the query and there is some backup hint that gets
refreshed, the distribution is identical to that of the regular hints. And, both of these events occur with all
but $\negl(\lambda,n)$ probability by our choice of $h = t = k\cdot \polylog(\lambda,n)$ and the
pseudorandomness of the PRFs (see \cite{CCS:RenMugSun24}).

The protocol $\query(\st,\qry; \db, R)$ as described will be correct whenever some active hint
contains the query $\qry$, which as mentioned happens with all but negligible probability.

The privacy of our protocol follows from the similar proofs given in prior works
\cite{CCS:RenMugSun24,EC:HPPY25}, which we extend here. First, each offset sent to the server is
distributed uniformly within the block and independently of the chosen index. This is because we only select
a hint based on whether the specific offset for $\qry$ is in the set. So, all other offsets are pseudorandom
and independent if the PRF is secure, and the offset for $\qry$ is replaced by a fresh random sample.
The partition sent is private, because it is only selected based on if the relevant block is in the partition.
But, then this block is removed and the partition is inverted with probability 1/2. So, the server only sees
a truly random partition of the $n/k$ blocks. Although this inversion is chosen at preprocessing time by the
client, the 2PC security ensures that the server cannot learn about which partition $\qry$ belongs to.
Finally, the indices sent to the server are $k$ random values, drawn without replacement from $[h+t]$,
because the permutation $\pi$ is also kept private from the server and no hint is ever re-used.

More formally, consider an adversary $\A$ gets as input $\txoff$ that is simulated from the server's
inputs only, and assume the pseudorandom functions are replaced with truly random ones. Then, no
matter what the input query set $\qrys$ is, the
uploaded transcript $\txon$ will appear as $k$ messages each with $n/k$ random offsets,
a random partition of $n/k$ blocks, and a uniformly random index drawn from $[h+t]$ without
replacement. Therefore, $\A$ cannot distinguish in this world between any two query sets $\qrys_0$
and $\qrys_1$.

Finally, data privacy is ensured by the security of the 2PC protocol. This is easy to see though.
The download of the $\query$ protocol is only $2$ bits, and one of these bits is padded with
a random value $r_{i,b}$ which is independent of the client's state. Because, in preprocessing,
the client committed through their string $B$ to one of the two $r_{i,0}, r_{i,1}$ values. The other
value is not used to generate any part of the output $H$. Therefore, one of the two bits
the client downloads is padded with a uniformly random and independent bit from the client's
view, which therefore gives the client no information about the database.

A bit more formally, if we give an adversary  $\A$, a state $\st$ and $\txoff$ which is
simulated using only the client's input and the output $H$ of the 2PC protocol, then
the adversary cannot distinguish between databases with $\db[Q]$ fixed. Because,
of their $2k$ bits downloaded in $\txon$, $k$ are independent of their view and
the other $k$ can be simulated from $\db[Q]$ and $H$. And, therefore, the client can never
distinguish between two databases $\db_0$ and $\db_1$ conditioned on their entries at $\qrys$
being the same.
\end{proof}

For simplicity, we give definitions that are more in line with the syntax from our definition of PIR
(\defref{def:pir}). However, our construction has some additional benefits that we do not make
formal in this work. 

\heading{Adaptive security.}
Our construction, although stated as a batch SPIR, is actually secure even if the client is allowed to issues
its queries adaptively. This is because the client can refresh the hint it consumed at the end of each query,
as is done in the adaptive schemes of \cite{CCS:RenMugSun24,EC:HPPY25}.

\heading{Online perfect data privacy and privacy.}
We comment that our scheme, as written, actually achieves perfect
data privacy against adversaries
without access to the offline transcript. This is essentially because the hints stored in $\st$ are all
padded and give no information about the database until queried. Then, on each query, the client is
only able to infer one of the database entries, as the other bit of information it receives is padded
with a bit that is independent of $\st$. Moreover, if we treat the
pseudorandom functions as random oracles, then the queries are
perfectly private
against adversaries without the offline transcript. This shows that our
scheme matches the limits shown in Lemma~\ref{lem:spir}.

\section*{Acknowledgements}
We'd like to thank David Cash for helpful suggestions pointing us toward
prior results in big-key cryptography.
We would also like to thank the
anonymous reviewers for their helpful comments and suggestions, including
one who pointed out a flaw in an earlier version of this work.
We are grateful to Wei-Kai Lin, Shiyu Li, and Mingxun Zhou
for pointing out a proof gap in our previous
version regarding database encodings.

Part of this work was done while the second author was at Google, New York.
The work of the second author at the Universit\`a di Salerno is partially supported by a Google Sponsored Research Agreement.
Part of this work was done while the last
author was at Columbia University that was partially supported by NSF grant CCF-2312242. 

{\small
\bibliographystyle{alpha}
\bibliography{others,abbrev0,crypto}
}

\appendix
\section{Communication Costs for Few Queries}
\label{ap:small-queries}

We note that many of our lower bounds show that PIR with preprocessing
requires $t = \Omega(n/s)$
amortized communication (see Section~\ref{sec:comm-lb}). However, all of
these only apply to constructions in our model (in particular, with
oracle-free encodings) with $s$-bit client storage
that support at least $k = \Omega(s)$ queries per offline phase.
One may wonder, what is the communication complexity
of schemes that potentially support a small and bounded number of
queries per offline phase where our lower bound
no longer applies? It turns out that our lower bound is not applicable for
good reason as there exist constructions with smaller communication in
this case.

In particular, the recent works of
Ghoshal et al.~\cite{EC:GhoZhoShi24,EC:GZSP25}
show that this communication limit is not inherent. They provide PIR which
achieve $\tilde{O}(\sqrt{n/s})$ and $\polylog(\lambda, n/s)$ (respectively) per-query communication
without re-initializing the state. Re-initializing the state costs amortized $O(n/s)$ download, matching our lower bound.
We leave it as an open problem to prove a lower bound on the upload
and download trade-offs for constructions that support only $k \le s$ queries.

\section{Generalization to Resampleable Oracles}
\label{sec:ap-oracles}

A very helpful anonymous reviewer pointed out that our techniques are not limited to crypto oracles. 
Prior works looking at crypto oracles
``compile out'' their oracle at some point in their
proof~\cite{EC:DujHaj24,EC:LinMooWic25}. This
leads to a setting where an honest party can simulate the oracle for the
entire protocol.

In contrast, our results require no such transformation.
We only require that the adversary can sample oracles $\Or'$ and
check for those which are \emph{consistent} with a given set
of $q$ pairs of query inputs and outputs 
$S = \{(\mathsf{in_1},\mathsf{out_1}),\ldots,(\mathsf{in_q},\mathsf{out_q})\}$.
In fact, our proof does not even require that the adversary can do this efficiently, only that it can be done in some
bounded amount of time. From the main body of the paper,
\lemref{lem:oracle_resample} shows that crypto oracles from prior works fall into this class.

More precisely, 
all we need is that the (random) oracle $\Or$ satisfies the following definition of resampleability.

\begin{definition}[Resampleable Oracles]
\label{def:resampleable-oracles}
We say an oracle (random variable) $\Or$ is
$(q,\delta)$-\emph{resampleable} if there exists
a (possibly inefficient) algorithm $\mathcal{A}$ such that for
all sets $S = \{(\mathsf{in}_i,\mathsf{out}_i)\}_{i\in [q]}$, 
we have that 
$\Or' \gets \mathcal{A}(S)$ and
\[
    \Delta(\Or', \Or \mid \forall_{i\in [q]}~
    \Or(\mathsf{in}_i) = \mathsf{out}_i ) < \delta.
\]
In other words, $\mathcal{A}$ can resample an oracle that is $\delta$-close
(in statistical distance) to a fresh sample of $\Or$ conditioned on the event 
that $\Or$ is consistent with the set of queries $S$.
\end{definition}

From this definition, it is clear that the construction in \secref{sec:reduce-dual-pir} and \thref{thm:main-lower-bound} apply to this class
of oracles and not only to crypto oracles, under the same restriction
that the encoding algorithm $\init$ makes no oracle queries
(\defref{def:oracle-free}).
This is just because we can appeal directly to the definition everywhere that we otherwise would've appealed
to \lemref{lem:oracle_resample}, as long as $q = \poly(\secpar)$ and
$\delta \le 2^{-(4n + 2s + 1)}$. Here, the bound on $\delta$ comes from the fact that $\hint$ and $\recon$ resample the oracle at most $2^{3n + 2s + 1}$ times in total and a hybrid argument charges $\delta$ per resampling, so the total statistical distance from the crypto oracle case is at most $2^{-n}$, which only changes the $2^{-n+1}$ term of \thref{thm:dualpirmain} by a constant factor. Note that $\delta$ must be exponentially small in $n$; negligible in $\secpar$ does not suffice.

\begin{remark}
While this is a generalization of our results, %
at the time of writing, the authors are unaware of any interesting oracle
class $\Or$ that is not a crypto oracle but is still useful for cryptography.
Depending on which oracles are shown to be $(q, \delta)$-resampleable, this 
definition may be useful for showing further impossibility results or general
barriers in constructing cryptography.
\end{remark}

\end{document}